\documentclass[aps,pra,superscriptaddress,notitlepage,tightenlines,longbibliography, reprint, 10pt]{revtex4-1}

\usepackage{amsmath}
\usepackage{amssymb}
\usepackage{bbold}
\usepackage{amsthm}
\usepackage{graphicx}
\usepackage{xcolor}
\usepackage{hyperref}
\definecolor{darkblue}{RGB}{0,0,127} 
\definecolor{darkgreen}{RGB}{0,150,0}
\hypersetup{breaklinks, colorlinks, linkcolor=darkblue, citecolor=darkgreen, filecolor=red, urlcolor=darkblue, pdftitle={Multi-qubit Randomized Benchmarking Using Few Samples}, pdfauthor={Jonas Helsen, Joel J. Wallman, Steven T. Flammia, and Stephanie Wehner}}

\usepackage{cleveref}
\usepackage{dsfont}
\usepackage{tikz}

\newcommand{\norm}[1]{\left\lVert#1\right\rVert}

\newcommand{\id}{\mathbb{1}}

\newcommand{\dens}[1]{|#1\rangle\langle#1|}

\newcommand{\ket}[1]{|#1\rangle}

\newcommand{\mc}[1]{\mathcal{#1}}

\newcommand{\md}[1]{\mathsf{#1}}

\newcommand{\inp}[2]{\langle #1,#2\rangle}

\newcommand{\laa}{\langle\!\langle}
\newcommand{\raa}{\rangle\!\rangle}

\newcommand{\tn}[1]{^{\otimes #1}}
\newcommand{\ct}{\ensuremath{^{\dagger}}}
\newcommand{\lv}{|}
\newcommand{\rv}{|}

\DeclareMathOperator{\tr}{Tr}

\newcommand{\cD}{\ensuremath{\mathcal{D}}}

\newcommand{\cG}{\ensuremath{\mathcal{G}}}

\newcommand{\cT}{\ensuremath{\mathcal{T}}}
\newcommand{\cU}{\ensuremath{\mathcal{U}}}

\newcommand{\bbC}{\ensuremath{\mathbb{C}}}

\newcommand{\bbE}{\ensuremath{\mathbb{E}}}

\newcommand{\bbV}{\ensuremath{\mathbb{V}}}

\newcommand{\bbZ}{\ensuremath{\mathbb{Z}}}

\newtheorem{theorem}{Theorem}
\newtheorem{definition}{Definition}
\newtheorem{lemma}{Lemma}
\newtheorem{cor}{Corollary}

\newcounter{notecounter}

\usepackage{appendix}
\usepackage{epstopdf}
\usepackage{todonotes}
\usepackage{enumitem}
\usepackage{algorithm}
\usepackage{algpseudocode}
\newcommand{\Vts}{\ensuremath{V_{\mathrm{TS}}}}
\newcommand{\M}{\mc{M}_d}

\newcommand{\vsp}{\mathrm{Span}}
\newcommand{\bsq}{\bf{\sigma}_q}
\date{\today}
\begin{document}
\title{Multi-qubit Randomized Benchmarking Using Few Samples}

\author{Jonas Helsen}
\affiliation{QuTech, Delft University of Technology, Lorentzweg 1, 2628 CJ Delft, The Netherlands}
\author{Joel J. Wallman}
\affiliation{Institute for Quantum Computing, University of Waterloo, Waterloo, Ontario N2L 3G1, Canada}
\affiliation{Department of Applied Mathematics, University of Waterloo, Waterloo, Ontario N2L 3G1, Canada}
\author{Steven T. Flammia}
\affiliation{Centre for Engineered Quantum Systems, School of Physics, University of Sydney, Sydney, NSW, Australia}
\affiliation{Center for Theoretical Physics, Massachusetts Institute of Technology, Cambridge, USA}
\author{Stephanie Wehner}
\affiliation{QuTech, Delft University of Technology, Lorentzweg 1, 2628 CJ Delft, The Netherlands}

\begin{abstract}
Randomized benchmarking (RB) is an efficient and robust method to
characterize gate errors in quantum circuits. Averaging over random sequences
of gates leads to estimates of gate errors in terms of the average fidelity. These estimates are isolated from the state preparation and measurement errors that
plague other methods like channel tomography and direct fidelity estimation. A
decisive factor in the feasibility of randomized benchmarking is the number of
sampled sequences required to obtain rigorous confidence intervals. Previous bounds were
either prohibitively loose or required the number of sampled sequences to scale
exponentially with the number of qubits in order to obtain a fixed confidence
interval at a fixed error rate.

Here we show that, with a small adaptation to the randomized benchmarking procedure, the number of sampled sequences required for a fixed confidence interval is dramatically smaller than could previously be justified.
In particular, we show that the number of sampled sequences required is essentially independent of the number of qubits and scales favorably with the average error rate of the system under investigation. We also investigate the fitting procedure inherent to randomized benchmarking in the light of our results and find that standard methods such as ordinary least squares optimization can give misleading results. We therefore recommend moving to more sophisticated fitting methods such as iteratively reweighted least squares optimization. Our results bring rigorous randomized benchmarking on systems with many qubits into the realm of experimental feasibility.
\end{abstract}
\maketitle

\section{Introduction}\label{sec:introduction}
One of the central problems in the creation of large-scale, functioning
quantum computers is the need to accurately and efficiently diagnose the
strength and character of the various types of noise affecting quantum operations that arise in experimental
implementations. This noise can be due to many factors, such as imperfect
manufacturing, suboptimal calibration, or uncontrolled coupling to the external
world. Tools that diagnose and quantify these noise sources provide vital feedback on device and control design leading to better quantum devices. They are also used as certification tools, quantifying a device's ability to e.g. perform successful error correction or implement quantum algorithms. A variety of tools have been developed for this purpose, including
state and channel tomography~\cite{Chuang1997,Poyatos1997}, direct fidelity
estimation (DFE)~\cite{Flammia2011,da-Silva2011}, gate set
tomography~\cite{Merkel2013,Blume-Kohout2013}, and randomized benchmarking
(RB)~\cite{Emerson2005,Knill2008,Magesan2011} together with its tomographic extension randomized benchmarking tomography~\cite{Kimmel2014}.
All of these tools have
different strengths and weaknesses. State and channel tomography allow the user
to get a full characterization of the quantum state or channel of interest but are subject to state preparation and measurement errors (SPAM), which place a noise
floor on the accuracy of these characterizations.
Moreover these protocols require
resources that scale exponentially with the number of qubits even for the more efficient 
variants using compressed sensing~\cite{Gross2010a,Flammia2012}, making them
prohibitively expensive for use in multi-qubit systems.
Randomized benchmarking tomography and gate set tomography
remedy the SPAM issue, but require even more resources. 

This exponential scaling with the number of qubits is problematic because even though on most quantum computing platforms multi-qubit gates are generally performed as circuits composed of one and two-qubit gates it is still vitally important to obtain aggregate measures of the behavior of multi-qubit quantum circuits. One can in principle gauge the behavior of these circuit by characterizing their component gates but such a characterization will typically give only loose bounds~\cite{dugas2016efficiently} on the behavior of the full circuit (even disregarding the possibility of correlated errors inside the circuit~\cite{mckay2017three}). Therefore their is a need for diagnostic tools that scale efficiently in the number of qubits.
Protocols designed with such efficiency in mind, like DFE and RB, do not aspire to a full characterization of the
system, but instead aim to estimate a single figure of merit that
ideally captures relevant properties of the system under
investigation. The figure of merit estimated by both DFE and RB is the average
gate fidelity to some target state or gate. However, RB is also robust to SPAM
errors (as opposed to DFE). This makes RB the protocol of choice for characterizing many candidate
quantum computing platforms~\cite{Knill2008,DiCarlo2009,Gaebler2012,Barends2014,Asaad2016}. 
Variants of RB that estimate output purity~\cite{Wallman2015}, and 
leakage~\cite{Wallman2016,Wallman2015b,wood2018quantification} have also been devised.

An important practical problem when using RB is choosing a number of random
gate sequences that is sufficiently small to be practical experimentally, and
yet gives a good estimate of the gate fidelity. This
problem becomes increasingly relevant as error rates improve since estimating
small errors accurately ordinarily requires more samples. Early treatments of this problem
demanded numbers of sequences that were orders of
magnitude larger than were feasible in experiment~\cite{Magesan2012a}. A more
specialized analysis allowed rigorous confidence intervals to be derived for a number of
random sequences comparable to the number used in experiments~\cite{Wallman2014}. However, this analysis only provided reasonable bounds on the number of sequences for short sequence lengths and for single qubit experiments
while more general multi-qubit bounds had an unfavorable exponential scaling with the number of qubits being benchmarked. The restriction to short sequence lengths is also problematic because long sequences generally lead to better
experimental fits~\cite{Epstein2014,Granade2014}.

In this paper we propose an adapted version of the standard RB protocol on the set of Clifford gates that requires little experimental overhead. For this protocol we provide a bound on the number of random sequences required to obtain
rigorous confidence intervals that is several orders of magnitude sharper than
previous multi-qubit bounds. 
Our result makes rigorous and efficient randomized benchmarking 
of multi-qubit systems possible using a reasonable amount of experimental
resources. In particular, our bounds are approximately independent of the number of qubits being benchmarked
As a special case, we also obtain bounds for the single-qubit
version of RB that are valid for all sequence lengths and improve
on the bounds of Ref.~\cite{Wallman2014} for long
sequence lengths. The key to the analysis of the
statistical performance is a novel understanding of the representations of the
Clifford group, developed in a companion paper~\cite{Clifford2016}.
Similar representation-theoretic questions have also been studied independently by Zhu \emph{et al}.~\cite{Zhu2016}.
We also prove a precise sense in which the derived bounds are optimal. Finally we analyze the fitting procedure inherent to randomized benchmarking in light of our results. We conclude that randomized benchmarking yields data that violates the core assumptions of the Ordinary Least Squares fitting procedure, a standard tool for processing randomized benchmarking data \cite{Epstein2014}. This means using OLS to analyze RB data can lead to misleading results. As an alternative we propose using the more sophisticated method of iteratively reweighted least squares optimization, which can be guaranteed to lead to correct results\cite{fedorov2013theory,Seber1989}.

In \cref{sec:summary} we present an overview of the
new contributions of this paper (equations of note here are
\cref{variance_final,variance_final_SPAM}) and explain their
context. In \cref{sec:discussion}, we discuss the implications of the
new bound for experiments, and investigate it in various limits.
Finally, in \cref{sec:methods} we discuss the derivation of the new bounds and
how to apply them in practice, notably with regard to the RB fitting procedure.
We also prove that our results are optimal in some well specified sense. We focus on
intuition and displace most of the technical proofs to the Supplementary Material.
We make heavy use of techniques from group and representation theory, which are of
independent interest, but were derived in a more general setting than needed
for the purpose of this paper. Readers interested in the details of this part
of the derivation are invited to the companion paper~\cite{Clifford2016} or
the closely related work of Zhu \emph{et al}.~\cite{Zhu2016}.

\subsection*{Figure of merit}

We begin by introducing the essential quantities we will use to state and
derive our results. The central problem that RB addresses is how to efficiently
obtain a rigorous figure of merit quantifying how close a physically-performed
operation $\mc{\tilde{U}}$ (represented by a completely positive,
trace preserving (CPTP) map~\cite{Chuang1997}) is to an ideal target operation $\mc{U}$, which is generally
taken to be unitary, that is $\mc{U}(\rho) = U\rho U \ct$ for some unitary $U$ and for all density matrices $\rho$. The quality of a noisy implementation $\mc{\tilde{U}}$ relative to its ideal implementation $\mc{U}$ is
quantified by the average (gate) fidelity,
\begin{equation}\label{eq:average_fidelity}
F_{\mathrm{avg}}(\mc{U}, \mc{\tilde{U}}) := \!\int \mathrm{d}\phi \tr(\mc{U}(\dens{\phi})\mc{\tilde{U}}(\dens{\phi})),
\end{equation}
where $\mathrm{d}\phi$ is the uniform Haar measure over pure quantum states.

It is convenient (and always possible) to write the physically-performed operation $\mc{\tilde{U}}$ as the ideal operation $\mc{U}$ up to composition with a ``noise operation'', that is we write $\mc{\tilde{U}} = \mc{E}\circ\mc{U}$ where $\mc{E}$ is a CPTP map. Note that in general the map $\mc{E}$ can depend on the unitary $\mc{U}$ being implemented. However, in this paper we shall always consider $\mc{E}$ to be the same for all possible unitary operations $\mc{U}$. This is called a gate-independent noise model. We will also work with the more general noise model $\mc{\tilde{U}} = \mc{L}\circ \mc{U}\circ \mc{R}$ where $\mc{R},\mc{L}$ are CPTP maps. This ensures compatibility of our results with recent results on RB with gate-dependent noise~\cite{wallman2018randomized,proctor2017randomized}. However we can always recover the presentation given here by choosing the right gauge. This is explained in \cref{subsec:gate_dependence}.
Because the map $\mc{U}$ is unitary we can also write
\begin{align}\label{eq:average_fidelity_identity}
F_{\mathrm{avg}}(\mc{U}, \mc{\tilde{U}}) =F_{\mathrm{avg}}(\mc{E},\mc{I})
\end{align}
where $\mc{I}$ is the identity operation.
A useful quantity is the average infidelity $r$ defined as
\begin{equation}\label{eq:infidelity}
r(\mc{E}) := 1- F_{\mathrm{avg}}(\mc{E}, \mc{I})
\end{equation}
We also use the quantity $f = f(\mc{E})$ defined as
\begin{equation}\label{eq:depolarizing_parameter}
f(\mc{E}) := \frac{dF_{\mathrm{avg}}(\mc{E},\mc{I})-1}{d-1}
\end{equation}
where $d$ is the dimension of the state space.
One can think of $f$ as the ``depolarizing parameter'' associated to the quantum channel $\mc{E}$. It is this quantity which randomized benchmarking can estimate.
In the text, we will often drop the channel $\mc{E}$ from the (in)fidelity and depolarizing parameter and
simply write $ r(\mc{E})=r$ because the only channel considered in the text is $\mc{E}$ (or equivalently $\mc{R}\mc{L}$, see \cref{subsec:gate_dependence}). 

We will also use another quantity associated to quantum channels called the unitarity
\begin{equation}
u(\mc{E}) := \frac{d}{d\!-\!1}\!\int\!\mathrm{d}\phi \tr\!\left(\bigl|\mc{E}\bigl(\dens{\phi} - \id/d\bigr)\bigr|^2\right).
\end{equation}
The unitarity has the property that $u(\mc{E})=1$ if and only if the quantum channel $\mc{E}$ is unitary~\cite{Wallman2015}. We will again drop the argument and write $u(\mc{E}) = u$. Introducing this extra parameter allows us to differentiate between situations where the noise is coherent or incoherent. Randomized benchmarking behaves fundamentally different in each of these situations, as we explain in \cref{subsec:opt_max_var}.

\subsection*{The randomized benchmarking protocol}
\begin{figure*}[ht!]
\centering
	\framebox{
\includegraphics{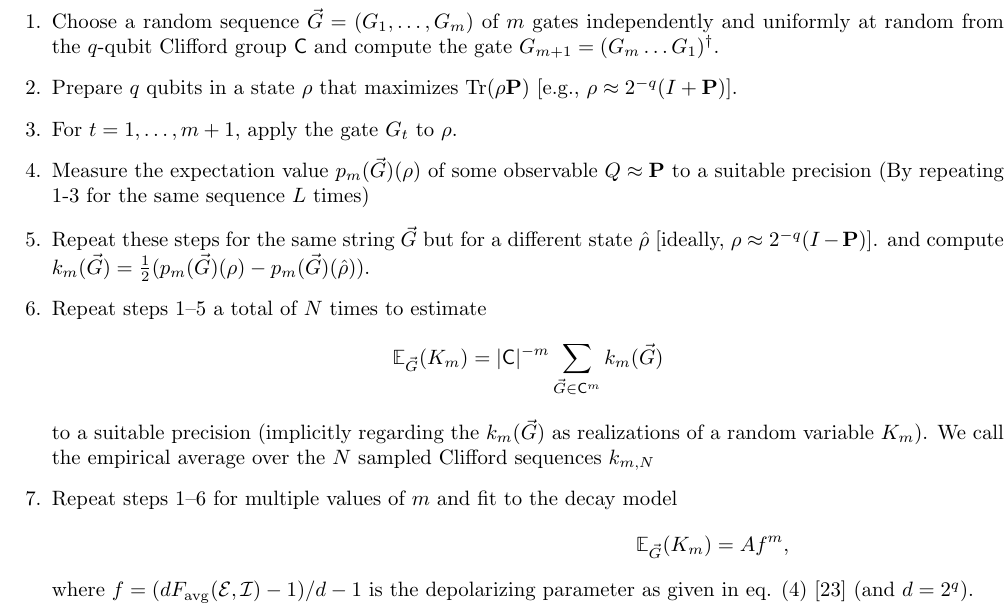}
	}
\caption{
{\bf The Randomized Benchmarking Protocol.}
 We perform randomized benchmarking using the Clifford group $\md{C}$, i.e. all gates that can be constructed by successive application of CNOT gates, Hadamard gates and $\pi/4$ phase gates. We assume the input states $\rho, (\hat{\rho})$ to be noisy implementations of the states $2^{-q}(I + \mathbf{P}), (2^{-q}(I - \mathbf{P}))$, and $Q$ a noisy implementation of the observable $\mathbf{P}$ where $\mathbf{P}$ is a Pauli operator. We denote the length of an RB sequence by $m$, the amount of random sequences for a given $m$ by $N$ and the amount of times a single sequence is repeated by $L$. The goal of this paper is to provide confidence intervals around the empirical average $k_{m,N}$ assuming that individual realizations $k_m(\vec{G})$ are estimated to very high precision (corresponding to the case $L\rightarrow \infty$). In experimental implementations, running the same sequence many times $(L)$ is typically easy, but running many different sequences $(N)$ is hard~\cite{Epstein2014}, meaning that the quantity that we want to minimize is $N$. See \cref{sec:methods} for a detailed discussion of the construction of confidence regions around the empirical average $k_{m,N}$
}\label{box:randomized_benchmarking}
\end{figure*}

In~\cref{box:randomized_benchmarking} we lay out our version of the randomized benchmarking protocol as it was analyzed in~\cite{Knill2008,Magesan2012a,Wallman2014}. We will perform randomized benchmarking over the Clifford group on $q$ qubits $\md{C}$. 
This is the group of unitary operations that can be constructed by considering all possible  products of CNOT gates, Hadamard gates and $\pi/4$ phase gates on the $q$ qubits~\cite{Gottesman1998}. 
We make two essential changes to the standard randomized benchmarking protocol, both of which lead to better guarantees on the precision of randomized benchmarking.
\begin{itemize}
\item A first modification is to perform each randomized benchmarking sequence twice, but with different input states $\rho, \hat{\rho}$ and then subtracting the result. This is equivalent to performing standard randomized benchmarking with the ``input operator'' $\nu = \frac{1}{2}(\rho-\hat{\rho})$. A similar idea was suggested in \cite{Knill2008,Granade2014,Muhonen2015,Fogarty2015}. The factor $(1/2)$ is not strictly necessary but it allows for a fairer comparison between the original benchmarking protocol and our proposal~\footnote{In particular this factor of two ensures that ``signal ranges'' of the two protocols are equal, that is, standard RB starts at $1$ and decays to $1/2$ for large $m$ and the new protocol starts at $1/2$ and decays to $0$.}.
\item Secondly, we do not assume the ideal measurement operator to be the projector on the $|0\cdots0\rangle$ state. Instead we perform some stabilizer measurement related to a pre-chosen Pauli matrix $\mathbf{P}$. An experimentally good choice would be for instance $\mathbf{P} = Z\tn{q}$ but our results hold for any choice of Pauli operator. Correspondingly we pick the input states to be some (impure) states $\rho,\hat{\rho}$ with support on the positive, resp.\ negative, eigenspaces of the Pauli operator $\mathbf{P}$ That is, we would like to prepare the impure states $\rho = \frac{I+ \mathbf{P}}{2d}, \hat{\rho} =  \frac{I- \mathbf{P}}{2d}$.
\end{itemize}
Both of these adjustments are done with the purpose of lowering the experimental requirements for rigorous randomized benchmarking. Our first change to the RB protocol, performing randomized benchmarking with a state difference, has two beneficial effects: (1) It changes the regression problem inherent to randomized benchmarking from an exponential fit with a non-zero off-set to an exponential fit (see \cref{eq:fitting}). This eliminates a fitting parameter, lowering experimental requirements. (2) It lowers the statistical fluctuations of randomized benchmarking regardless of what input states are actually used. This improvement is mostly noticeable in the limit of long sequence lengths. We discuss this in more detail in \cref{subsec:rel_to_reg_bench}.\\
A much stronger improvement to the statistical fluctuations inherent to randomized benchmarking stems from our second change to the RB protocol; preparing states and performing measurements proportional to $\id +\mathbf{P}$ where $\mathbf{P}$ is a Pauli operator. This change allows us to prove a radically sharper bound on the statistical fluctuations induced by finite sampling relative to preparing other input states. This is discussed in \cref{subsec:variance_bound} (see in particular \cref{eq:main_mid_variance}). In \cref{subsec:opt_max_var} we argue that this behavior is not an artifact of our proof techniques but rather inherent to the statistical behavior of randomized benchmarking. Note that for a single qubit the state $(I \pm \mathbf{P})/2$ is in fact a pure state for any choice of $\mathbf{P}$ (in particular $(I + Z)/2 = \dens{0}$). Note that (1) and (2) both reduce the amount of resources needed in a different and independent manner. Using a difference of two input states amounts to effectively “preparing” a traceless input operator. The tracelessness of this operator has two distinct effects. The first effect is that it fixes the constant offset of the decay to be zero, thereby eliminating a fitting parameter. The second effect, which is more subtle, is that it eliminates in the variance expression a representation (which has support on the identity matrix), and hence an extra term in the sequence variance. This means the sequence variance is reduced compared to the sequence variance of standard RB. This effect remains even in the case of imperfect state preparation, as the difference of two density matrices is always traceless (assuming no leakage during the preparation)

As seen in \cref{box:randomized_benchmarking} the RB protocol starts by, for a given sequence of Clifford operations $\vec{G}$ of length $m$, computing the expectation value $p_m(\vec{G})(\rho)$ of an observable $Q$ for two different input states $\rho$ and $\hat{\rho}$. We subtract these two numbers to obtain a number $k_m(\vec{G}) :=\frac{1}{2}(p_m(\vec{G})(\rho) - p_m(\vec{G})(\hat{\rho}))$.
Next we obtain an average of this quantity over all possible sequences $\vec{G}$.
\begin{align}
		\bbE_{\vec{G}}(K_m) = |\md{C}|^{-m} \sum_{\vec{G}\in\md{C}^m} k_m(\vec{G})
\end{align}
This average over all possible Clifford strings of length $m$ can be fitted for various values $m$ to the exponential decay curve
\begin{equation}\label{eq:fitting}
\bbE_{\vec{G}}(K_m) =_{\mathrm{fit}}  Af^m,
\end{equation}
with two fitting parameter $A$ and $f$. In the case where all gates performed in the experiment suffer from the same noise, that is $\mc{\hat{G}} = \mc{E}\circ \mc{G}$ for all Clifford operations $\mc{G}$ the number $f$ can be interpreted as the depolarizing parameter of the channel $\mc{E}$ (as defined in \cref{eq:depolarizing_parameter}) giving an estimate of the average fidelity of the noisy operation $\mc{\hat{G}}$ w.r.t. its ideal version $\mc{G}$.

In practice the number of possible sequences for a given $m$ is too large to
average over completely.  Instead one averages over a randomly sampled subset of sequences,
which generates an  empirical estimate $k_{m,N}$ the validity of which we can interpret using
\emph{confidence regions}. A confidence region, for some set confidence level
$1-\delta$ and size $\epsilon$, is an interval $[k_{m,N}-\epsilon,
k_{m,N}+\epsilon]$ around the estimate $k_{m,N}$ such that the probability
that the (unknown) parameter $\bbE_{\vec{G}}(K_m)$ lies in this interval with
probability greater than $1-\delta$, i.e.,
\begin{equation*}
\text{Prob}\!\left[\bbE_{\vec{G}}(K_m)\!\in\![k_{m,N}
\!-\!\epsilon,k_{m,N}\!+\!\epsilon]\right] \geq 1\!-\!\delta.
\end{equation*}

These confidence intervals, obtained for various values of sequence length during the experiment can then be used in the fitting procedure \cref{eq:fitting} to generate a confidence interval around the empirical estimate $\hat{F}$ for the true channel average fidelity $F_{avg}(\mc{E},\mc{I})$. This can be done using standard statistical procedures (see e.g.~\cite{Beck1978}).
The number of random sequences $N$ used to obtain $k_{m,N}$ will depend on $\epsilon$ and $ \delta$ which are set before the beginning of the experiment, and in general also on some prior estimate of the infidelity $r$ and unitarity $u$. The rest of the paper will be concerned with making this $N$ as small as possible given $\delta$ and $\epsilon$ and (if possible) an a priori bound on the average infidelity $r$.

\section{Results}\label{sec:summary}
In this section we state the main contributions of the paper. We present practical bounds on the number of sequences required to obtain rigorous confidence intervals for randomized benchmarking using the Clifford group under the assumption that the expectation value difference $k_m(\vec{G})$ for a given Clifford sequence $\vec{G}$ is estimated easily to a very high precision. This means we assume that any uncertainty on the number $k_m$ is mostly due to the fact that we only sample $N$ sequences $\vec{G}$.~\cite{Epstein2014,Wallman2014}, or equivalently that the uncertainty on the number $k_m(\vec{G})$ for a fixed sequence $\vec{G}$ is negligible. In order to construct a $1-\delta$ confidence interval of size $\epsilon$ around a randomized benchmarking sequence average $k_{m,N}$ with sequence length $m$, system dimension $d$ and a prior estimate of the channel infidelity $r$ and unitarity $u$ one needs to average over $N$ random sequences where $N$ is given by~\cite{Hoeffding1963}:
\widetext
\begin{equation}
N(\delta,\epsilon, m, r, \chi, d) = -\log(2/\delta)\left[\log
\left(\frac{1}{1-\epsilon}\right)\frac{1-\epsilon}{\bbV^2+1}+
\log\left(\frac{\bbV^2}{\bbV^2+\epsilon}\right)\frac{\bbV^2+\epsilon}{\bbV^2+1}\right]^{-1},
\end{equation}
\endwidetext
\noindent{}where $\bbV^2$ is the variance of the distribution of the samples $k_m(\vec{G})$ from a uniform distribution over the Clifford sequences $\vec{G}$. This variance is given below.

\begin{figure*}[ht]
\centering
\begin{minipage}[t]{0.45\textwidth}
\includegraphics[scale=0.27]{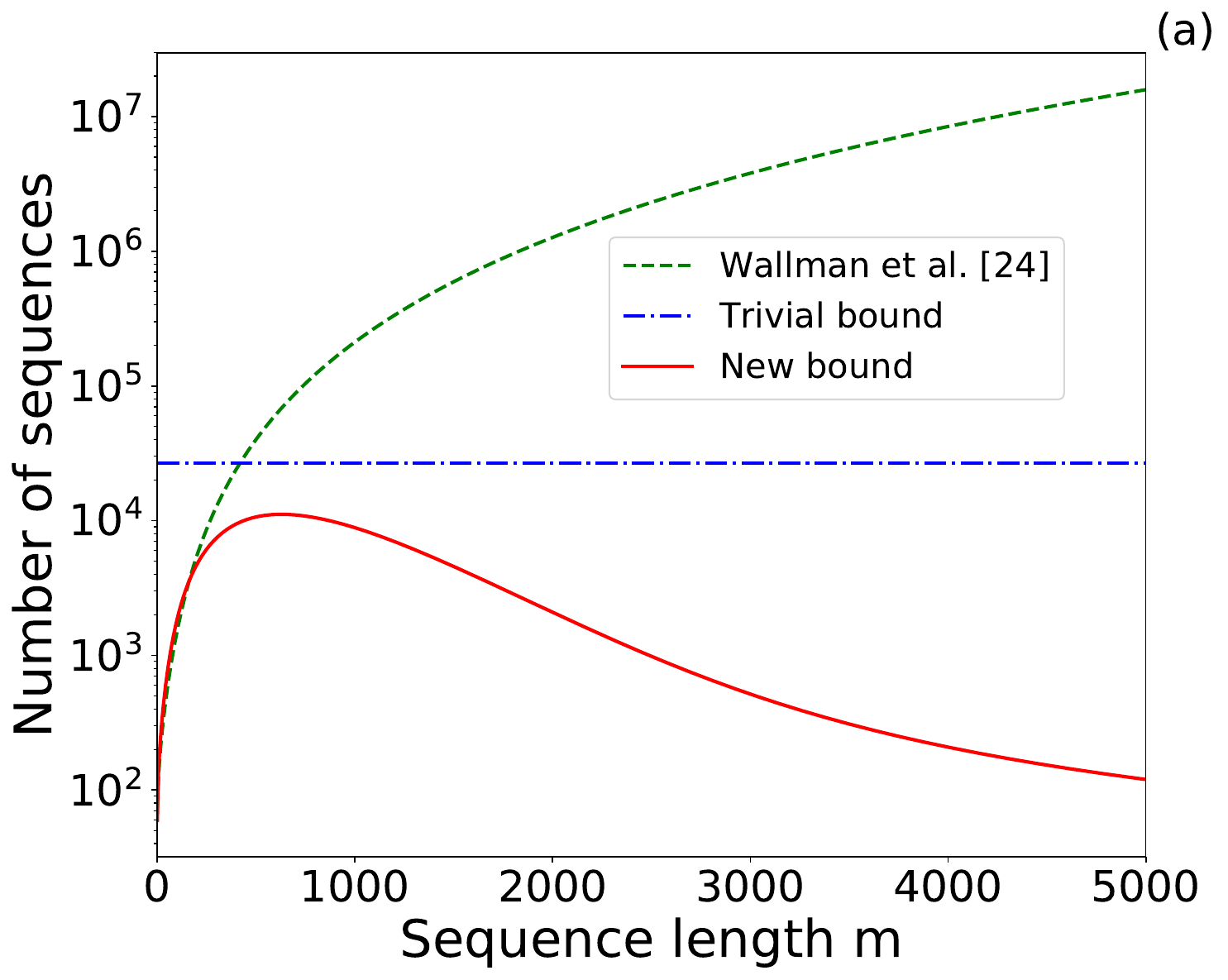}\label{fig:qubit_vs_new}

\end{minipage}
\hspace{5mm}
\begin{minipage}[t]{0.45\textwidth}
\includegraphics[scale=0.27]{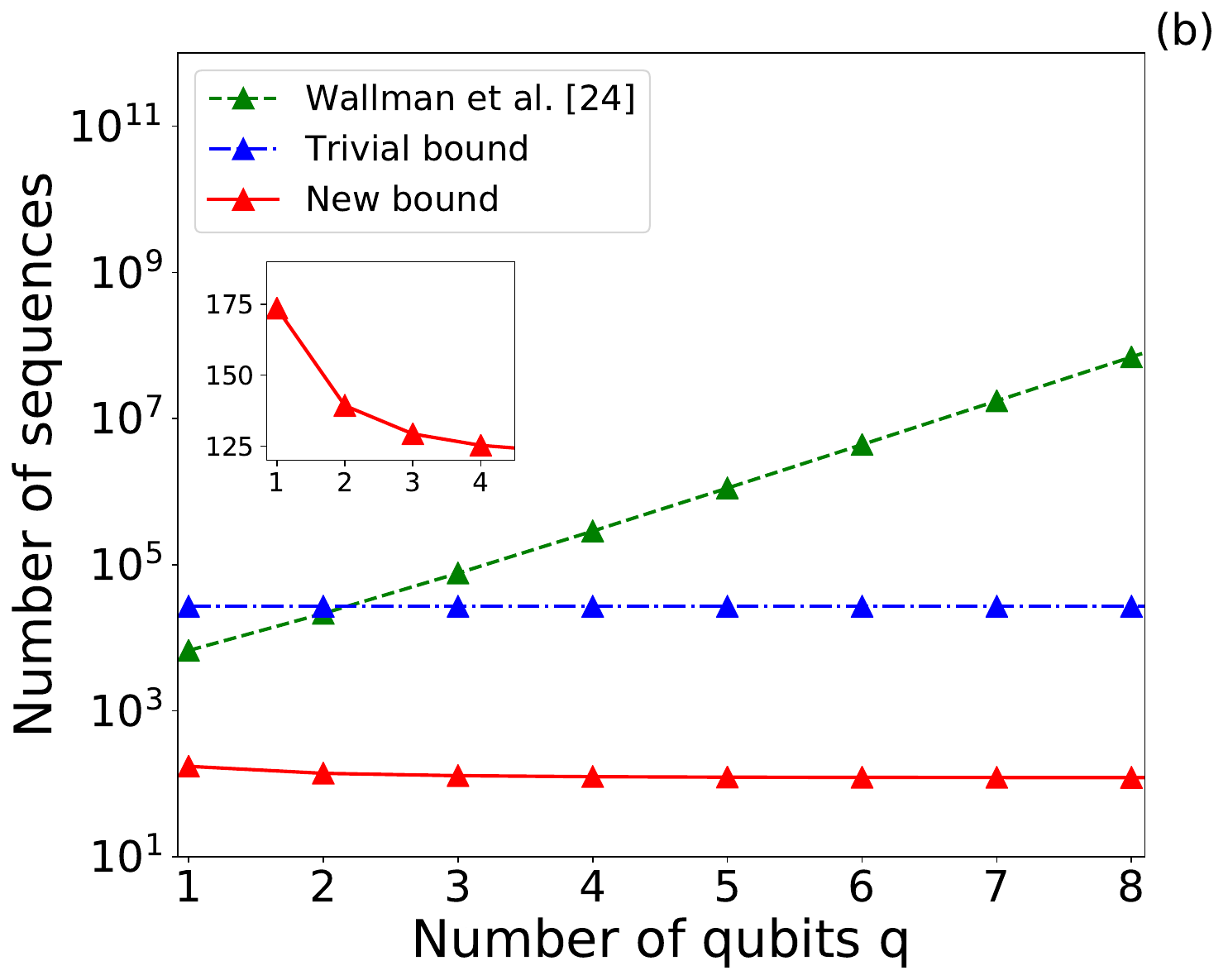}\label{fig:qudit_vs_new}
\end{minipage}
\caption{\textbf{Improvements in dimensional and sequence length scaling}
The number of sequences needed (on a log scale) to obtain a $99\%$ confidence
interval around $p_{m,N}$ with $\epsilon = 10^{-2}$ for a prior infidelity $r=10^{-3}$ as a function of
{\bf (a)} the sequence length $m$ for a single qubit $(q=1)$ from \cref{variance_final} (full line red)
compared to the single-qubit bound from~\cite[Eq. (6)]{Wallman2014} (dashed green) and a trivial bound that arises from noting that the distribution sampled from is bounded on the interval $[0,1]$ and hence has a variance at most $1/4$ (dot-dashed blue) and {\bf (b)} the number of qubits from \cref{variance_final_SPAM} (full line) for sequence length $m=100$ compared to the multi-qubit bound from from~\cite[Eq. (4)]{Wallman2014} (dashed green). In both cases,
our bounds are asymptotically constant while the bounds from~\cite{Wallman2014} diverge. Our bounds are also substantially smaller than the trivial bound.
For multiple qubits, we set the SPAM contribution to $\eta =0.05$ while for a single qubit we set the SPAM contribution to $\eta=0$ in both bounds. We also assumed the unitarity to be $u= (1+f^2)/2$ where $f$ is the depolarizing parameter, corresponding to somewhat, but not fully coherent noise.}
\label{fig:dimensional_scaling}
\end{figure*}

\subsection*{The variance of randomized benchmarking}

The most important contribution of this paper is a bound on the number of
sequences $N$ needed for multi-qubit randomized benchmarking. Previous bounds for
multi-qubit RB~\cite{Wallman2014,Magesan2012a} are either prohibitively loose
or scale exponentially with the number of qubits. Our new bounds, which are derived in detail in \cref{thm:variance bound} of the Supplementary Material, resolve both these issues using techniques from representation theory, enabling multi-qubit RB with practical numbers of random sequences.

\subsubsection*{Variance bound for SPAM-free multi-qubit RB}

For states and measurements that are (very close to) ideal, \cref{subsec:variance_bound} yields a bound on the variance in terms of the sequence length $m$, the infidelity $r$, the unitarity $u$ and the system size $d$. It is given by
\widetext
\begin{equation}\label{variance_final}
\bbV_m^2 \leq \frac{d^2-2}{4(d-1)^2}r^2mf^{m-1}  + \frac{d^2}{(d-1)^2}r^2 u^{m-2}\frac{(m\!-\!1)\!\left(\frac{f^2}{u}\right)^{m}\!-\!m\left(\frac{f^2}{u}\right)^{m-1}+ 1}{(1-f^2/u)^2}.
\end{equation}
\endwidetext
\noindent{}This bound is asymptotically independent of system size $d$.

To illustrate the improvements due to our bound, consider a single qubit $(d=2)$ RB experiment with
sequences of length $m=100$ and average infidelity $r\leq 10^{-4}$.
To obtain a rigorous $99\%$ confidence interval of size $\epsilon = 10^{-2}$ around $p_{m,N}$,
Ref.~\cite{Wallman2014} reported that $N=145$ random sequences were needed (In the case of perfect state preparation and measurement) while our bounds imply that $N= 173$ random sequences are sufficient.
However, the new bound has substantially better scaling with $m$. For instance, with $m=5000$, $\epsilon =0.05$ and
other parameters as above, our bound only requires $N=470$ compared to the
$N=1631$ required by the single qubit bound of Ref.~\cite{Wallman2014}. We illustrate the difference in scaling of the number of sequences needed for a given confidence interval with respect to sequence length $m$ in \cref{fig:dimensional_scaling}.\\

A notable upper bound on \cref{variance_final}, which is easier to work with, is
\begin{align}\label{variance_small_m}
\bbV_m^2 \leq \!f^{m-1} \frac{(d^2\!-\!2)m}{4(d\!-\!1)^2}r^2+ \!u^{m-2}\frac{d^2 m (m\!-\!1)}{2(d-1)^2}r^2.
\end{align}
This bound can be further weakened and simplified by setting $u=1$, yielding an upper bound on the variance that is independent of the unitarity. This bound will however rapidly become trivial with increasing sequence length.

\subsubsection*{Variance bound including SPAM}

The above variance bound is sensitive to SPAM errors, which introduce terms into the variance which scale linearly in the infidelity $r$. In
\cref{thm:variance bound} of the Supplementary Material, we prove that in the
presence of SPAM errors the variance is bounded by
\widetext
\begin{equation}\label{variance_final_SPAM}
\bbV_{\mathrm{SPAM}}^2\, \!\leq\! \frac{d^2-2}{4(d-1)^2}r^2mf^{m-1}\! +\! \frac{d^2(1+4\eta)r^2}{(d-1)^2}\!\frac{(m\!-\!1)\!\left(\!\frac{f^2}{u}\!\right)^{m}\!-\!m\left(\!\frac{f^2}{u}\!\right)^{m-1}\!\!+ \!1}{(1-f^2/u)^2}u^{m-2} \!+ \!\frac{2\eta dm r}{d-1}f^{m-1}.
\end{equation}
\endwidetext
The correction factor $\eta$ only depends on SPAM. As we show in \cref{subsec:asymp_beh_var}, this SPAM dependence is impossible to avoid if one wants to retain the preferred quadratic scaling in infidelity $r$. This bound is also asymptotically independent of the number of
qubits. This means we can perform rigorous randomized benchmarking even in the limit of very many qubits. We illustrate the difference in scaling with respect to
system size in \cref{fig:dimensional_scaling}.

To illustrate the improvements our methods yield we can again compare
to~\cite{Wallman2014}. Consider a system with $4$ qubits, that is, $d=16$, with
sequence length $m=100$, an a priori estimate of $r\leq 10^{-4}$, and $ \eta =0.05$. For a $99\%$ confidence region of size $\epsilon=10^{-2}$ the
previous best known bound for multiple qubits~\cite{Wallman2014} would require
$N=3\times 10^{5}$ random sequences, while our dimension independent
bound from \cref{variance_final_SPAM} only requires $N=249$.

\subsubsection*{Optimality of results}

We also prove (see \cref{sec:methods}) that for \emph{arbitrary} SPAM a bound on the variance which is linear in the infidelity $r$ is in fact optimal. This means the result stated above is in some sense the best possible bound on the variance of a randomized benchmarking sequence. It is important to note that this optimality result also holds when RB is performed using a different set of gates than the Clifford group and also when one considers the standard protocol~\cite{Knill2008,Magesan2011} as opposed to the protocol involving differences of quantum states which we presented in this paper.

Both the SPAM and SPAM-free variance bound also approach a constant independent of the infidelity $r$ in the limit of large sequence length $m$ when the unitarity is one, that is when the noise in the system is purely coherent. In \cref{subsec:asymp_beh_var} we argue that this behavior is not an artifact of the proof techniques used but is in fact a generic feature of a randomized benchmarking procedure with a unitary noise process. 


\subsubsection*{Fitting procedure}
In \cref{subsec:fitting} we discuss the consequences of \cref{variance_final,variance_final_SPAM} on the fitting procedure used to fit the data $\{k_{m,N}\}$ generated by \cref{box:randomized_benchmarking} to the RB fitting relation \cref{eq:fitting}. Our results show that the variance of randomized benchmarking data is strongly heterogeneous with respect to the sequence length $m$. This invalidates the key assumption of homogeneity of variance (homoskedasticity)~\cite{Seber1989} that is necessary for the correct functioning of Ordinary Least Squares (OLS), the standard method used for fitting RB data~\cite{Epstein2014}. Because of this inferences drawn from can give misleading results when applied to RB data. We recommend switching from OLS to the more sophisticated method of Iteratively Reweighted Least Squares, which can deal with non-homoskedastic data.

\section{Discussion}\label{sec:discussion}

\begin{figure*}[t]
\centering
\begin{minipage}[t]{0.45\textwidth}
\includegraphics[scale=0.27]{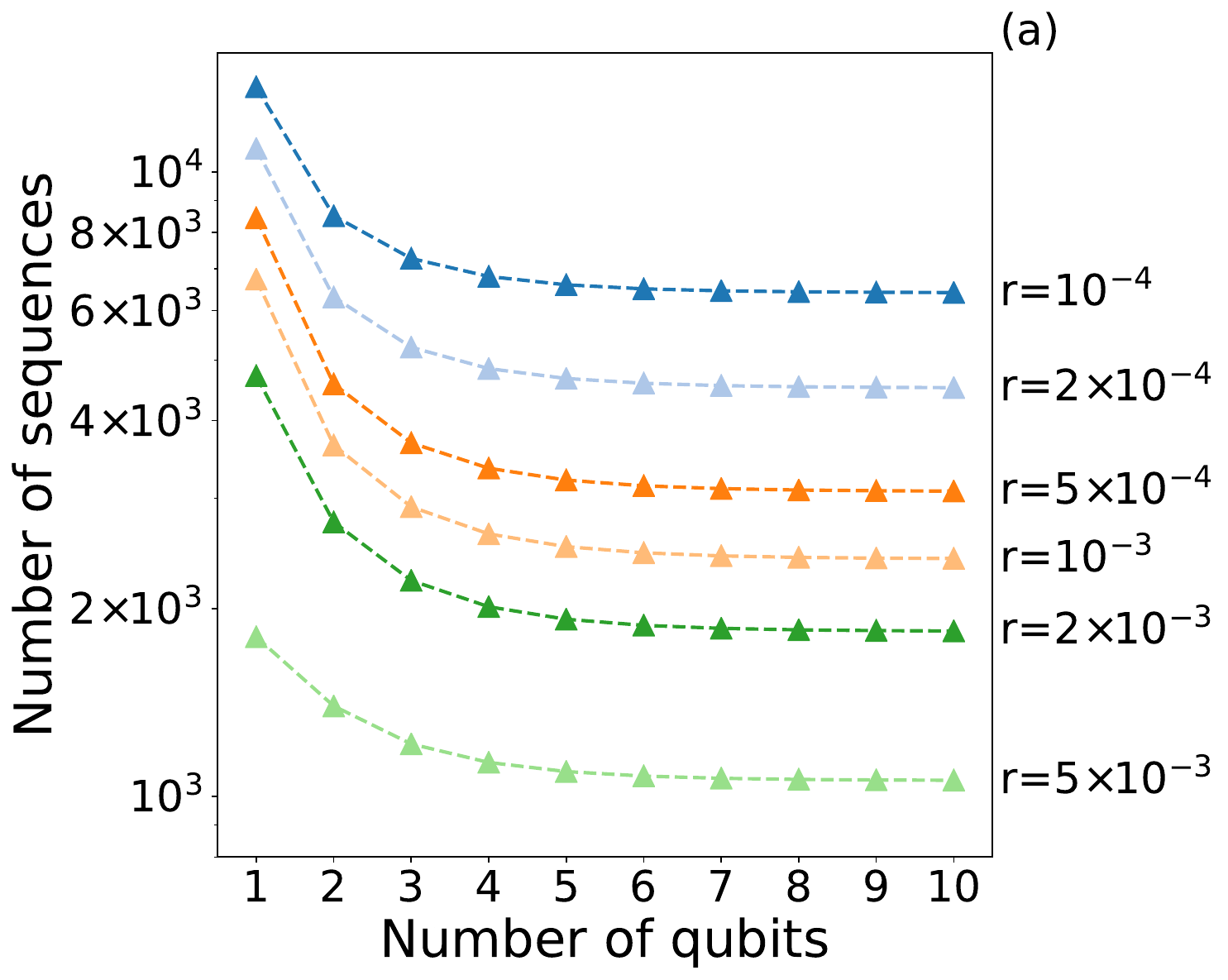}
\end{minipage}
\hspace{5mm}
\begin{minipage}[t]{0.45\textwidth}
\includegraphics[scale=0.27]{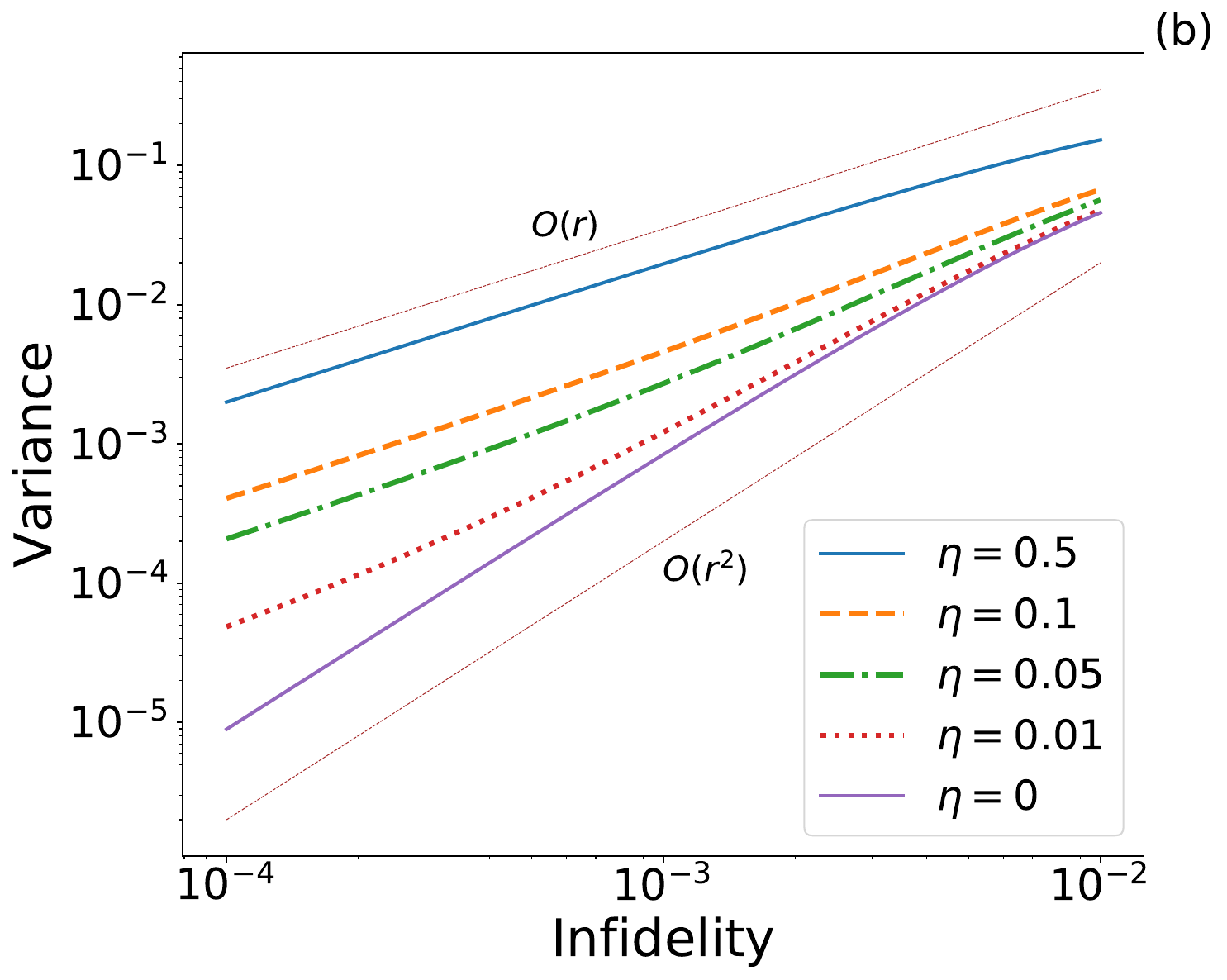}
\end{minipage}
\caption{\textbf{(a)} Number of sequences needed for a $99\%$ confidence interval of size $\epsilon =
5r$ for various infidelities $r$ (ranging from $r=5\cdot10^{-3}$ to$r=10^{-4}$) , number of qubits $q\in[1,10]$ and sequence
length $m=100$ using \cref{variance_small_m} under the assumption of negligible SPAM. (similar plots can be made without this
assumption). The number of sequences needed increases with decreasing
infidelity, reflecting the generic statistical rule that higher precision
requires more samples. 
Note that even in the case of infidelity $r=2\times10^{-4}$ the
number of sequences required is within experimental limits.
\textbf{(b)} Variance, as given by \cref{variance_final_SPAM} versus infidelity $r$ (taking $d=16$ and $m=100$ for illustration) for various levels of SPAM $\eta\in \{0,0.01,0.05,0.1,0.5\}$. Note that the size of the SPAM term has a strong influence on the variance and hence the number of sequences required, especially in the small $r$ limit. As indicated by the visual aids this is due to the transition from a variance scaling quadratically in infidelity $r$ (small $\eta$) to a variance scaling linearly in the infidelity $r$ (large $\eta$). }\label{fig:first_set_of_plots_discussion}
\end{figure*}

In this section we will discuss the behavior of the variance bound \cref{variance_final,variance_final_SPAM} in various regimes. Of interest are its scaling with respect to the number of qubits in the system, the presence of state preparation and measurement noise and varying amounts of coherence in the noise process.

\subsection{Scaling with number of qubits.}

We begin by discussing the effect of the number of qubits in the system on the
variance and the number of necessary sequences.

As illustrated in \cref{fig:dimensional_scaling} (red full) and as can be seen from
\cref{variance_final}, the derived bound is almost independent of the number of qubits $q$ (where
$d=2^q$). In fact, the bound on the variance decreases asymptotically to a
constant in the limit of many qubits despite the number of possible sequences
(that is, $\lvert \md{C}\rvert^m$) increasing exponentially with the number of
qubits. This constitutes a notable improvement over previous multi-qubit variance bounds with an explicit dependence on the infidelity (dashed green in \cref{fig:dimensional_scaling}), given in~\cite{Wallman2014} which had a linear scaling with infidelity but scaled exponentially with the number of qubits. The qualitative behavior of the variance bound in terms of dimension matches a trivial bound on the number of sequences, which can be made by noting that the numbers $k_{m,N}$ are sampled from a distribution bounded on an interval of unit size (and hence has variance at most $1/4$ (dashed blue in \cref{fig:dimensional_scaling})) but is much sharper in absolute terms due to its quadratic dependence on the infidelity $r$.

To further illustrate the behavior of the bound, \cref{fig:first_set_of_plots_discussion}(a) shows the number of sequences needed for a $99\%$ confidence interval around $k_{m,N}$ of size $5r$ versus the number of qubits in the system for various values of $r$ ranging from $5\cdot10^{-3}$ to $10^{-4}$ and sequence length $m=100$. 
The size of $\epsilon$ was chosen to reflect that for fixed sequence length a smaller infidelity will lead to the need for greater precision around $k_{m.N}$ for a successful fit to the exponential \cref{eq:fitting}~\cite{Epstein2014}. 
This plot was made using the unitarity independent bound in \cref{variance_small_m} for ideal SPAM, but similar plots can be made for non-negligible SPAM errors using \cref{variance_final_SPAM}. 
Note also that greater numbers of sequences are needed when the infidelity is small even though the variance in \cref{variance_final} decreases with infidelity. 
This is due to our setting of the size of the confidence interval and reflects the statistical truism that more samples are in general needed to detect small differences.

\subsection{Effects of SPAM terms}

In practice it will always be the case that the input state difference $\nu$ and the output measurement POVM element $Q$ are not ideal. This means that in general we must take into account the contributions from non-ideal SPAM when calculating the number of required sequences. These contributions scale linearly in the infidelity $r$ (see \cref{variance_final_SPAM}) rather than quadratically and so will increase the amount of required sequences. The degree to which $\nu$ and $Q$ deviate from the ideal situation is captured by the prefactor $\eta$ (see \cref{sec:methods} for more on this factor). To illustrate the effect of the SPAM terms on the variance we plot in \cref{fig:first_set_of_plots_discussion}(b) the variance versus the infidelity $r$ using \cref{variance_final_SPAM} taking the sequence length $m=100$ and the dimension of the system $d=16$ (four qubits) for SPAM of size $\eta\in \{0, 0.01,0.05,0.1,0.5\}$. From this plot we note that for non-zero $\eta$ the variance, and hence the amount of sequences needed increases rapidly, especially in the regime of small $r$. This is due to the fact that increasing the SPAM contribution interpolates the variance between a regime where the terms quadratic in infidelity $r$ are dominant and a regime where the terms linear in infidelity $r$ are dominant. This means that, especially when dealing with systems with very small $r$ it is advantageous to try to suppress SPAM errors. In \cref{subsec:opt_max_var} we show that this type of quadratic-to-linear interpolation behavior is in fact optimal for the variance of randomized benchmarking.

\begin{figure*}[t]
\begin{minipage}[t]{0.32\textwidth}
{\small Number of sequence vs sequence length for different levels of coherence}
\includegraphics[width=\linewidth]{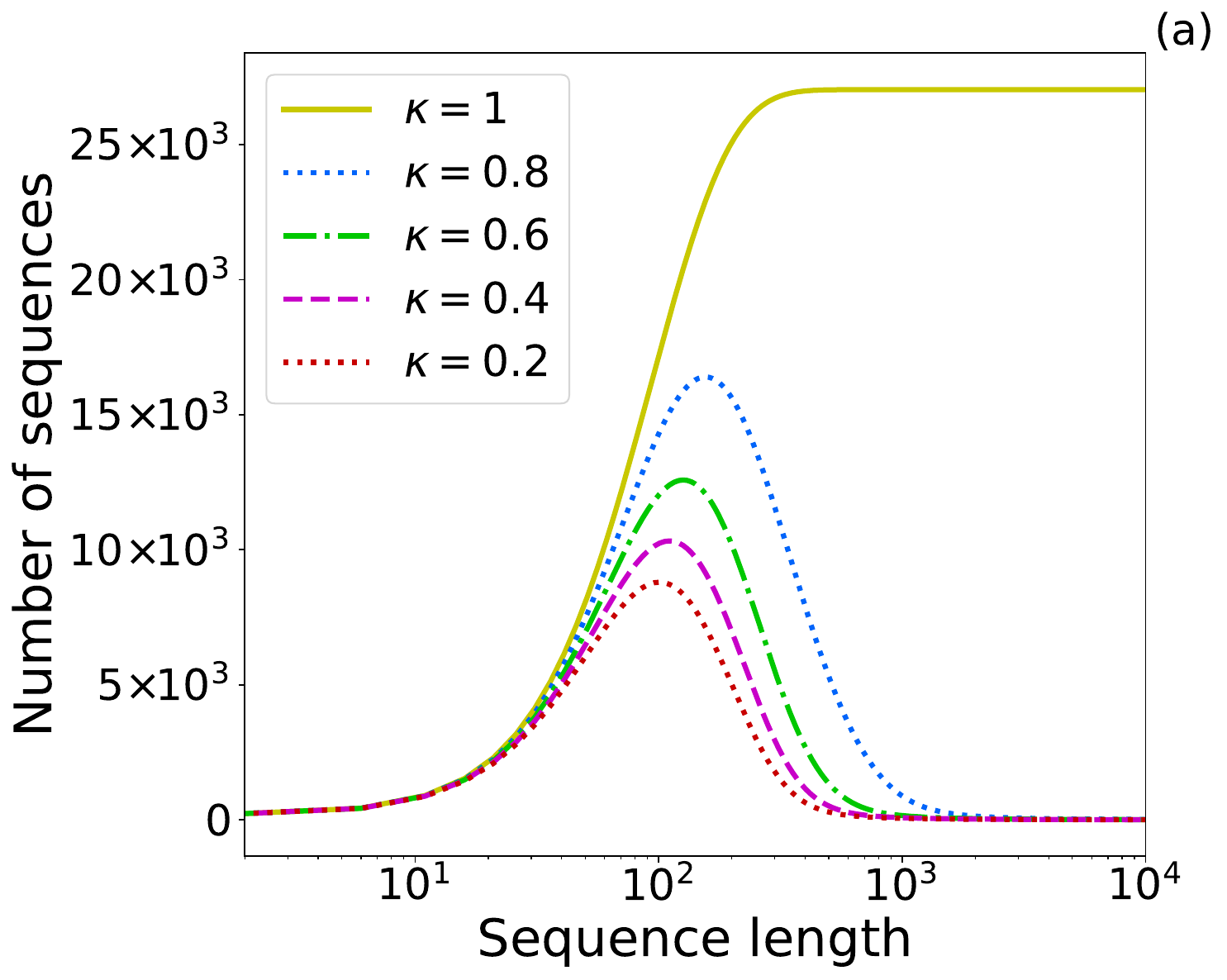}
\end{minipage}
\begin{minipage}[t]{0.33\textwidth}
{\small Contour plot of variance for incoherent noise}
\includegraphics[width=\linewidth]{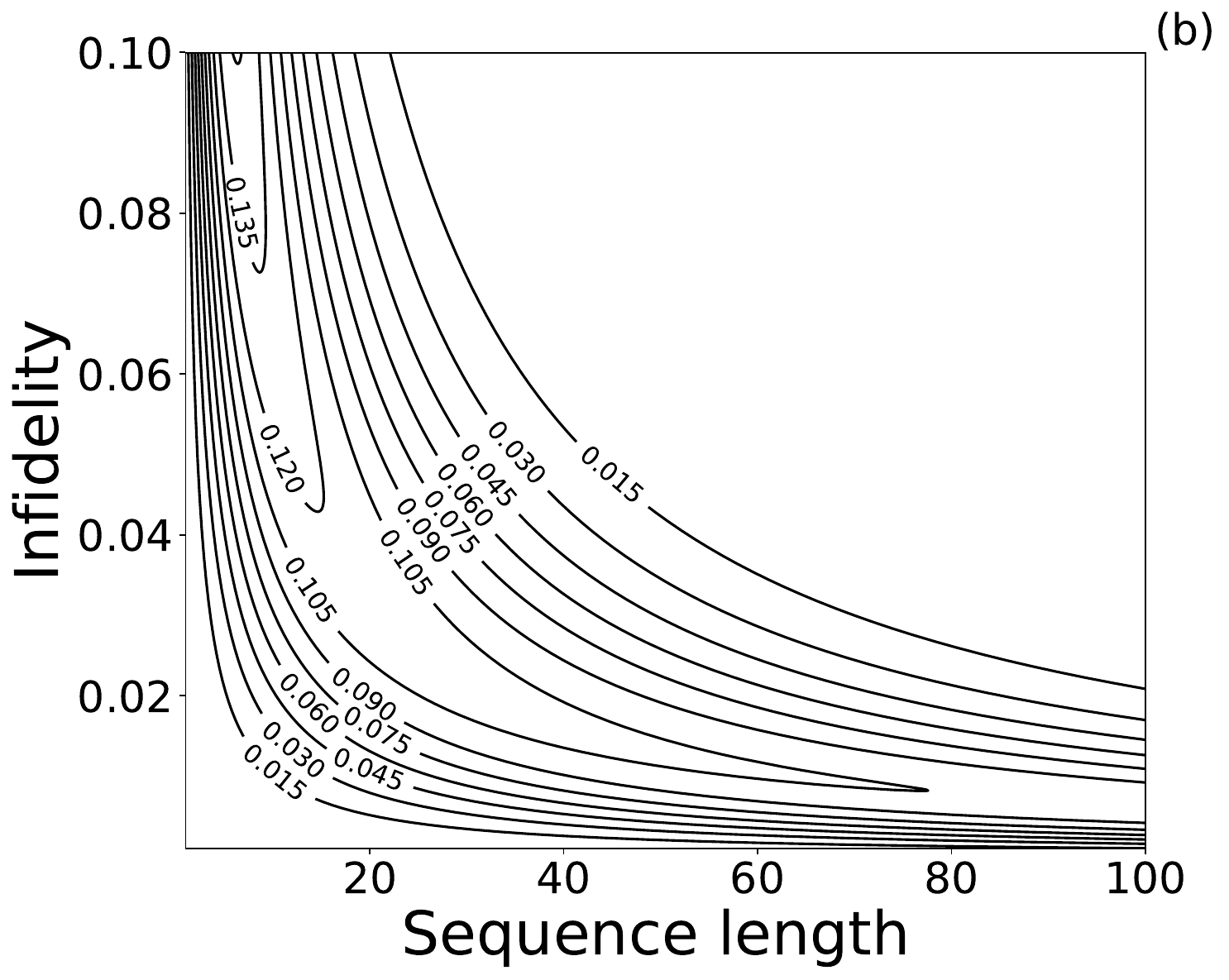}
\end{minipage}
\begin{minipage}[t]{0.33\textwidth}
{\small Contour plot of variance for coherent noise}
\includegraphics[width=\linewidth]{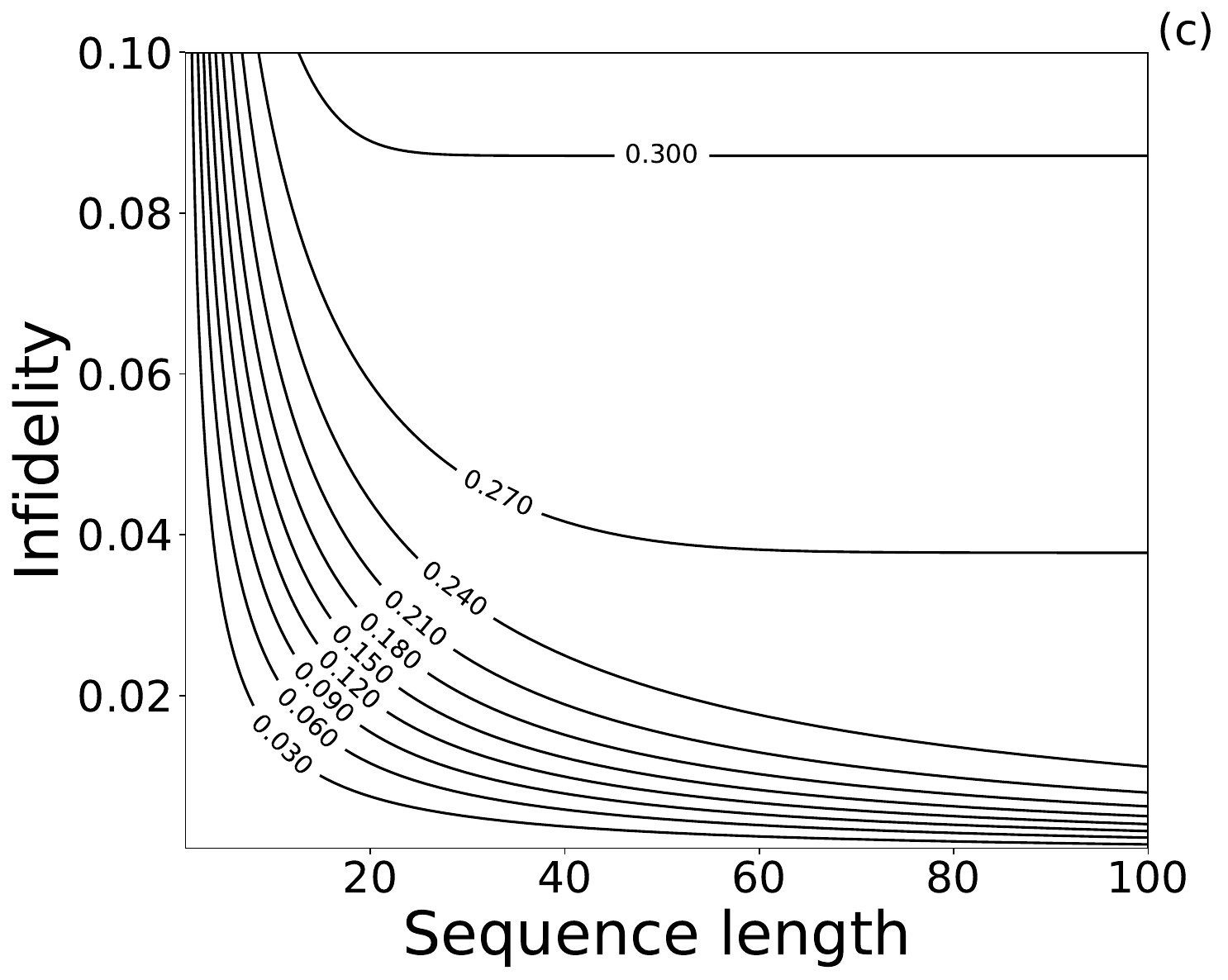}
\end{minipage}
\caption{\textbf{(a)} Number of sequences needed for a $99\%$ confidence interval of size $\epsilon = 0.01$ around $k_{m,N}$ for various values of the unitarity (given by a linear interpolation between $f^2$ and $1$ where $\kappa=1$ corresponds to $u=1$ (unitary noise) and $\kappa=0$ corresponds to $u=f^2$ (depolarizing noise)) for  fixed infidelity $r=0.01$ and sequence length in the interval $m\in [1, 10000]$ (log scale) using the variance \cref{variance_final}. We also assume $d=16$ (four qubits) and ideal SPAM ($\eta$ =0). Note that the number of sequences differs radically for $u=1$ (unitary noise). In the case of $u<1$ the number of sequences needed rises with increasing sequence length $m$, peaks and then decays to zero but for $u=1$ the number of sequences keeps rising with increasing sequence length $m$ until it converges to a non-zero constant (which will be independent of $r$). In \cref{subsec:asymp_beh_var} we argue that this is expected behavior for randomized benchmarking with unitary noise. \textbf{(b),(c)} Contour plot of the variance bound with infidelity on the $y$-axis ($r\in [0.01,0.1]$) and sequence length $m$ on the $y$-axis ($m\in [1,100]$). For \textbf{(b)} we have set the unitarity to $u= (1+f^2)/2$ corresponding to relatively incoherent noise and for \textbf{(c)} we have set the unitarity $u=1$ corresponding to coherent noise. Note again the radical difference in behavior. For $u=1$ the variance rises monotonically in the sequence length $m$ to a constant independent of the infidelity $r$ . Moreover the variance is monotonically increasing in infidelity $r$. However for incoherent noise the variance will peak strongly around $mr\approx 1$ and then decay to zero with increasing sequence length $m$. This means that both an upper and lower bound on the infidelity is required to make full use of the bound in \cref{variance_final}. The looser bound of \cref{variance_small_m} does not share this property and can be used with only an upper bound on the infidelity $r$. }\label{fig:second_pair_of_plots_discussion}
\end{figure*}
\subsection{Scaling with sequence length}

Of more immediate relevance is the scaling of the bound with the sequence
length. It is easy to see that the variance bound \cref{variance_final} scales quadratically in the sequence length $m$ for any noise process when the sequence length is small (see also \cref{variance_small_m}) but when the sequence length is very long the precise nature of the noise under consideration heavily impacts the variance. If the noise is purely coherent, i.e. the unitarity $u=1$, we see that the scaling of the second term in \cref{variance_final} is set by the factor
\begin{equation}
\frac{(m-1)f^{2m}-mf^{2(m-1)}+ 1}{(1-f^2)^2}.
\end{equation}
In the limit of $m$ going to infinity this factor goes to
\begin{equation}
\frac{1}{(1-f^2)^2} \approx O(1/r^2)
\end{equation}
which means the variance \cref{variance_final} converges to a constant independent of the infidelity $r$. This behavior for unitary noise is strikingly different from the behavior for incoherent noise, that is $u<1$. Here we see that the variance in the limit of long sequences is dominated by the exponential terms $u^{m-2}$ and $f^{2(m-1)}$. Since $f$ and $u$ are strictly less than one by the assumption of incoherence, the variance will decay to zero in the limit of long sequences. As $u\geq f^2$ for all possible noise processes~\cite{Wallman2015} the decay rate will be dominated by the size of the unitarity. This is also evident in \cref{fig:second_pair_of_plots_discussion}(a). In this figure we see the number of sequences needed (as given by \cref{variance_final}) versus sequence length $m$ for fixed infidelity $r=0.1$ and dimension $d=16$, and a fixed confidence interval $\delta = 0.99,\epsilon = 0.01$ but for different values of the unitarity $u$. Here we have chosen $u =(\kappa + (1-\kappa)f^2)$ for $\kappa \in \{0.2, 0.4,0.6, 0.8,1\}$ corresponding to the situations where the noise is relatively incoherent going all the way up to a situation where the unitarity is one. We see that for $u<1$ the number of sequences needed first rises quadratically, tops out and subsequently decays to zero whereas in the case of $u=1$ the number of sequences needed keeps rising with sequence length $m$  until it tops out at some asymptotic value. In \cref{subsec:asymp_beh_var} we argue that this behavior is not a feature of the variance bound but rather a feature of the variance of randomized benchmarking itself. Therefore, in the case of highly unitary noise, we recommend performing more experiments at shorter sequence lengths rather than trying to map out the entire decay curve. 

Another noteworthy feature of the variance bound \cref{variance_final} is the fact that, for non-unitary noise (that is $u<1$) it is in general not monotonically increasing in infidelity $r$. Rather, for a fixed sequence length, the variance increases at first with increasing infidelity but then peaks and decays towards zero. This behavior is illustrated in \cref{fig:second_pair_of_plots_discussion}(b). Here we plot a contour plot of the variance with infidelity on the $y$-axis ($r\in [0.01,0.1]$) and sequence length $m$ on the $y$-axis ($m\in [1,100]$) and have set the unitarity to $u=(f^2+1)/2$ corresponding to relatively incoherent noise. The take-away from this plot is that it is not enough to have an upper bound on the infidelity to get an upper bound on the variance, rather one must have both an upper and a lower bound on the variance to make full use of the bound \cref{variance_final}. Note that the looser upper bound \cref{variance_small_m} does not share this behavior and always yields an upper bound on the variance given an upper bound on the infidelity $r$.

On the other hand, when the underlying noise process is unitary, that is $u=1$ the variance does increase monotonically with increasing $r$. This strikingly different behavior is illustrated in \cref{fig:second_pair_of_plots_discussion} (c). Here we plot a contour plot of the variance with infidelity on the $y$-axis ($r\in [0.01,0.1]$) and sequence length $m$ on the $y$-axis ($m\in [1,100]$) and have set the unitarity to $u=1$ corresponding to fully coherent noise.

\subsection{Future work}
 An important caveat when applying the confidence bounds is the assumption of
 gate and time independent noise (this can be relaxed to Markovian, gate
 independent noise~\cite{Wallman2014}). This is an assumption that many analyses of RB suffer from
 to various degrees, hence a major open problem would be to generalize the
 current bounds to encompass more general noise models. Note however, that since our upper bound captures the correct functional behavior of the RB variance with respect to sequence length (for gate and time independent noise) one could in principle check if these assumptions hold true by computing estimates for the variance at each sequence length (from the measured data) and checking if these estimates deviate significantly from the proposed functional form.

 Recent work has also argued that the exponential behavior of randomized benchmarking is robust against Markovian gate-dependent fluctuations~\cite{wallman2018randomized}. This however comes at a substantial increase in mathematical complexity. We suspect that similar robustness statements can be made for the variance of randomized benchmarking but new mathematical tools will be needed (perhaps using the Fourier analysis framework proposed recently in~\cite{merkel2018randomized}) to make this suspicion rigorous. 

Our work can be straightforwardly extended to interleaved RB~\cite{Magesan2012b}. However the dominant source of error in the interleaved RB protocol is usually systematic rather than stochastic (due to the fact that the protocol does not yield an estimate of the interleaved gate fidelity but rather provides upper and lower bounds). Interleaved RB essentially consists of two RB experiments: a reference experiment and an interleaved experiment, the latter of which has an extra ‘interleaved gate’ inserted between the random gates of the standard RB protocol. Hence the fidelity extracted from the second experiment corresponds to the fidelity of the composition of the noise due to the random gates and the noise due to the interleaved gate. An estimate of the fidelity of the interleaved gate is then extracted by considering the ratio of the fidelity of the random gates (from the reference experiment) and the fidelity of the above composition. However, the fidelity of a composition of two noise maps is in general not equal to the product of the fidelities of the individual maps and can, depending on the specifics of the noise processes, differ quite radically. Hence in the absence of more knowledge about the underlying noise processes, IRB gives an inaccurate estimate of the fidelity of the interleaved gate. This inaccuracy is not remedied by reducing the imprecision of the fidelity estimates (for a fixed amount of resources), which is what we provide here. And since the inaccuracy due to this lack of fidelity-composition can be much larger than the imprecision for even a modest amount of resources it is less useful to spend significant energy on increasing precision in IRB.\\

Moreover, it should be noted that while randomized benchmarking is efficient in the complexity theoretical sense, i.e. the amount of resources needed scales polynomially with the number of qubits in the system, the amount of resources required is still significant, and no RB experiment has been performed beyond $3$ qubits so far~\cite{mckay2017three}. Recently several protocols have been devised and implemented that are similar to randomized benchmarking but less resource intensive~\cite{helsen2018new,erhard2019characterizing,xue2019benchmarking}, making larger-scale characterization of multi-qubit systems possible. We suspect the bounds derived in this paper can be adapted to these new proposal but we leave this for future work.

 Also, successful and rigorous randomized benchmarking not only depends on the number of random
 sequences needed per sequence length but also on the fitting procedure used to
 fit the points generated by randomized benchmarking of various lengths to a decay curve in order to extract an estimate of the average gate fidelity. Finding the optimal way to perform this fitting procedure is still an open problem~\cite{Epstein2014}. Accounting for heteroskedasticity, as we have done here, can be considered a first step in this direction. Performing this accounting is standard practice in statistics but does not seem to be in widespread use in the experimental community. One could also consider directly estimating the variance at each sequence length from obtained data and then using these estimates directly as inputs to a weighted least-squares fitting procedure. We however believe that the parametric model we propose here will be more efficient in terms of data needed for a fixed precision.

  Finally, a
 major theoretical open problem is the extension of the present bounds to non-qubit
 systems, different varieties of randomized benchmarking~\cite{Dugas2015,Gambetta2016,barends2014rolling}, and to different $2$-designs~\cite{Dugas2015,Dankert2009,Turner2005} or even orthogonal $2$-designs~\cite{Hashagen2018, Harper2018}. If these $2$-designs are assumed to be
 groups, similar techniques from representation theory might be used~\cite{Gross2007a} but how
 this would be done is currently unknown. 

\section{Methods}\label{sec:methods}

In this section, we will discuss the new contributions in detail, and explain
how to apply them in an experimental setting. 
We will give a high level overview of the proof of the bound on the variance of a randomized benchmarking sequence; full details can be found in the Supplementary Material. 
We will also discuss the behavior of noise terms in the case of non-ideal SPAM and prove that the bounds we obtain are in some sense optimal. 
Finally we briefly comment on how the variance changes when performing regular randomized benchmarking (using an input state $\rho$ rather than an input state difference $\nu = \frac{1}{2}(\rho-\hat{\rho})$).

\subsection{Estimation theory}

In this section, we review confidence intervals and relate the bounding of
confidence intervals to the bounding of the variance of a distribution. A first
thing of note is that all the variance bounds stated in \cref{sec:summary} are
dependent on the infidelity $r$. The appearance of
$r$ in the bound might strike one as odd since this is precisely the quantity
one tries to estimate through RB. It is however a general feature of estimation
theory that one needs some knowledge of the quantity one tries to estimate in
order to use nontrivial estimation methods~\cite{Beck1978}.
Note also that while our results are stated in frequentist language, they should also be
translatable to Bayesian language, that is, as credible regions on the
infidelity given prior beliefs as in Ref.~\cite{Granade2014} for example.
Bayesian methods are more natural because our bounds depend on prior
information about the infidelity, however, a full Bayesian treatment would
involve the fitting process, obscuring our primary technical result, i.e. the variance bounds.\\

Let us now discuss how to use the variance bounds to construct confidence intervals around numbers $k_{m,N}$.
We can in general define a $1-\delta$ confidence interval of size $\epsilon$ to be
\begin{align}
\text{Pr}\left[|k_{m,N} - \bbE_{\vec{G}}(K_m)|\leq \epsilon\right]\geq 1-\delta.
\end{align}
Once we have an upper bound on the variance $\bbV_m^2$ of an RB distribution
we can relate this to an upper bound on number of required sequences through the use of concentration inequalities.

 Note that for the case of randomized benchmarking there are two sets of confidence parameters.
$(\delta_{N},\epsilon_{N})$ is associated with estimating the average over all
possible Clifford sequences, where the relevant parameter is the number of
performed sequences $N$ and $(\delta_{L},\epsilon_{L})$ is associated with
getting an estimate for the survival probability difference $k_{m}(\vec{G})$ for a given fixed sequence.
Here the relevant parameter is $L$, the number of times a single sequence is
performed. Since in practice $L<\infty$ there will be some finite
$(\delta_{L},\epsilon_{L})$ confidence region around the survival probability difference
$k_{m}(\vec{G})$ for a given sequence $\vec{G}$. So in general, when looking at a $\epsilon,
\delta$ confidence region for an RB procedure of a given length on should look
at $(\epsilon_{N}+\epsilon_{L}, \delta_{N}+\delta_L)$ confidence regions. In
what follows we will assume that $L$ is high enough such that
$\epsilon_{L},\delta_{L}$ are negligible relative to
$(\delta_{N},\epsilon_{N})$. This approach is motivated by experimental
realities where it is usually much easier to perform a single string of
Cliffords many times quickly than it is to generate, store and implement a
large number of random sequences.\\

For a given variance $\bbV^2$ we can
relate the number of sequences $N$ needed to obtain $1-\delta$ confidence
intervals of size $\epsilon$ using the following concentration inequality due to Hoeffding \cite{Hoeffding1963}:
\begin{align}\label{eq:concentration}
\text{Pr}\left[|k_{m,N} - \bbE_{\vec{G}}(K_m)|\geq \epsilon\right]&\leq \delta\leq2 H(\bbV^2,\epsilon)^N,
\end{align}
with
\begin{equation}
H(\bbV^2,\epsilon) = \left(\frac{1}{1-\epsilon}\right)^{\frac{1-\epsilon}{\bbV^2+1}} \left(\frac{\bbV^2}{\bbV^2+\epsilon}\right)^\frac{\bbV^2+\epsilon}{\bbV^2+1}.
\end{equation}
We can invert this statement to express the number of necessary sequences $N$ as a function of $\delta, r, \epsilon$ as
\begin{equation}\label{confidence}
N = -\frac{\log(2/\delta)}{\log(H(\bbV^2, \epsilon))}.
\end{equation}
Note that this expression can also be inverted to yield a bound on $\delta, \epsilon$ in terms of a given number of samples $N$. This identity heavily depends on the size of the variance $\bbV_m^2$.

\subsection{State preparation and measurement costs}
We have argued that our adapted RB protocol allows for a reduction in the number of needed sequences to make rigorous estimates. However implicit in this cost reduction argument is the assumption that estimating the number $k_m(\vec{G})$ for a fixed sequence $\vec{G}$ is not more costly than estimating the number $p_m(\vec{G})$. Here we justify this assumption for the two changes we made to the randomized benchmarking protocol: using a state difference as input and using an impure input state defined by a single Pauli matrix. In the following we forgo rigor in favor of intuition. We are however only applying standard statistical techniques that can easily be made rigorous.\\

\noindent {\bf State difference}\\
At first glance one might think that estimating the same sequence twice for difference input states as we propose yields a two-fold overhead in the number of samples per sequence. To see that this is not the case consider the variance $\bbV_\rho^2$ associated with estimating the expectation value for a single sequence for a single state $\rho$. From the standard rules of error addition we now have, for the state difference $\nu= (\rho-\hat{\rho})/2$ that
\begin{equation}
\bbV_{\nu}^2 = \bbV^2_{(\rho-\hat{\rho})/2}= \frac{1}{2^2}(\bbV^2_{\rho} +\bbV^2_{\hat{\rho}} )
\end{equation}
since the random variables associated to $\rho$ and $\hat{\rho}$ are independently distributed (making the covariance zero).
Now assuming that $\rho$ incurs the largest variance, we get
\begin{equation}
\bbV_{\nu}^2 \leq \frac{1}{2}\bbV^2_{\rho}
\end{equation}
which means that estimating the expectation value of a single sequence for a difference of states is statistically not harder than estimating it for a single state.\\

\noindent {\bf Optimal input state and measurement}\\
In our adapted RB procedure we call for preparing the input states $\rho=\frac{\id+\mathbf{P}}{2},\rho=\frac{\id-\mathbf{P}}{2} $ for some Pauli matrix $\mathbf{P}$ and measuring the output operator $\mathbf{P}$. This is different from standard RB where one is asked to prepare and project onto the all zero state $\dens{0\ldots0}$.  We argue that performing RB this way is not more costly than using the standard approach. For concreteness we shall set $\mathbf{P} = Z\tn{q}$. Measuring the expectation value of the operator $Z\tn{q}$ is trivial; one simply measures all qubits in the standard basis (as one would do in standard RB) and then computes the parity of the outcome. Since standard basis states with even parity precisely span the positive eigenspace of $Z\tn{q}$ this amounts to measuring the expectation value of $Z\tn{q}$. Preparing the states  $\rho=\frac{\id+Z\tn{q}}{2},\rho=\frac{\id-Z\tn{q}}{2}$ is a little more involved. The state $\rho$ is a probabilistic mixture of all computational basis states $\ket{x}$ of even parity. By the linearity of expectation one could compute (for a fixed Clifford sequence $\vec{G}$ ) the survival probability $p_m(\vec{G},\ket{x})$ and then compute $p_m(\vec{G},\rho) = 2^{-q/2}\sum_x p_m(\vec{G},\ket{x})$. This requires measuring $2^{2/q}$ expectation values $p_m(\vec{G},\ket{x})$, making this approach not scalable. We can remedy this by realizing that we are only interested in a good estimate of the mean $p_m(\vec{G},\rho)$. Considering $p_m(\vec{G},\ket{x})$ to be the mean of a Bernoulli random variable with outcomes $0$ and $1$, and thus $p_m(\vec{G},\rho)$ to be the mean of a normalized binomial distribution we can estimate this mean efficiently by sampling $\ket{x}$ at random (with even parity), estimating $p_m(\vec{G},\ket{x})$ and then computing the empirical mean. Moreover, since we do not need to know the means $p_m(\vec{G},\ket{x})$ very well to get a good estimate of $p_m(\vec{G},\rho)$ the about of single data points (clicks) gathered to estimate $p_m(\vec{G},\rho)$ is not higher than it would be to accurately estimate $p_m(\vec{G},\ket{\psi})$ for $\ket{\psi}$ some pure state.

\subsection{The fitting procedure}\label{subsec:fitting}

In the previous section we outlined how to use the bound~\cref{variance_final} to construct confidence intervals around $k_{m,N}$. However, we have not yet discussed how to integrate the variance bound~\cref{variance_final} into the fitting procedure required by~\cref{eq:fitting}. A fitting procedure is any method that takes in the set of data points $\{k_{m,N}\}_m$, with $m \in \md{M}$, where $\md{M}$ is some set of integers and outputs a tuple $(A^*,f^*)$ such that $A^* {f^*}^m$ is a `good' description of the data $\{k_{m,N}\}_m$. There are many ways to approach this problem, we refer to~\cite{Seber1989} for a good overview, and finding an optimal procedure is outside the scope of this paper. However we would like to discuss the most commonly used fitting procedure: Ordinary Least Squares (OLS) in the light of the bounds~\cref{variance_final,variance_final_SPAM}.\\

\noindent{\bf Ordinary least squares}\\
Given data $\{k_{m,N}\}_m$ and the function $F(A,f) = Af^m$ the OLS procedure returns estimates $(\hat{A},\hat{f})$. Through a linearization procedure, as outlined for RB in~\cite{Epstein2014}, confidence intervals can then be constructed around these estimates. However, for this procedure to yield correct results each data point $k_{m,N}$ must be distributed around $\mathbb{E}_{\vec{G}}(K_m)$ \emph{with the same variance}~\cite[Chapter 2.8]{Seber1989}. This assumption, called homoskedasticity in the statistics literature, is not universally valid for randomized benchmarking data $\{k_{m,N}\}_m$. This shows in the functional form of the upper bound~\cref{variance_final}, which strongly depends on the sequence length and from \cref{eq:main_mid_variance} one can see that this is not an artifact of bounding techniques but rather an innate feature of RB data. Moreover OLS assumes that the variance of $k_{m,N}$ is independent of the fitting parameters $A,f$, an assumption which is also explicitly violated in RB data. The violation of these two assumptions (homoskedasticity and independence of fitting parameters) creates problems when performing OLS on the RB data $\{k_{m,N}\}_m$. In particular OLS no longer provides an unbiased estimate of the standard error on the fitting parameters $(f,A)$~\cite[Chapter 3.3]{Seber1989}, which can lead to mis-estimation of confidence intervals around the fitting parameters. Therefore we recommend using a more sophisticated approach.\\

\noindent{\bf Iteratively re-weighted least squares}\\
Heteroskedasticity (violation of homoskedasticity) and functional dependence of the data distribution on the fitting parameters are well studied problems, and many robust solutions are available. Here we will focus on one particular solution called Iteratively Re-weighted Least Squares (IRLS). 
For the purposes of this construction we will assume that the data $\{k_{m,N}\}_m$ is drawn from a random variable with mean $\mathbb{E}_{\vec{G}}(K_m)$ and variance $\bbV_m^2(m,r)/N$.
IRLS constructs estimates for the parameters $(A,f)$ by minimizing the function
\begin{equation}\label{eq:optimize}
\min_{A,f} \sum_{m\in \md{M}} w_{m}( k_{m,N} - Af^m)^2
\end{equation}
where the weights $w_m$ can depend on $f$ and $A$. Under the assumption that \cref{variance_final} is the actual variance $\bbV_m^2$ up to a constant factor we can set the weights~\cite[Section 2.8.8]{Seber1989} to be $w_m= w(f,u,m) = 1/\sigma(f,u,m)$ where $\sigma$ is the RHS of \cref{variance_final} (if one suspects that $\eta\neq 0$ the \cref{variance_final_SPAM} can be used instead). We note that this procedure is fairly robust against misspecification of the weights, and moreover that $\sigma$ captures the behavior of $\bbV_m^2$ with respect to the sequence length very well (see \cref{subsec:variance_bound}). IRLS now proceeds in the following manner:

\begin{algorithm}[H]
\caption{Iteratively Reweighted Least Squares}
\label{alg:IRLS}
\begin{algorithmic}[1]

\Require{Initial estimates $f_0,u_0,A_0$ and  a dataset $k_{m,N}$}
\Ensure{Final estimates $\hat{f},\hat{A}$}
\State \textbf{Set} $f_{-1} =0$
\State \textbf{Set} $i=0$

\State // Optimization loop (here $\epsilon$ is some preset sensitivity)\\ 

\While{$|f_{i-1} - f_i|\geq \epsilon$}
\State \textbf{Set} $w_m = w(f_i,u_i,m) = \sigma(f_i,u_i,m)^{-1}$
\State Optimize \cref{eq:optimize} with weights $w_m$ to get $A_{i+1}, f_{i+1}$
\State Estimate $u_{i+1}$ by fitting $\sigma(f_{i+1},u_i,m)/N$ to the empirical variance of $k_{m,N}$
\State \textbf{Set} $i=i+1$
\EndWhile
\State \textbf{Set} $\hat{A}= A_i,\hat{f} = f_i$
\State \Return{$\hat{A},\hat{f}$}
\end{algorithmic}
\end{algorithm}
It as been shown~\cite[Page 45]{fedorov2013theory} (under some mild regularity conditions) that this algorithm converges to estimates $\hat{A},\hat{f}$. If the weights $w_m$ are exactly proportional to the variance $\bbV_m^2$ then these estimates are asymptotically consistent. In \cref{subsec:sample complexity} in the Supplementary Material we provide a detailed estimate of how close the estimate $\hat{f}$ is to the real depolarizing parameter $f$ in terms of the number of data points in $\{k_{m,N}\}_m$ and the number of sequences $N$ sampled per data point.

Finally we would like to note that we have in this procedure kept the number of sequences $N$ constant for varying $N$. It is however possible to let $N$ depend on the sequence length $m$. One choice would be to vary $N$ proportionally to $\bbV_0^2$ (assuming a good estimate of $f$ is available). In this scenario, since $k_{m,N}$ is drawn from a distribution with variance $\bbV_m^2/N$ this would remedy the issue with heteroskedasticity and OLS could be used to provide reliable fitting.

\subsection{Gate dependent noise and gauge invariance}\label{subsec:gate_dependence}

In recent work~\cite{proctor2017randomized,wallman2018randomized} it has been noted that the relation between the parameter estimated by randomized benchmarking and the average fidelity is less than straightforward when the noise channel is allowed to depend on the gate being implemented, that is $\tilde{\mc{G}} = \mc{E}_G\mc{G}$. At the heart of the issue is that the only quantities measurable in the lab, probabilities of the form $\tr(Q\tilde{\mc{G}}(\rho))$ for a state $\rho$ and an observable $Q$ are \emph{gauge invariant}. That is, for any invertible superoperator $\mc{S}$ we have that
\begin{equation}
\tr(Q\tilde{\mc{G}}(\rho)) = \tr( \mc{S}^{-1}(Q)\mc{S}\tilde{\mc{G}}\mc{S}^{-1}(\mc{S}(\rho))).
\end{equation}
This difficulty can be remedied by considering a more general noise model. Instead of choosing $\tilde{\mc{G}} = \mc{E}\mc{G}$ one chooses $\mc{\tilde{G}} = \mc{L}_G\mc{G}\mc{R}_G$ for superoperators $\mc{R}_G,\mc{L}_G$~\cite{wallman2018randomized}. The individual operators $\mc{R}_G,\mc{L}_G$ are not gauge invariant but the combined operator $\mc{R}_G\mc{L}_G$ is. Since in this paper we deal exclusively with gate-independent noise we can choose the gauge such that $\mc{L} = \mc{I}$ and $\mc{R} =\mc{E}$ but our results also hold for the more general choice of gauge with the express caveat that our bounds then work in terms of the infidelity $r$ and unitarity $u$ of the noise in between gates $\mc{RL}$. That is we have $r = r(\mc{RL})$ and $u = u(\mc{RL})$. It is possible to see this explicitly by making the substitution $\mc{E}\rightarrow \mc{RL}$ in all steps of the derivation of the variance bound in \cref{subsec:variance_bound} (and \cref{thm:variance bound} in the Supplementary Material).

\subsection{Variance bound}\label{subsec:variance_bound}

In this section we present a derivation of the multi-qubit variance bound in
\cref{variance_final} under the assumption of ideal input difference operator $\nu = \frac{1}{2}(\rho-\hat{\rho})$ and output POVM element $Q$, i.e.
\begin{gather}\label{eq:def_spamfree}
\nu = \frac{\mathbf{P}}{2d}\\
Q = \frac{1}{2}(\id + \mathbf{P})
\end{gather}
where $\mathbf{P}$ is some pre-specified target Pauli matrix (\cref{box:randomized_benchmarking}). Under these ideal conditions we can guarantee that the variance scales quadratically in the infidelity $r$. We will focus on intuition and relegate most technical work to the Supplementary Material.
For the remainder of the text we will choose a basis for the space of linear operators $\M$. This means we can think of density matrices and POVM elements as column and row vectors which we denote with a Dirac-like notation, i.e. $\nu \rightarrow |\nu\raa$ and $Q\rightarrow \laa Q|$. Quantum channels can then be though of as matrices acting on vectors (which represent density matrices). Moreover, in this picture, composition of channels corresponds to matrix multiplication. When measuring the state $\mc{E}(\rho)$ using a two component POVM $\{Q, \id-Q\}$ for some quantum channel $\mc{E}$ and state $\rho$ and positive operator $Q$ we can write the expectation value $\tr(Q\mc{E}(\rho)$ as a vector inner product
\begin{equation}
\tr(Q\mc{E}(\rho))  = \laa Q|\mc{E}(\rho)\raa = \laa Q|\mc{E}|\rho\raa
\end{equation}
where we abuse notation by referring to the matrix representation of the quantum channel $\mc{E}$ as $\mc{E}$ as well. This is variously called the affine or Liouville representation~\cite{Wolf2012,Wallman2014}.

We assume that every experimental implementation of a Clifford gate $\mc{\tilde{G}}$
can be written as $\mc{\tilde{G}} = \mc{E}\mc{G}$ for some fixed CPTP map $\mc{E}$ where $G$
is the ideal Clifford gate. That is, we assume the noise is Markovian, constant
and independent of the target gate. These assumptions can be
relaxed partially~\cite{Epstein2014, Wallman2014, Wallman2015d,wallman2018randomized}.

The key to randomized benchmarking is that randomly applying elements of the
Clifford group $\md{C}$ and then inverting produces, on average, the depolarizing
channel~\cite{DiVincenzo2001}
\begin{align}\label{def:depolarizing}
	\cD_{f}(\rho) = f\rho + \frac{1-f}{d}\id_d,
\end{align}
that is, we have
\begin{equation}
\sum_{G\in {\md{C}_q}} \mc{G}\ct\mc{E} \mc{G} = {\cD}_{f}
\end{equation}
with the depolarizing parameter $f$ related to the fidelity
by~\cite{Nielsen2002}
\begin{equation}
F_{\mathrm{avg}}(\mc{E},\mc{I}) = \frac{(d-1)f+ 1}{d}\ .
\end{equation}
Therefore applying a sequence of independently-random gates and then inverting
produces $\cD_{f^m}$ on average. Hence the expectation value of any operator
decays as $f^m$ on average.

The value of $k_m(\vec{G})$ for a fixed sequence of Clifford gates $\vec{G}$ (as defined in \cref{box:randomized_benchmarking}), and the variance
over $\vec{G}\in\bbC_q$ are
\begin{align}\label{RB_variance}
k_m(\vec{G}) &= \laa Q|\mc{G}_m\ct \mc{E} \mc{G}_m\cdots \mc{G}_1\ct \mc{E} \mc{G}_1|\nu\raa\\
\bbV_m^2 &= \bbE_{\vec{G}}[k_m(\vec{G})^2] -
[\bbE_{\vec{G}}(k_m(\vec{G}))]^2
\end{align}
respectively.
We can use the identity $a^2 = a\otimes a$ for $a\in\bbC$, the
distributivity and associativity of the tensor product, and the linearity of
quantum channels to write this as~\cite{Magesan2011a,Wallman2014}
\begin{align}\label{eq:variance}
\bbV_m^2 = \laa Q\tn{2}\big|T_\md{C}(\mc{E}\tn{2})^{m} \!- \!\left[T_\md{C}(\mc{E})^{m}\right]\tn{2}\!\big|\nu\tn{2}\raa
\end{align}
where
\begin{align}\label{eq:twirls}
T_{\md{C}}(\mc{E}) &= \frac{1}{|\md{C}_q|}\sum_{G\in \md{C}_q}\mc{G}\ct\mc{E}\mc{G} = \cD_f,\\
T_{\md{C}}(\mc{E}\tn{2}) &= \frac{1}{|\md{C}_q|}\sum_{G\in \md{C}_q}{\mc{G}\ct}\tn{2}\mc{E}\tn{2}\mc{G}\tn{2}.
\end{align}
The superoperator $T_\md{C}(\mc{E})$ is often referred to as the \emph{twirl} of the quantum channel $\mc{E}$.

At this point, our analysis diverges from that of Ref.~\cite{Wallman2014}.
First, note that for our modified scheme, $\nu\tn{2}$ is traceless and symmetric under the interchange of the tensor factors (we will refer to such a matrix as a traceless symmetric matrix) so
\begin{align}
\left[T_\md{C}(\mc{E})^{m}\right]\tn{2}\big|\nu\tn{2}\raa = f^{2m}\big|\nu\tn{2}\raa.
\end{align}
Furthermore, $T_\md{C}(\mc{E}\tn{2})$ preserves the trace and symmetry under interchange of tensor factors.
Therefore we can define $T_{\rm{TS}}(\mc{E}\tn{2})$ to be the restriction of $T_\md{C}(\mc{E}\tn{2})$ to the space of traceless symmetric matrices.
As we prove in \cref{lem:Clifford_irreps} and \cite{Clifford2016}, the representation $\mc{G}\tn{2}$ of the Clifford group restricted to the traceless symmetric subspace decomposes into inequivalent irreducible representations.
Therefore by Schur's lemma (see Supplementary Materials for an explanation of Schur's lemma),
\begin{equation}\label{eq:twirl_two_copy_main}
T_{\rm{TS}}(\mc{E}\tn{2}) = \sum_{i\in \mc{Z}} \chi_i\mc{P}_i
\end{equation}
where $\mc{Z}$ indexes the irreducible subrepresentations of $\mc{G}\tn{2}$ on the space of traceless symmetric matrices, $\mc{P}_i$ are projectors associated to each representation and $\chi_i = \chi_i(\mc{E})\in \mathbb{R}$ are numbers that depend on the quantum channel $\mc{E}$. \footnote{When $\nu$ is not traceless, as is the case in regular randomized benchmarking, we can not restrict  the two-copy twirl to a twirl over the traceless-symmetric subspace. However the derivation below will still hold, up to the addition of extra terms stemming from equivalent irreducible subrepresentations present in \cref{eq:twirl_two_copy_main}. This extra term is discussed in \cref{subsec:rel_to_reg_bench} and also \cite{Wallman2014}}
As the $\mc{P}_i$ are orthogonal projectors that span the space of traceless symmetric matrices, we can write the variance as
\begin{align}\label{eq:variance_projectors_main}
\bbV_m^2 &=\sum_{i\in \mc{Z}}\laa Q\tn{2}|\mc{P}_i|\nu\tn{2}\raa (\chi_i^m - f^{2m}).
\end{align}
Now we use a telescoping series trick~(\cref{lem:telescoping series} and in particular \cref{cor:second_order_series}) on the last factor to write this as
\begin{align}\label{eq:main_mid_variance}
\bbV_m^2 &=\!\sum_{i\in \mc{Z}}\laa Q\tn{2}|\mc{P}_i|\nu\tn{2}\raa\Big[mf^{2(m-1)}(\chi_i-f^2) \\&\hspace{10mm}+\!(\chi_i\!-\!f^2)^2\sum_{j=1}^m (j-1)\chi_i^{m-j}f^{2(j-2)}\Big].
\end{align}
Here we see that getting a sharp bound on the variance will depend on getting sharp bounds on the difference between the $\chi_i$ prefactors and the square of the depolarizing parameter $f^2$. Before we start giving upper bounds to \cref{eq:main_mid_variance}, we would like to note that the behavior of \cref{eq:main_mid_variance} with respect to the sequence length $m$ is very well matched to that of the final upper bounds given in \cref{variance_final,variance_final_SPAM}. This justifies the use of \cref{variance_final,variance_final_SPAM} to set the weights in \cref{alg:IRLS}.\\

\noindent Up to this point the derivation has been valid for any input state difference $\nu$ and output positive operator $Q$. However now we will restrict to the case of ideal $Q$ and $\nu$. For the general case of non-ideal $Q$ and $\nu$ see the Supplementary materials. In the case of ideal $Q$ and $\nu$ we can use \cref{lem:projector_lemma,lem:diagonal_squared_channel} to upper bound
\begin{equation}\label{eq:r_squared_upper_bound}
\sum_{i\in \mc{Z}}\laa Q\tn{2}|\mc{P}_i|\nu\tn{2}\raa(\chi_i-f^2)\leq \frac{1}{4}\frac{d^2-1}{(d-1)^2}r^2
\end{equation}
where $r = 1-  F_{\mathrm{avg}}(\mc{E},I)$ is the infidelity of the quantum channel.
We would like to note here that \cref{lem:projector_lemma} can only be applied if $\nu$  ($Q$) are proportional to $\mathbf{P}$ ($\id + \mathbf{P}$). Moreover, we will see in \cref{subsec:opt_max_var} that without this assumption the variance of RB will scale linearly in infidelity $r$. 
Continuing the calculation, for $r\leq \frac{1}{3}$, we can say that (\cref{lem:chi_f_bound})
\begin{equation}
|\chi_i-f^2|\leq \frac{2dr}{d-1}.
\end{equation}
Hence we can say
\begin{align}
\begin{split}
\bbV_m^2 &\leq m f^{2(m-1)}\frac{d^2-2}{4(d-1)^2}r^2 \\&\hspace{5mm}+ \sum_{i\in \mc{Z}}\frac{4d^2r^2\laa Q\tn{2}|\mc{P}_i|\nu\tn{2}\raa}{(d-1)^2}\\&\hspace{10mm}\times \sum_{j=1}^m (j-1)\chi_i^{m-j}f^{2(j-2)}
\end{split}
\end{align}
for ideal $Q$ and $\nu$.
Now we only need to deal with the $\chi_i$ factors in the sum. To do this we will use the fact that every $\chi_i$ term is upper bounded by the unitarity $u$ of the quantum channel $\mc{E}$. This is derived in \cref{lem:unitarity_upper_bound} in the Supplementary Material. Inserting this we get
\begin{align}
\begin{split}
\bbV_m^2 &\leq m f^{2(m-1)}\frac{d^2-2}{4(d-1)^2}r^2 \\&\hspace{5mm}+ \sum_{i\in \mc{Z}}\frac{4d^2r^2\laa Q\tn{2}|\mc{P}_i|\nu\tn{2}\raa}{(d-1)^2}\\&\hspace{10mm}\times \sum_{j=1}^m (j-1)u^{m-j}f^{2(j-2)}.
\end{split}
\end{align}
Now we factor $u^{m-2}$ out of the sum over $j$ and use the fact that this sum has a closed form.
Using this and \cref{lem:projector_lemma} to bound the projector inner products we obtain a final bound on the variance
\begin{align}\label{eq:var_proof_main_final}
\begin{split}
\bbV_m^2 &\leq m f^{2(m-1)}\frac{d^2-2}{4(d-1)^2}r^2 \\&\hspace{5mm}+ \frac{d^2}{(d-1)^2}r^2u^{m-2}\\&\hspace{10mm}\times \frac{(m\!-\!1)(\frac{f^2}{u})^m \!- \!m(\frac{f^2}{u})^{m-1} \!+\!1}{(1\!-\!(\frac{f^2}{u}))^2},
\end{split}
\end{align}
which is the bound we set out to find. To obtain from this the bound given in \cref{variance_small_m} we note that $u\geq f^2$ and moreover that the fractional term in \cref{eq:var_proof_main_final} is monotonically decreasing in $u$ (for fixed $f^2$) and reaches a limiting value of $m(m-1)/2$ in the limit of $u\rightarrow f^2$ (This can be seen by using l'H$\hat{\text{o}}$pital's rule).

\subsection{State preparation and measurement}

When $Q, \nu $ do not satisfy \cref{eq:def_spamfree}, (which will always happen in practice) the above derivation will not hold exactly and the deviation of $Q,\nu$ from their ideal forms will introduce terms of order $\eta r$ i.e., terms which scale linearly and not quadratically in the infidelity $r$. Deriving an expression of the variance taking into account these these contributions is a little tedious so we will relegate it to the Supplementary Material and instead discuss the form of the prefactor $\eta$. Let $\nu$ be some non-ideal input state difference and let $Q$ be some non-ideal observable. Note from \cref{eq:def_spamfree} that the ideal input state difference $\nu$ and output POVM $Q$ are related to a pre-chosen ``target Pauli matrix'' $\mathbf{P}$. We hence have
\begin{align}
Q_{\mathrm{id}} &= \frac{1}{2}(\id +{\mathbf{P}})\\
\nu_{\mathrm{id}} &= \frac{\mathbf{P}}{2d}
\end{align}
the ideal $Q$ and $\nu$. Suppressing some prefactors (the exact expression can be found in \cref{eq:eta_full_term} in the Supplementary material) we get the following approximate expression for the SPAM factor $\eta$:
\begin{align}
\begin{split}
\eta &\approx \norm{Q - Q_{\mathrm{id}}}_2\norm{\nu - \nu_{\mathrm{id}}}_2\\&\hspace{5mm} + \norm{Q - Q_{\mathrm{id}}}_2^2 + \norm{\nu - \nu_{\mathrm{id}}}_2^2
\end{split}
\end{align}
where $\norm{\cdot}_2$ is the Schatten-$2$ norm~\cite{Wolf2012} and $Q,\nu$ are the non-ideal operators that are actually implemented. There are several important things to notice here:
\begin{itemize}
	\item $\eta$ goes to zero in the limit of ideal $Q,\nu$. This justifies our choice of the ideal $Q$ and $\nu$ as being defined in terms of a single Pauli matrix rather than preparing and measuring in the $\ket{0}$ state as was the case in the original randomized benchmarking proposal~\cite{Magesan2012a}
	\item $\eta$ scales quadratically in the deviation from the ideal of $Q$ and $\nu$. This means that for small deviations $\eta$ is likely to be small.
	\item $\eta$ is non-zero for non-ideal $Q$ even when $\nu$ is ideal and vice versa. This is unfortunate as it means that both state preparation and measurement must be good to ensure small variance. However, as we argue in \cref{subsec:opt_max_var}, this is actually optimal.
\end{itemize}

To get a feel for how the parameter $\eta$ behaves we discuss a particular error model for state preparation and measurement errors, inspired by recent research in superconducting qubits~\cite{Sagastizabal2018}. Here we see that the dominant error source when preparing states in the computational basis is given by decay to the ground state when in the excited ($\ket{1}$) state and residual excitations when preparing the ground ($\ket{0}$) state. The dominant contribution to measurement errors when measuring in the computational basis are here discrimination errors (mistaking $0$ for $1$ and vice versa) as well as errors due to finite sampling. When performing our version of RB, and choosing $\mathbf{P}=Z$, we see that $\nu_{\mathrm{id}} = (\dens{0}-\dens{1})/2$ and hence we want to ideally prepare the states $\dens{0}$ and $\dens{1}$. Following~\cite{Sagastizabal2018} we assume $0.5\%$ residual excitations when preparing the $\dens{0}$ state, $0.8\%$ decay to the ground when preparing $\dens{1}$ and a $1\%$ discrimination error (modeled by a symmetric bit-flip channel) (Here we use the discrimination fidelity given in~\cite{riste2012initialization}). Plugging these numbers into the assumed error models and calculating $\eta$ using \cref{eq:eta_full_term} in the Supplementary Material we see that in this case $\eta = 0.001$. Hence we can say that under realistic scenarios $\eta$ will be quite small. It is possible to make a more fine-grained analysis of the SPAM term $\eta$ as it is defined under \cref{eq:eta}, as opposed to upper bounding it. However this is likely to be rather involved and given that $\eta$ is already small in realistic scenarios we have opted not to pursue this here.

 \subsection{Optimality of maximal variance}\label{subsec:opt_max_var}
 In this section we will argue that the bounds on the variance in the case of non-ideal SPAM are optimal in the sense that it is impossible for the variance to scale better that linearly in the infidelity $r$ for arbitrary noise maps when the input POVM element $Q$ is non-ideal even when the input state difference $\nu$ is ideal. The same reasoning will also hold for non-ideal $\nu$ even when $Q$ is ideal. (More generally the reasoning below will also work when randomized benchmarking is performed using a state rather than a state difference but we will not show this explicitly here).

Consider the variance as in \cref{eq:variance} for a randomized benchmarking experiment with a quantum channel $\mc{E}$ with infidelity $r$ and for simplicity set the sequence length $m=1$ (the argument will work for general $m$). Then we have an expression for the variance
\begin{equation}
\bbV^2 = \laa Q\tn{2}\big|T_\md{C}(\mc{E}\tn{2})- T_\md{C}(\mc{E})\tn{2}\big|\nu\tn{2}\raa
\end{equation}
with the $T_\md{C}(\mc{E}\tn{2}), T_\md{C}(\mc{E})\tn{2}$ defined in \cref{eq:twirls}. Now consider setting $\nu = \nu_{\mathrm{id}}$ and maximizing over the POVM element $Q$. That is consider

\begin{equation*}
\bbV^2 =\max_{0\leq Q\leq \id} \laa Q\tn{2}\big|T_\md{C}(\mc{E}\tn{2})- T_\md{C}(\mc{E})\tn{2}\big|\nu_{\mathrm{id}}\tn{2}\raa.
\end{equation*}
Now note that for any unitary $U$ the operator $\mc{U}(Q) = UQU\ct$ is also a POVM element. This means we can write
\begin{align*}
\bbV^2 &= \max_{0\leq Q\leq \id} \laa Q\tn{2}\big|T_\md{C}(\mc{E}\tn{2})- T_\md{C}(\mc{E})\tn{2}\big|\nu_{\mathrm{id}}\tn{2}\raa\\
&= \max_{0\leq Q\leq \id} \laa (\mc{U}(Q))\tn{2}\big|T_\md{C}(\mc{E}\tn{2})- T_\md{C}(\mc{E})\tn{2}\big|\nu_{\mathrm{id}}\tn{2}\raa\\
&\geq \max_{0\leq Q\leq \id} \laa\int dU (\mc{U}(Q))\tn{2}\big|T_\md{C}(\mc{E}\tn{2})\\&\hspace{42mm}- T_\md{C}(\mc{E})\tn{2}\big|\nu_{\mathrm{id}}\tn{2}\raa,
\end{align*}
where we used the linearity of the inner product and the definition of maximum and the integral is taken over the uniform or Haar measure of the unitary group. Now we use a well known fact from the representation theory of the unitary group which states that the integrated operator $\int dU (\mc{U}(Q))$ is precisely proportional to one of the projectors defined in \cref{eq:twirl_two_copy_main}.~\cite{Wallman2015}. In particular it is proportional to the rank one  projector $P_{\mathrm{tr}} = |\Delta\raa\laa\Delta|$ where $\Delta\in \M$ is some matrix operator (see \cref{lem:Clifford_irreps} in the appendix) and $\mathrm{tr}$ is an element of the set $\mc{Z}$ which indexed the irreducible representations of the Clifford group in \cref{eq:twirl_two_copy_main}. This means we can we can write using \cref{eq:variance_projectors_main}
\begin{align}
\begin{split}
\bbV^2 &\geq\max_{0\leq Q\leq \id}\sum_{i\in \mc{Z}} \alpha(Q)\laa \Delta| \mc{P}_i|\nu\raa (\chi_i-f^2)\\
&=\max_{0\leq Q\leq \id}\alpha(Q)\laa \Delta| P_{\mathrm{tr}}|\nu\raa (\chi_{\mathrm{tr}}-f^2)
\end{split}
\end{align}
where $\alpha(Q)$ is some positive prefactor function of $Q$. From \cref{lem:unitarity_upper_bound} and \cite{Wallman2015} it can be seen that $\chi_{\mathrm{tr}}$ is precisely the unitarity $u$ of the quantum channel $\mc{E}$. If we now consider $\mc{E}$ to be a unitary channel (that is $u=1$), we get (ignoring the prefactors, which can be proven to be strictly positive)
\begin{equation}
\bbV^2 \approx 1-f^2 = \frac{dr}{d-1}\left(2-\frac{dr}{d-1} \right)
\end{equation}
which is linear in infidelity $r$. Hence when the POVM element $Q$ is allowed to vary freely a linear scaling of the variance with the infidelity $r$ can not be avoided even when the input state difference $\nu$ is ideal. One can perform a similar thought experiment maximizing over $\nu$ while setting $Q = Q_{\mathrm{id}}$ and get the same result. Hence the expression for $\eta$ we discussed in the above section is essentially optimal.

\subsection{Asymptotic behavior of the variance}\label{subsec:asymp_beh_var}

When looking at the bound on the variance \cref{variance_final} the difference between unitary and non-unitary noise is striking. When the noise is non-unitary, and thus $u<1$ the upper bound on the variance (and hence the variance itself) decays exponentially to zero in the sequence length $m$ but when the noise process is unitary the variance keeps increasing and eventually saturates on a constant that is independent of the infidelity of the noise process. Here we argue that this is not an artifact of the bounding techniques but rather a fundamental feature of performing randomized benchmarking over unitary noise. Moreover this effect is independent of whether RB is performed using a state difference input $\nu$ or a state input $\rho$ (as in standard RB).\\ 
Consider a unitary noise process $\mc{U} = U\cdot U\ct$ with infidelity $r>0$ (That is $\mc{U}$ is not the identity). Now consider a randomized benchmarking experiment of sequence length $m$. That is, for a random sequence of Clifford unitaries $G_1,\ldots G_m$ we perform the unitary
\begin{equation}
V_m = U(G_m\cdots G_1)\ct UG_m U\cdots UG_1
\end{equation}
Following the reasoning of \cite{Magesan2012a} we can write $V_m$ as 
\begin{equation}
V_m = U{G'}_m\ct UG'_m \cdots {G'}_1\ct U G'_1
\end{equation}
where the unitaries $G'_m, \ldots G'_1$ are sampled uniformly at random from the Clifford group. We can equally well think of the unitary $U\ct V_m$ as being the product of $m$ uniformly random samples from the set
\begin{equation}
\md{G}_U = \{G\ct UG\;\;\|\;\; G\in \md{C}\}.
\end{equation}
Note that this set depends on the unitary $U$. 
In \cite{emerson2005convergence} it was shown that the distribution of the product of $m$ unitaries sampled uniformly at random from a set of unitaries converges to the Haar measure (uniform measure) on the unitary group in the limit of large $m$ as long as this set contains a \emph{universal} set of gates. Note that this convergence phenomenon is independent of the initial set~\footnote{This is similar to how the limiting distribution of a random walk is independent of the initial step-size}. 

Note now that as long as the unitary $U$ is not a Clifford gate the set $\md{G}_U$ will contain a universal gateset~\cite{Chuang1997}. 
This means that the distribution from which $V_m$ is sampled will converge to the Haar measure in limit of long sequence length (the extra $U\ct$ factor gets absorbed into the Haar measure). This will happen independently of the unitary $U$ (as long as $U$ is not Clifford). From this we can conclude that the variance of randomized benchmarking with unitary noise must, in the limit of long sequences, converge to the variance of the randomized benchmarking expectation value over the Haar measure independently of what the original unitary noise process is. Note again that the above argument is independent of whether RB is performed using a state difference input or a state input.

\subsection{Relation to regular randomized benchmarking}\label{subsec:rel_to_reg_bench}
When performing regular randomized benchmarking, that is using an input state $\rho = \frac{1}{2}(\id + \mathbf{P})$ rather than an input state difference $\nu = \frac{\mathbf{P}}{2}$ the upper bounds on the variance given in \cref{variance_final_SPAM,variance_final} still hold provided an extra additive term is added to them. This term will stem from the addition of an extra superoperator (that is not a projector) in the sum in \cref{eq:twirl_two_copy_main} which stem from the appearance of two equivalent trivial subrepresentations of the two-copy representation $\mc{G}\tn{2}$ of the Clifford group. This term is of the form
\begin{align}\label{eq:regular_RB_var}
\begin{split}
T &= \frac{1}{4}\norm{\mc{E}(\id/d) - \id/d}_2^2\frac{1-u^m}{1-u}\\&\leq \frac{(d+1)^2}{2d^2}r^2 \frac{1-u^m}{1-u}
\end{split}
\end{align}
where $\mc{E}$ is the noise process under investigation, with infidelity $r$ and unitarity $u$ and system dimension $d$. Here $\norm{\mc{E}(\id/d) - \id/d}_2^2$ is a measure of how `non-unital' the quantum channel $\mc{E}$, that is how far its output deviates from the identity when the identity is the input. This measure can be upper bounded using~\cite[Theorem 3]{Wallman2015e} and is already implicitly analyzed in~\cite{Wallman2014}. We will not prove the above explicitly but it can be derived straightforwardly by following the derivation in \cref{thm:variance bound} using $\rho$ as input state. Note however that the upper bound on $T$ does not decay to zero exponentially but rather converges to a non-zero constant even for non-unitary channels. This is not a feature of the upper bound itself but rather of the long sequence behavior of standard randomized benchmarking. It was proven in \cite[Theorem 17]{Wallman2014} that the upper bound $T$ is actually saturated for almost all non-unitary channels. Moreover, for physically relevant noise models such as amplitude damping $T$ can be quite substantial. This very different behavior in the limit of long sequence lengths further motivates the use the state difference $\nu$ for rigorous randomized benchmarking.

\begin{acknowledgments}
We would like to thank Le Phuc Thinh, Jeremy Ribeiro and David Elkouss for
enlightening discussions. We would also like to thank Bas Dirkse for pointing out several typos in an earlier version of this paper and diligently checking all proofs in the newest version. We would also like to thank Ramiro Sagastizabal and Adriaan Rol for their suggestions on the dominant SPAM errors in superconducting qubits. JH and SW are funded by STW Netherlands, NWO VIDI and
an ERC Starting Grant. This research was supported by the U.S. Army Research
Office through grant numbers W911NF-14-1-0098 and W911NF-14-1-0103.
STF was supported by the Australian Research Council via EQuS project number CE11001013 and by an Australian Research Council Future Fellowship FT130101744.
\end{acknowledgments}

\bibliography{RBlibrary}

\onecolumngrid

\newpage
\appendix

\section{Preliminaries}
\subsection{Clifford and Pauli groups}
In this section we recall definitions for the Pauli and Clifford groups on $q$ qubits. We begin by defining the Pauli group.
\begin{definition}[Pauli group]
Let $\{v_0,v_1\}$ be an orthonormal basis of $\mathbb{C}^2$ and in this basis define the following linear operators by their action on the basis
\begin{align*}
X v_l = v_{l+1},\
Z v_l =(-1)^lv_l,\
Yv_l = iZXv_l = i(-1)^{l+1}v_{l+1},
\end{align*}
for $l\in \{0,1\}$ and addition over indices is taken modulo $2$. Note that $X,Y,Z \in U(2)$. The $q$-qubit Pauli group $\mc{P}_q$ is now defined as the subgroup of the unitary group $\mathrm{U}(2^q)$ consisting of all $q$-fold tensor products of $q$ elements of $\mc{P}_1:=\langle X,Z,i\id_2 \rangle$.
\end{definition}
Elements $P, P'$ of the Pauli group have the property that they either \emph{commute} or \emph{anti-commute}, that is
\begin{equation}
[P,P'] := PP' - P'P = 0 \;\;\;\;\;\;\;\;\;\;\;\; \text{or}\;\;\;\;\;\;\;\;\;\;\;\;\{P,P'\} := PP' + P'P = 0.
\end{equation}

We also define $\hat{\mc{P}}_q $ as the subset of $\mc{P}_q$ consisting of all $q$-fold tensor products of element of the set $\{\id,X,Y,Z\}$, i.e.$\hat{\mc{P}}_q = \{\id,X,Y,Z\}\tn{q}$. Note that the Hermitian subset $\hat{\mc{P}}_q$ of the
Pauli group forms a basis for the Hilbert space
$\M$. We can turn this into an orthonormal basis under
the Hilbert-Schmidt inner product which is defined as
\begin{equation}
\inp{A}{B} := \tr( A\ct B) ,\;\;\;\;\;\;\;\;\;\;\;\;\;\forall A,B\in \M.
\end{equation}
To see this note that $\tr(P) = 0$ for all $P \in \mc{P}_q/\{\id\}$ and that $\tr(\id)=d$. We introduce the set of \emph{normalized} Hermitian Pauli matrices.
\begin{equation}
\sigma_0 := \frac{\id}{\sqrt{d}},\hspace{3mm} {\bsq} := \left\{\frac{P}{\sqrt{d}}\;\;\|\;\; P\in \hat{\mc{P}}_q \backslash \{\id\}\right\},
\end{equation}
where we have given the normalized identity its own symbol for later convenience. We will denote the elements of the set $\bsq$ by Greek letters ($\sigma, \tau, \nu,...$). We also, for later convenience, introduce the \emph{normalized} matrix product of two normalized Pauli matrices as
\begin{equation}
\sigma\cdot \tau  := \sqrt{d} \sigma\tau \hspace{5mm}\sigma,\tau\in \bsq\cup \sigma_0.
\end{equation}
Note that $\sigma\cdot \tau \in \pm\bsq\cup\sigma_0$ if $[\sigma, \tau]=0$ and $i\sigma\cdot \tau \in \pm\bsq$ if $\{\sigma, \tau\}=0$.
Lastly we define the following parametrized subsets of $\bsq$ .For all $\tau
\in\bsq$ we define
\begin{align}
{\bf{N}_\tau}:= \{\sigma \in {\bsq} \;\;|\;\; \{\sigma,\tau\}=0\},\\
{\bf{C}_\tau}:= \{\sigma \in {\bsq}\backslash\{\tau\} \;\;|\;\; [\sigma,\tau]=0\},
\end{align}
Note that we have $|{\bf{N}}_\tau| = \frac{d^2}{2}$, $|{\bf{C}}_\tau| =\frac{d^2}{2}-2$ and
${\bf{C}}_\tau$ and ${\bf{N}}_\tau$ are disjoint for all $\tau \in \bsq$. We also have for $\sigma, \sigma'\in \bsq$ and $\sigma \neq \sigma'$ that $|{\bf{C}}_\sigma \cap {\bf{C}}_{\sigma'}| = \frac{d^2}{4}-3$. For a proof of this see \cite[Lemma 1]{Clifford2016}.

Next we define the Clifford group. We have
\begin{definition}
 The $q$-qubit Clifford group $\md{C}_q$ is the normalizer (up to complex phases) of
$\mc{P}_q$ in $\mathrm{U}(2^q)$, that is,
\begin{align*}
\md{C}_q := \{U\in\mathrm{U}(2^q)\;\;\|\;\; U\mathcal{P}_q U\ct \subseteq \mathcal{P}_q \}/U(1).
\end{align*}
\end{definition}
The Clifford group is also often introduced as the group generated by the Hadamard (H), $\pi/4$ phase gate and CNOT gates on all qubits. These are equivalent definitions (up to global phases)~\cite{farinholt2014ideal}.

For a more expansive introduction to the Pauli and Clifford groups see e.g.~\cite{farinholt2014ideal} and references therein.

\subsection{Representation theory}
We recall some useful facts about the representations of finite groups. For a more in depth treatment of this topic we refer to~\cite{Fulton2004,Goodman2009}.
Let $\md{G}$ be a finite group and let $V$ be some finite dimensional complex vector space. Let also $\md{GL}(V)$ be the group of invertible linear transformations of $V$. We can define a \emph{representation} $\phi$ of the group $G$ on the space $V$ as a map
\begin{equation}\label{eq:representation}
\phi:\md{G} \rightarrow \md{GL}(V): g\mapsto\phi(g)
\end{equation}
that has the property
\begin{equation}
\phi(g)\phi(h) = \phi(gh),\;\;\;\;\;\; \forall g,h\in \md{G}.
\end{equation}
In general we will assume the operators $\phi(g)$ to be unitary.
If there is a non-trivial subspace $W$ of $V$ such that
\begin{equation}\label{eq:subrep}
\phi(g)W\subset W,\;\;\;\;\;\;\forall g\in \md{G},
\end{equation}
then the representation $\phi$ is called \emph{reducible}. The restriction of $\phi$ to the subspace $W$ is also a representation, which we call a \emph{subrepresentation} of $\phi$. If there are no non-trivial subspaces $W$  such that \cref{eq:subrep} holds the representation $\phi$ is called \emph{irreducible}.
Two representations $\phi,\phi'$ of a group $G$ on spaces $V,V'$ are called \emph{equivalent} if there exists an invertible linear map $T:V\rightarrow V'$ such that
\begin{equation}
T\circ\phi(g) = \phi'(g)\circ T,\;\;\;\;\;\;\;  \forall g\in \md{G}.
\end{equation}
We can also define the twirl $\cT_\phi(A)$ of a linear map $A:V\to V$ with respect to the representation $\phi$ to be
\begin{equation}
\cT_{\phi}(A):= \frac{1}{|G|}\sum_{g\in \md{G}} \phi(g)A\phi(g)\ct.
\end{equation}
The following corollary of Schur's lemma, an essential result from representation theory.~\cite{Fulton2004,Goodman2009}, allows us to evaluate twirls over certain types of representations.
\begin{lemma}\label{lem:Schur}
Let $\md{G}$ be a finite group and let $\phi$ be a representation of $\md{G}$ on a complex vector space $V$ with decomposition
\begin{align}
\phi(g) \simeq \bigoplus_i \phi_i(g),\;\;\;\;\;\;\forall g\in \md{G}
\end{align}
into inequivalent irreducible subrepresentations $\phi_i$.
Then for any linear operator $A$ from $V$ to $V$, the twirl of $A$ over $\md{G}$ takes the form
\begin{align}
\cT_{\phi}(A) = \frac{1}{| G |} \sum_{g\in \md{G}} \phi(g)A\phi(g)\ct = \sum_i \frac{\tr(AP_i)}{\tr(P_i)} P_i .
\end{align}
where $P_i$ is the projector onto the subspace carrying the irreducible subrepresentation $\phi_i$. In the rest of the text we will often denote the prefactor $\tr (AP_i)/\tr(P_i)$ by $\chi_i$.
\end{lemma}

\subsection{Liouville representation of quantum channels}

Quantum channels~\cite{Chuang1997} are completely positive and trace-preserving (CPTP) linear maps $\mc{E}:\M \rightarrow \M$.
We will denote quantum channels by calligraphic font throughout.
The canonical example of a quantum channel is conjugation by a unitary $U$, which we denote by the corresponding calligraphic letter, i.e., $\cU(\rho) = U\rho U\ct$ for all density matrices $\rho$. We will denote the noisy implementation of a channel by an overset tilde, e.g., $\tilde{\mc{G}}$ denotes a noisy implementation some ideal quantum channel channel $\cG$.

It is often useful to think of quantum channels as matrices acting on vectors. This is variously known as the Liouville~\cite{Wallman2014} or affine~\cite{Wolf2012} representation.
This representation corresponds to fixing an orthonormal basis for $\M$ according to the Hilbert-Schmidt or trace-inner product and then expressing elements of $\M$ as vectors in $\bbC^{d^2}$. The Hilbert-Schmidt inner product is again defined as
\begin{equation}
\inp{A}{B} := \tr(AB\ct),\;\;\;\;\;\; \forall A,B \in \M.
\end{equation}
Now let $\{B_j\}_j$ for $j\in\bbZ_{d^2}$ be an orthonormal basis for $\bbC^{d\times d}$ with respect to the Hilbert-Schmidt inner product. We can
construct a map $|.\raa:\M\to \bbC^{d^2}$ by setting $|B_j\raa = e_j$ where $e_j$ is the $j$th canonical basis vector for $\bbC^{d^2}$. Linearly extending the map $|\cdot \raa$ to all elements $M\in \M$ we get
\begin{align}
|M\raa = \sum_j \tr(B_j\ct M)|B_j\raa.
\end{align}
Defining $\laa M| = |M\raa\ct$, we then have
\begin{align}
\laa M|N\raa =\inp{M}{N} = \tr(M\ct N),
\end{align}
so that the Hilbert-Schmidt inner product is equivalent to the standard vector inner product.

We will generally construct the Liouville representation using the basis spanned by the \emph{normalized} (with respect to the Hilbert-Schmidt inner product) Pauli matrices $\{\sigma_0\}\cup \boldsymbol{\sigma}_q$ where
$\sigma_0 := I_d/\sqrt{d}$ with $d = 2^q$ is the normalized identity matrix and
\begin{equation}\label{pauli_norm}
\boldsymbol{\sigma}_q :=\frac{1}{\sqrt{d}}\{I_2,X,Y,Z\}\tn{q}\backslash\{\sigma_0\},
\end{equation}
is the set of normalized Hermitian Pauli matrices excluding the identity.

As any quantum channel $\mc{E}$ is a linear map from $\M$ to $\M$ we have
\begin{align}
|\mc{E}(\rho)\raa = \sum_{\sigma \in \bsq\cup \sigma_0} |\mc{E}(\sigma)\raa\!\laa \sigma|\rho\raa,
\end{align}
so that we can represent $\mc{E}$ by the matrix
\begin{align}
\mc{E} = \sum_{\sigma \in \bsq\cup \sigma_0} |\mc{E}(\sigma)\raa\!\laa \sigma|,
\end{align}
where we abuse notation by using the same symbol to refer to an abstract channel and its matrix representation.
The action of a channel $\mc{E}$ on a density matrix $\rho$ now corresponds to the standard matrix action on the vector $|\rho\raa$, hence for a density matrix $\rho$ and a POVM element $Q$ in $\M$ we have
\begin{align}
\mc{E}|\rho\raa &= |\mc{E}(\rho)\raa,\\
\tr(Q\mc{E}(\rho)) &= \laa Q|\mc{E}|\rho\raa.
\end{align}

The Liouville representation has the nice properties (as can be easily checked) that the composition of quantum channels is equivalent to matrix multiplication of their Liouville matrices and that tensor products of channels correspond to tensor products of the corresponding Liouville matrices, that is, for all channels $\mc{E}_1$ and $\mc{E}_2$ and all $A\in\M
$,
\begin{align}\label{eq:composition}
|\mc{E}_1\circ\mc{E}_2(A)\raa &= \mc{E}_1\mc{E}_2|A\raa \notag\\
|\mc{E}_1\otimes \mc{E}_2(A\tn{2})\raa &= \mc{E}_1\otimes \mc{E}_2|A\tn{2}\raa.
\end{align}
 In the Liouville picture the depolarizing parameter and the unitarity~\cite{Wallman2015} of a quantum channel $\mc{E}$ are
\begin{align}\label{eq:liouville_pars}
	f(\mc{E}) &= \frac{1}{d^2-1}\sum_{\tau\in\bf{\sigma}_q} \laa \sigma\rv \mc{E}\lv \sigma \raa \\
	u(\mc{E}) &= \frac{1}{d^2-1}\sum_{\tau\in\bf{\sigma}_q} \laa \sigma\rv \mc{E}\mc{E}\ct\lv \sigma \raa.
\end{align}
and the Liouville representation of a depolarizing channel with depolarizing parameter $f$ is given by~\cite{Wallman2014}
\begin{equation}
\mc{D}_f = |\sigma_0\raa\laa \sigma_0| + f\sum_{\tau \in \bsq} |\tau \raa\laa\tau |.
\end{equation}

\subsection{Traceless-Symmetric representation}
In the rest of the text we will often work with quantum channels which have a tensor product structure. That is we will often be dealing with channels of the form
\begin{equation}\label{eq:traceless_symmetric_channel}
\mc{W} := \sum_{i}\lambda_i\mc{E}_i\tn{2}
\end{equation}
where $\mc{E}_i$ is a CPTP map for all $i$ and $\lambda_i\in \mathbb{C}$ is some abstract parameter. Note that $\mc{W}$ is now a linear map from $\M\tn{2}$ to $\M\tn{2}$. Maps of these form have a number of useful properties which we will now consider. We begin by defining the \emph{traceless-symmetric} subspace $\Vts$ which is a subspace of $\M\tn{2}$ of the form
\begin{equation}
\Vts:= \vsp\left\{ S_{\sigma,\tau} :=\frac{1}{\sqrt{2}}(|\sigma\tau\raa+ |\tau\sigma\raa)\;\;\|\;\;\sigma, \tau \in \bsq \right\}.
\end{equation}
where we have suppressed the tensor product (that is $\sigma\tau := \sigma\otimes \tau$).
The traceless-symmetric subspace has several desirable properties which we note here. First let $\rho, \hat{\rho}\in \M$ be density matrices and call their difference $\nu := \rho -\hat{\rho}$, then we have that
\begin{equation}
|\nu\tn{2}\raa = |(\rho - \hat{\rho})\tn{2}\raa \in V_\mathrm{TS}
\end{equation}
Moreover, for any quantum channel $\mc{W}$ of the form defined in \cref{eq:traceless_symmetric_channel} we have that
\begin{equation}
\mc{W}|v\raa \in \Vts ,\;\;\;\;\;\;\;\;\; \forall |v\raa \in \Vts,
\end{equation}
or equivalently we have that
\begin{equation}
\mc{P}_{\mathrm{TS}}\mc{W} = \mc{W}\mc{P}_{\mathrm{TS}}
\end{equation}
where $\mc{P}_{\mathrm{TS}}$ is the projector onto the space $\Vts$ (note that $\mc{P}_{\mathrm{TS}}$ is a linear map from $\M\tn{2}$ to $\M\tn{2}$). This observation follows from the fact than $\mc{W}$ is a linear combination of two-fold tensor products of quantum channels (which preserve the trace and map operators that are symmetric under interchange of the two copies of $\M\tn{2}$ to operators that are symmetric under interchange of the two copies of $\M\tn{2}$). 

We will in particular be interested in how a representation of of the Clifford group $\md{C}$ behaves on the traceless symmetric subspace. Define the two-fold tensor product representation of the Clifford group on $\M\tn{2}$ as
\begin{equation}\label{eq:two-fold tensor product}
\phi_2: G\longrightarrow \mathcal{G}\tn{2}
\end{equation}
for all where $\mc{G}$ is the Liouville representation of $G$ for all $G\in \md{C}$. This representation has a natural restriction to the subspace $V_{\mathrm{TS}}$ since $\mc{G}\tn{2}$ is of the form described in \cref{eq:traceless_symmetric_channel}.  We can define the subrepresentation $\phi_{\mathrm{TS}}$ of $\phi_2$ as

\begin{equation}\label{eq:traceless_symmetric_representation}
\phi_\mathrm{TS}: G\longrightarrow \mc{P}_{\mathrm{TS}}\mathcal{G}\tn{2}\mc{P}_{\mathrm{TS}}
\end{equation}
for all $G\in \md{C}$. This representation is in general not irreducible but decomposes further into a collection of irreducible subrepresentations. In \cite{Clifford2016} we derived these irreducible subrepresentations of $\phi_{\mathrm{TS}}$ and studied their properties. In  the following lemma we will quote several results from \cite{Clifford2016} which will be useful for our purposes.
\begin{lemma}\label{lem:Clifford_irreps}
Let $\md{C}$ be the Clifford group and let $\phi_\mathrm{TS}$ be the traceless symmetric representation. This representation is a direct sum of three subrepresentations $\phi_\mathrm{d}$ (diagonal), $\phi_\mathrm{[S]}$ (symmetric commuting) and $\phi_\mathrm{\{S\}}$ (symmetric anti-commuting) acting on the spaces
\begin{align}
V_{\mathrm{d}} &:= \vsp\left\{ |\sigma\sigma \raa \;\;\;\|\;\;\; \sigma \in \bf{\sigma}_q\right\}\tag{diagonal}\\
V_\mathrm{[S]} &:= \vsp\{S_{\nu,\nu\cdot\tau}\;\;\|\;\;\tau\in \boldsymbol{\sigma}_q,\;\;\nu\in \bf{C}_\tau\}\tag{symmetric commuting}\\
V_\mathrm{\{S\}} &:= \vsp\{S_{\nu,i \nu\cdot\tau}\;\;\|\;\;\tau\in \boldsymbol{\sigma}_q,\;\;\nu\in \bf{N}_\tau\}\tag{symmetric anti-commuting}
\end{align}
The diagonal subrepresentation $\phi_{\mathrm{d}}$ decomposes into three subrepresentations denoted by $\phi_{\mathrm{tr}},\phi_1,\phi_2$ with $\phi_\mathrm{tr}$ the trivial representation spanned by
\begin{align}
V_{\mathrm{tr}} &= \left\{\frac{1}{\sqrt{d^2-1}}\sum_{\tau \in \boldsymbol{\sigma}_q}|\tau\tau\raa\right\}. \tag{trivial}
\end{align}
We will index these representations by the set $\mc{Z}_{\mathrm{d}} := \{\mathrm{tr},1,2\}$. 

The symmetric commuting representation $\phi_{[S]}$ decomposes into $3$ irreducible subrepresentations denoted as $\phi_{[\mathrm{adj}]},\phi_{[1]},\phi_{[2]}$. We will index these representations by the set $\mc{Z}_{[S]}: = \{[\mathrm{adj}], [1], [2]\}$. The spaces carrying these representations can be written as a direct sum of subspaces in the following way
\begin{equation}\label{eq:symmetric_commuting_decomposition}
V_i = \bigoplus_{\tau\in \bsq} V^\tau_i
\end{equation}
where $V_i^\tau \subset V^{[\tau]}$ with
\begin{equation}
V^{[\tau]} := \vsp\{ S_{\nu,\nu\cdot\tau} \;\;\|\;\; \nu \in \bf{C}_\tau\}.
\end{equation}

The symmetric anti-commuting representation $\phi_{\{S\}}$ decomposes into $2$ irreducible subrepresentations denoted as $\phi_{\{1\}},\phi_{\{2\}}$. We will index these representations by the set $\mc{Z}_{\{S\}}: = \{\{1\},\{2\}\}$. The spaces carrying these representations can be written as a direct sum of subspaces in the following way
\begin{equation}\label{eq:anti_symmetric_commuting_decomposition}
V_i = \bigoplus_{\tau\in \bsq} V^\tau_i
\end{equation}
where $V_i^\tau \subset V^{\{\tau\}}$ with
\begin{equation}
V^{\{\tau\}} := \vsp\{ S_{\nu,i\nu\cdot\tau} \;\;\|\;\; \nu \in \bf{N}_\tau\}.
\end{equation}

Finally we denote the set indexing all irreducible subrepresentations of $\phi_{\mathrm{TS}}$ as $\mc{Z}  = \mc{Z}_{\mathrm{d}}\cup \mc{Z}_{[S]}\cup\mc{Z}_{\{S\}}$ and we note that all irreducible representations indexed by $\mc{Z}$ are mutually inequivalent.
\end{lemma}
Note that we have only given an explicit basis for the space on which the representation $\phi_{\mathrm{tr}}$ acts. It is possible to write down explicit bases for all relevant vector spaces but we will not need to do see (see however \cite{Clifford2016}).


\section{Randomized benchmarking}

\subsection{Variance bound}
In this section we prove the main theorem of the paper. Concretely we prove the following.
\begin{theorem}\label{thm:variance bound}
Let $Q$ be an observable and $\rho,\hat{\rho}$ density matrices and set $\nu=\frac{1}{2}(\rho-\hat{\rho})$. Consider a randomized benchmarking experiment using the Clifford group $\md{C}$ with noisy implementation $\mc{\tilde{G}} = \mc{E}\mc{G}$ for all $G\in \md{C}$. Then the variance $\bbV^2_m$ of this experiment is upper bounded by
\begin{align}
\begin{split}
\bbV^2_m &\leq mf^{m-1} \frac{d^2-2}{(d+1)^2}r^2 + \frac{d^2}{(d-1)^2}r^2 u^{m-2} \frac{(m-1)\left(f^2/u\right)^{m} - m\left(f^2/u\right)^{m-1} + 1}{(1-\left(f^2/u\right))^2}\\
	&\hspace{10mm}+ \eta(Q,\nu)mf^{m-1}r + \eta(Q,\nu) r^2 u^{m-2} \frac{(m-1)\left(f^2/u\right)^{m} - m\left(f^2/u\right)^{m-1} + 1}{(1-\left(f^2/u\right))^2}
\end{split}
\end{align}
where $u=u(\mc{E})$ is the unitarity, $r=r(\mc{E})$ is the infidelity, $d$ is the system dimension, $m$ is the sequence length, $f = 1-\frac{dr}{d-1}$ is the depolarizing parameter and $\eta$ is a function capturing the deviation from the ideal $Q$ and $\nu$. This bound is valid for $r\leq \frac{1}{3}$.
\end{theorem}

\begin{proof}
We begin from an exact expression of the variance expressed in the Liouville representation \cref{eq:variance}:
\begin{equation}
\bbV^2_m = \laa Q\tn{2}|\mc{T}_{\phi_{2}}(\mc{E}\tn{2})^m|\nu\tn{2}\raa  - \laa Q\tn{2}|\big(\mc{T}_{\phi}(\mc{E})\tn{2}\big)^m|\nu\tn{2}\raa
\end{equation}
where $\mc{T}_{\phi_{2}}$ is the twirl over the two-copy representation of the Clifford group as defined in \cref{eq:two-fold tensor product} and $\mc{T}_{\phi}$ is the twirl over the (single copy) Liouville representation. Note now that  $|\nu\tn{2}\raa\in \Vts$ and that both $\mc{T}_{\phi_{2}}(\mc{E}\tn{2})$ and $\mc{T}_{\phi}(\mc{E})\tn{2}$ are CPTP maps of the form described in \cref{eq:traceless_symmetric_channel}. This means we can restrict both twirls to the traceless symmetric subspace. In this subspace we have from \cref{lem:Schur} and \cref{lem:Clifford_irreps} that $\mc{T}_{\phi_{2}}(\mc{E}\tn{2})$ and $\mc{T}_{\phi}(\mc{E})\tn{2}$ are of the form
\begin{align}
\mc{T}_{\phi_{2}}(\mc{E}\tn{2}) &= \sum_{i\in \mc{Z}}\chi_i\mc{P}_i \\
\mc{T}_{\phi}(\mc{E})\tn{2} &= \sum_{i\in \mc{Z}}f^2 \mc{P}_i
\end{align}
where $\mc{Z}$ (as defined in \cref{lem:Clifford_irreps}) indexes the irreducible subrepresentations of the traceless symmetric representation of the Clifford group and $\chi_i = \tr(\mc{P}_i\mc{E}\tn{2})/\tr(\mc{P}_i)$ are the prefactors associated to the different subrepresentations. We also used that $\mc{T}_{\phi}(\mc{E})$ is a depolarizing channel with depolarizing parameter $f$~\cite{Wallman2014}. Using that $\mc{P}_i^2 = \mc{P}_i$ and $\mc{P}_i\mc{P}_j = 0$ for $i,j\in \mc{Z}, i\neq j$ we can rewrite the variance as
\begin{equation}
\bbV^2_m = \laa Q\tn{2}|\sum_{i\in \mc{Z}}\chi_i^m \mc{P}_i|\nu\tn{2}\raa  - \laa Q\tn{2}|\sum_{i\in \mc{Z}}f^{2m} \mc{P}_i|\nu\tn{2}\raa = \sum_{i\in \mc{Z}}\laa Q\tn{2}|\mc{P}_i|\nu\tn{2}\raa (\chi_i^m - f^{2m}).
\end{equation}
We now apply a telescoping series identity, which is proven in \cref{cor:second_order_series} of \cref{lem:telescoping series}, to the factor $\chi_i^m-f^{2m}$ in the above equation (for all $i\in \mc{Z}$). This gives
\begin{subequations}
\begin{align}
 \bbV^2_m&= mf^{2(m-1)} \sum_{i\in \mc{Z}}\laa Q\tn{2}|\mc{P}_i|\nu\tn{2}\raa(\chi_i-f^2) \label{eq:variance_first_term}\\&\hspace{15mm}+ \sum_{i\in \mc{Z}}\laa Q\tn{2}|\mc{P}_i|\nu\tn{2}\raa (\chi_i-f^2)^2\sum_{s=2}^m (s-1)\chi_i^{m-s} f^{2(s-2)}.\label{eq:variance_second_term}
\end{align}
\end{subequations}
This equation contains two terms, \cref{eq:variance_first_term} and \cref{eq:variance_second_term} which we will bound separately. We now proceed to upper bound the first term, that is \cref{eq:variance_first_term}. For this we will split the the input and output operators $Q,\nu$ into their ideal parts (that is, the Pauli operator $\sigma_{\mathbf{P}}: = \mathbf{P}/\sqrt{d}$) and deviations from that ideal. We define the functions
\begin{equation}
H_i(Q,\nu) := \laa Q\tn{2}|\mc{P}_i|\nu\tn{2}\raa - Q_{\mathbf{P}}^2\nu_{\mathbf{P}}^2 \laa \sigma_{\mathbf{P}}\tn{2} |\mc{P}_i|\sigma_{\mathbf{P}}\tn{2}\raa
\end{equation}
for all $i\in \mc{Z}$ where $Q_{\mathbf{P}} = \tr(Q\sigma_{\mathbf{P}}) $ and similarly for $\nu_{\mathbf{P}}$.
Using this we can write \cref{eq:variance_first_term} as
\begin{subequations}
\begin{align}\label{eq:first_term_first_step}
mf^{2(m-1)} \sum_{i\in \mc{Z}}\laa Q\tn{2}|\mc{P}_i|\nu\tn{2}\raa(\chi_i-f^2)&= Q_{\mathbf{P}}^2\nu_{\mathbf{P}}^2 mf^{2(m-1)} \sum_{i\in \mc{Z}}\laa \sigma_{\mathbf{P}}\tn{2}|\mc{P}_i|\sigma_{\mathbf{P}}\tn{2}\raa(\chi_i-f^2)\\
&\hspace{10mm} + mf^{2(m-1)} \sum_{i\in \mc{Z}}H_i(Q,\nu)(\chi_i-f^2).\label{eq:second_term_first_step}
\end{align}
\end{subequations}
Now consider the first term of the RHS, \cref{eq:first_term_first_step}. First note from \cref{lem:projector_lemma} that for $i\not\in \mc{Z}_{\mathrm{d}} =  \{\mathrm{tr},1,2\}$ we have $\mc{P}_i|\sigma_{\mathbf{P}}\tn{2}\raa =0$. 
Hence we have
\begin{align}
\begin{split}
Q_{\mathbf{P}}^2\nu_{\mathbf{P}}^2 mf^{2(m-1)} \sum_{i\in \mc{Z}}\laa \sigma_{\mathbf{P}}\tn{2}&|\mc{P}_i|\sigma_{\mathbf{P}}\tn{2}\raa(\chi_i-f^2)\\ &= Q_{\mathbf{P}}^2\nu_{\mathbf{P}}^2mf^{2(m-1)} \sum_{i\in\mc{Z}_{\mathrm{d}}} \laa \sigma_{\mathbf{P}}\tn{2}|\mc{P}_i|\sigma_{\mathbf{P}}\tn{2}\raa(\chi_i-f^2)\\
&=  Q_{\mathbf{P}}^2\nu_{\mathbf{P}}^2mf^{2(m-1)} \sum_{i\in\mc{Z}_{\mathrm{d}}} \frac{\tr(\mc{P}_i)}{d^2-1}\left(\frac{\tr(\mc{P}_i\mc{E}\tn{2})}{\tr(\mc{P}_i)}-f^2\right)\\
&= Q_{\mathbf{P}}^2\nu_{\mathbf{P}}^2mf^{2(m-1)} \left[\frac{1}{d^2-1}\tr\left[\sum_{i\in \mc{Z}_{\mathrm{d}}}\mc{P}_i\mc{E}\tn{2}\right] - f^2\right]\\
&= Q_{\mathbf{P}}^2\nu_{\mathbf{P}}^2mf^{2(m-1)}\left[ \frac{1}{d^2-1} \sum_{\tau\in \bsq} \laa\tau\tn{2}|\mc{E}\tn{2}|\tau\tn{2}\raa - f^2\right]
\end{split}
\end{align}
where we used \cref{lem:projector_lemma} in the first and second equalities and the fact that
\begin{equation}
\sum_{i\in \mc{Z}_{\mathrm{d}}}\mc{P}_i = \sum_{\tau\in \bsq} |\tau\tn{2}\raa\laa \tau \tn{2}|
\end{equation}
in the last equality (this can be seen from \cref{lem:Clifford_irreps}). Now we use \cref{lem:diagonal_squared_channel}  and the  fact that $Q_{\mathbf{P}}\nu_{\mathbf{P}}\leq 1/4$ to obtain an upper bound
\begin{equation}
Q_{\mathbf{P}}^2\nu_{\mathbf{P}}^2 mf^{2(m-1)} \sum_{i\in \mc{Z}}\laa \sigma\tn{2}|\mc{P}_i|\sigma\tn{2}\raa(\chi_i-f^2) \leq mf^{2(m-1)} \frac{d^2-2}{4(d-1)^2}r^2.
\end{equation}

This leaves us with the second term in the RHS, \cref{eq:second_term_first_step}. 
Here we cannot attain a bound that is quadratic in $r$. 
Instead we will attempt a bound that is linear in $r$ using \cref{lem:chi_f_bound}. We can write
\begin{align}
\begin{split}
mf^{2(m-1)} \sum_{i\in \mc{Z}}H_i(Q,\nu)(\chi_i-f^2)&\leq mf^{2(m-1)} \sum_{i\in \mc{Z}}|H_i(Q,\nu)| |\chi_i- f^2|\\
&\leq mf^{2(m-1)}\frac{2dr}{d-1} \sum_{i\in \mc{Z}}|H_i(Q,\nu)|
\end{split}
\end{align}
subject to the condition $r\leq \frac{1}{3}$. Writing $\eta(Q,\nu) := \sum_{i\in \mc{Z}}|H_i(Q,\nu)| $ we have a bound on \cref{eq:variance_first_term}.\\
We continue by upper bounding the second term in the variance, that is \cref{eq:variance_second_term}.
We again split off the ideal components of $Q$ and $\nu$ and write
\begin{align}\label{eq:eta}
\begin{split}
\sum_{i\in \mc{Z}}\laa Q\tn{2}|\mc{P}_i|\nu\tn{2}\raa& (\chi_i-f^2)^2 \sum_{s=2}^m (s-1)\chi_i^{m-s} f^{2(s-2)}\\
 &= Q_{\mathbf{P}}^2\nu_{\mathbf{P}}^2\sum_{i\in \mc{Z}}\laa\sigma_{\mathbf{P}}\tn{2}|\mc{P}_i|\sigma_{\mathbf{P}}\tn{2}\raa (\chi_i-f^2)^2 \sum_{s=2}^m (s-1)\chi_i^{m-s} f^{2(s-2)}\\
 		&\hspace{10mm}+\sum_{i\in \mc{Z}}H_i(Q,\nu) (\chi_i-f^2)^2 \sum_{s=2}^m (s-1)\chi_i^{m-s} f^{2(s-2)}\\
&\leq \frac{1}{4}\sum_{i\in \mc{Z}_{\mathrm{d}}}\frac{\tr(\mc{P}_i)}{d^2-1} (\chi_i-f^2)^2 \sum_{s=2}^m (s-1)\chi_i^{m-s} f^{2(s-2)}\\
 		&\hspace{10mm}+\sum_{i\in \mc{Z}}|H_i(Q,\nu)| (\chi_i-f^2)^2 \chi_i^{m-2}\sum_{s=2}^m (s-1)\chi_i^{m-s} f^{2(s-2)}
\end{split}
\end{align}
where we have used the definition of the function $H_i(Q,\nu)$, \cref{lem:projector_lemma} and the triangle inequality. Now we use \cref{lem:chi_f_bound} to upper bound this quantity as
\begin{align}
\begin{split}
\sum_{i\in \mc{Z}}\laa Q\tn{2}|\mc{P}_i|\nu\tn{2}\raa& (\chi_i-f^2)^2 \sum_{s=2}^m (s-1)\chi_i^{m-s} f^{2(s-2)}\\
&\leq \sum_{i\in \mc{Z}_{\mathrm{d}}}\frac{\tr(\mc{P}_i)}{d^2-1} \left(\frac{dr}{d-1}\right)^2\sum_{s=2}^m (s-1)\chi_i^{m-s} f^{2(s-2)}\\
 		&\hspace{10mm}+\sum_{i\in \mc{Z}}|H_i(Q,\nu)| \left(\frac{2dr}{d-1}\right)^2 \sum_{s=2}^m (s-1)\chi_i^{m-s} f^{2(s-2)}\\
 &\leq \frac{d^2 r^2}{(d-1)^2}\sum_{s=2}^m (s-1)\chi_i^{m-s} f^{2(s-2)}	\\
&\hspace{10mm}+\frac{4d^2 r^2}{(d-1)^2} \sum_{i\in \mc{Z}}|H_i(Q,\nu)| \sum_{s=2}^m (s-1)\chi_i^{m-s} f^{2(s-2)}
\end{split}
\end{align}
where we have used the fact that $\sum_{i \in \mc{Z}_{\mathrm{d}}}\tr(\mc{P}_i) = d^2-1$. It remains to deal with the last factor. This we do by using \cref{lem:unitarity_upper_bound} which states that $\chi_i\leq u$ for all $i \in \mc{Z}$, where $u$ is the unitarity of the channel $\mc{E}$. Writing again $\eta(Q,\nu) := \sum_{i\in \mc{Z}}|H_i(Q,\nu)| $ we then have
\begin{align}
\begin{split}
\sum_{i\in \mc{Z}}\laa Q\tn{2}|\mc{P}_i|\nu\tn{2}\raa& (\chi_i-f^2)^2 \sum_{s=2}^m (s-1)\chi_i^{m-s} f^{2(s-2)}\\
&\leq \frac{d^2 r^2}{(d-1)^2}\sum_{s=2}^m (s-1)u^{m-s} f^{2(s-2)}	\\
&\hspace{10mm}+\frac{4d^2 r^2}{(d-1)^2} \sum_{i\in \mc{Z}}|H_i(Q,\nu)| \sum_{s=2}^m (s-1)u^{m-s} f^{2(s-2)}
\end{split}
\end{align}
 We can further make sense of this quantity by using the well known series identity
\begin{equation}
\sum_{k=1}^{m}(k-1)x^{k-2} = \frac{(m-1)x^{m} - mx^{m-1} + 1}{(1-x)^2},\;\;\;\;\;\;\;\;\;\;\;m\in \mathrm{N},
\end{equation}
Factoring out a factor of $u^{m-2}$ and setting $x=f^2/u$ we obtain the following
\begin{align}
\begin{split}
\sum_{i\in \mc{Z}}\laa Q\tn{2}|\mc{P}_i|\nu\tn{2}\raa& (\chi_i-f^2)^2 \sum_{s=2}^m (s-1)\chi_i^{m-s} f^{2(s-2)}\\&\leq \frac{d^2 r^2}{(d-1)^2}(1+ 4\eta(Q,\nu))u^{m-2}\frac{(m-1)(f^2/u)^{m} - m(f^2/u)^{m-1} + 1}{(1-(f^2/u))^2}.
\end{split}
\end{align}
This finishes the upper bounding of \cref{eq:variance_second_term}. Gathering all terms we come to a final bound
\begin{align}\label{eq:variance_total_final}
\begin{split}
\bbV &\leq mf^{m-1} \frac{d^2-2}{4(d+1)^2}r^2 + \frac{d^2}{(d-1)^2}r^2 u^{m-2} \frac{(m-1)\left(f^2/u\right)^{m} - m\left(f^2/u\right)^{m-1} + 1}{(1-\left(f^2/u\right))^2}\\
	&\hspace{10mm}+ \eta(Q,\nu)mf^{m-1}r + \eta(Q,\nu) r^2 u^{m-2} \frac{(m-1)\left(f^2/u\right)^{m} - m\left(f^2/u\right)^{m-1} + 1}{(1-\left(f^2/u\right))^2}
\end{split}
\end{align}
which is the bound we set out to find.
\end{proof}
Noting that $f^2\leq u$ and that the factor
\begin{equation}
\frac{(m-1)\left(f^2/u\right)^{m} - m\left(f^2/u\right)^{m-1} + 1}{(1-\left(f^2/u\right))^2},
\end{equation}
is monotonically decreasing in $u$ we can upper bound this factor by taking the limit $u\rightarrow f^2$. This gives
\begin{equation}
\lim_{u\rightarrow f^2}  \frac{(m-1)\left(f^2/u\right)^{m} - m\left(f^2/u\right)^{m-1} + 1}{(1-\left(f^2/u\right))^2} = \frac{m(m-1)}{2}.
\end{equation}
which can be confirmed by an application of l'H$\hat{\text{o}}$pital's rule.
Plugging this in to \cref{eq:variance_total_final} we obtain \cref{variance_small_m}.

\subsection{State preparation and measurement (SPAM) terms}
In the central bound on the variance (~\cref{thm:variance bound}) we had to account for the fact that the variance can depend on how well the input states $\rho,\hat{\rho}$ and the output POVM $Q$ can be implemented. The ideal behavior of $\nu=\frac{1}{2}(\rho-\hat{\rho})$ and $Q$ are given by
\begin{align}
Q_{\mathrm{id}} &= \frac{1}{2}(\id + \mathbf{P})\\
\nu_{\mathrm{id}} &= \frac{\mathbf{P}}{2d}
\end{align}
where $\mathbf{P}$ is a pre-specified element of the Pauli group (see \cref{box:randomized_benchmarking}). The deviation of $Q$ and $\nu$ from this ideal can be captured by writing
\begin{align}
Q &= Q_{\mathrm{id}} + Q_{\mathrm{spam}}\\
\nu &= \nu_{\mathrm{id}} + \nu_{\mathrm{spam}}
\end{align}
where $\inp{Q_{\mathrm{id}}}{Q_{\mathrm{spam}}} =\inp{\nu_{\mathrm{id}}}{\nu_{\mathrm{spam}}} =0$. 

In the variance bound the deviation from the ideal has an effect which is measured by the parameter $\eta(Q,\nu)$. This parameter $\eta(Q,\nu)$ was defined as
\begin{equation}
\eta(Q,\nu)  = \sum_{i\in \mc{Z}}H_i(Q,\nu) = \sum_{i\in \mc{Z}} |\laa Q\tn{2}|\mc{P}_i|\nu\tn{2}\raa - \laa Q_{\mathrm{id}}\tn{2}|\mc{P}_i|\nu_{\mathrm{id}}\tn{2}\raa|
\end{equation}
where $\mc{Z}$ indexes the irreducible representations of the traceless symmetric representation of the Clifford group and the $\mc{P}_i$ are projectors onto the spaces carrying these subrepresentations (\cref{lem:Clifford_irreps}). Let us now analyze these terms further. For $i\in \mc{Z}_{\mathrm{d}}$ we have
\begin{align}
\begin{split}
H_i(Q,\nu) &= |\laa (Q_{\mathrm{id}} + Q_{\mathrm{spam}})\tn{2}|\mc{P}_i|(\nu_{\mathrm{id}} + \nu_{\mathrm{spam}})\tn{2}\raa - \laa Q_{\mathrm{id}}\tn{2}|\mc{P}_i|\nu_{\mathrm{id}}\tn{2}\raa|\\
&= |\laa Q_{\mathrm{id}}\tn{2}|\mc{P}_i|\nu_{\mathrm{spam}}\tn{2}\raa + \laa Q_{\mathrm{id}}\tn{2}|\mc{P}_i|\nu_{\mathrm{id}}\tn{2}\raa+ \laa Q_{\mathrm{spam}}\tn{2}|\mc{P}_i|\nu_{\mathrm{spam}}\tn{2}\raa|
\end{split}
\end{align}
where we have used that $\inp{Q_{\mathrm{id}}}{Q_{\mathrm{spam}}} =\inp{\nu_{\mathrm{id}}}{\nu_{\mathrm{spam}}} =0$ which implies that $\laa Q_{\mathrm{id}}\otimes Q_{\mathrm{spam}}|\mc{P}_i =\mc{P}_i|\nu_{\mathrm{id}}\otimes \nu_{\mathrm{spam}}\raa =0$ for $i \in \mc{Z}_{\mathrm{d}}$. Using the triangle inequality and the Cauchy-Schwarz inequality we can get
\begin{align}
\begin{split}
H_i(Q,\nu)&\leq |\laa Q_{\mathrm{id}}\tn{2}|\mc{P}_i|\nu_{\mathrm{spam}}\tn{2}\raa| + |\laa Q_{\mathrm{spam}}\tn{2}|\mc{P}_i|\nu_{\mathrm{id}}\tn{2}\raa|+ |\laa Q_{\mathrm{spam}}\tn{2}|\mc{P}_i|\nu_{\mathrm{spam}}\tn{2}\raa|\\
&\leq\norm{Q_{\mathrm{id}}\tn{2}}_2\norm{\mc{P}_i(\nu_{\mathrm{spam}}\tn{2})}_2+ \norm{Q_{\mathrm{spam}}\tn{2}}_2\norm{\mc{P}_i(\nu_{\mathrm{id}}\tn{2})}_2 + \norm{Q_{\mathrm{spam}}\tn{2}}_2\norm{\mc{P}_i(\nu_{\mathrm{spam}}\tn{2})}_2\\
&\leq \norm{\mc{P}_i}_{2\rightarrow2}\left(\norm{Q_{\mathrm{id}}}^2_2\norm{\nu_{\mathrm{spam}}}_{2}^2+ \norm{Q_{\mathrm{spam}}}_2^2\norm{\nu_{\mathrm{id}}}_2^2 + \norm{Q_{\mathrm{spam}}}_2^2\norm{\nu_{\mathrm{spam}}}_2^2\right)
\end{split}
\end{align}
where $\norm{\mc{P}_i}_{2\rightarrow2}$ is the induced $2$-norm of the superoperator $\mc{P}_i$. It is well known that this norm is equal to the largest singular value of the Liouville representation of $\mc{P}_i$~\cite{Wallman2014}, which since the Liouville representation of $\mc{P}_i$ is an orthonormal projection, is equal to one. This means we have for $i \in \mc{Z}_{\mathrm{d}}$ that
\begin{align}
\begin{split}
H_i(Q,\nu) &\leq \norm{Q_{\mathrm{id}}}^2_2\norm{\nu_{\mathrm{spam}}}_{2}^2+ \norm{Q_{\mathrm{spam}}}_2^2\norm{\nu_{\mathrm{id}}}_2^2 + \norm{Q_{\mathrm{spam}}}_2^2\norm{\nu_{\mathrm{spam}}}_2^2\\
&= \norm{Q_{\mathrm{id}}}^2_2\norm{\nu-\nu_{\mathrm{id}}}_{2}^2+ \norm{Q - Q_{\mathrm{id}}}_2^2\norm{\nu_{\mathrm{id}}}_2^2 + \norm{Q -Q_{\mathrm{id}}}_2^2\norm{\nu - \nu_{\mathrm{id}}}_2^2.
\end{split}
\end{align}
Note that this expression is zero when both $Q$ and $\nu$ are ideally implemented but is non-zero when either of them is not. This behavior is in general unavoidable as we argue in the main text (\cref{subsec:opt_max_var}). But first we will consider the functions $H_i(Q,\nu)$ for $i \in Z_{[S]}\cup Z_{\{S\}}$. Note first that
since $\mathrm{supp}(\mc{P}_i)\subset \vsp\{S_{\sigma,\sigma'}\;\|\;\;\sigma,\sigma' \in \bsq, \; \sigma \neq \sigma'\}$ we must have that $ \mc{P}_i|\nu_{\mathrm{id}}\tn{2}\raa =\laa Q_{\mathrm{id}}\tn{2}|\mc{P}_i =0$. This means we can write
\begin{align}
H_i(Q,\nu) &= |\laa Q\tn{2}|\mc{P}_i|\nu\tn{2}\raa - \laa Q_{\mathrm{id}}\tn{2}|\mc{P}_i|\nu_{\mathrm{id}}\tn{2}\raa|\\
\begin{split}
&= | \laa Q_{\mathrm{spam}}\tn{2}|\mc{P}_i|\nu_{\mathrm{id}}\!\otimes\!\nu_{\mathrm{spam}} \!\!+\!\!\nu_{\mathrm{spam}}\!\otimes\!\nu_{\mathrm{id}} \raa + \laa Q_{\mathrm{id}}\!\otimes \!Q_{\mathrm{spam}} \!\!+\!\!Q_{\mathrm{spam}}\!\otimes\! Q_{\mathrm{id}} |\mc{P}_i|\nu_{\mathrm{spam}}\tn{2}\raa \\&\hspace{7mm}+\!\laa Q_{\mathrm{spam}}\tn{2}|\mc{P}_i|\nu_{\mathrm{spam}}\tn{2}\raa + \laa  Q_{\mathrm{id}}\!\otimes\! Q_{\mathrm{spam}}\! +\! Q_{\mathrm{spam}}\!\otimes\! Q_{\mathrm{id}} |\mc{P}_i|\nu_{\mathrm{id}}\!\otimes\!\nu_{\mathrm{spam}}\! +\!\nu_{\mathrm{spam}}\!\otimes\!\nu_{\mathrm{id}} \raa|
\end{split}\\
\begin{split}
\leq& \norm{\mc{P}_i}_{2\rightarrow2}\bigg(\norm{Q_{\mathrm{spam}}\tn{2}}_2 \norm{\nu_{\mathrm{spam}}\tn{2}}_2 + 2\norm{Q_{\mathrm{spam}}}_2 \norm{Q_{\mathrm{id}}}_2\norm{\nu_{\mathrm{spam}}\tn{2}}_2 \\&\hspace{8mm}+2\norm{\nu_{\mathrm{spam}}}_2 \norm{\nu_{\mathrm{id}}}_2\norm{Q_{\mathrm{spam}}\tn{2}}_2 + 4 \norm{\nu_{\mathrm{spam}}}_2 \norm{\nu_{\mathrm{id}}}_2\norm{Q_{\mathrm{spam}}}_2 \norm{Q_{\mathrm{id}}}_2\bigg)
\end{split}
\end{align}
which we can rewrite as
\begin{align}
\begin{split}
H_i(Q,\nu) &\leq \norm{Q - Q_{\mathrm{id}}}_2\norm{\nu - \nu_{\mathrm{id}}}_2\bigg(\norm{Q - Q_{\mathrm{id}}}_2\norm{\nu - \nu_{\mathrm{id}}}_2 \\&\hspace{8mm}+ 2 \norm{\nu - \nu_{\mathrm{id}}}_2\norm{Q_{\mathrm{id}}}_2 + 2 \norm{Q - Q_{\mathrm{id}}}_2\norm{\nu_{\mathrm{id}}}_2 + 4\norm{\nu_{\mathrm{id}}}_2\norm{Q_{\mathrm{id}}}_2 \bigg)
\end{split}
\end{align}
which makes manifest that $H_i(Q,\nu)=0$ if $Q$ and $\nu$ are ideal and moreover that this term actually scales with the product of the deviations in $Q$ and $\nu$ (as measured in the $2$-norm). Hence we see that to lowest order in $Q_{\mathrm{spam}}$ and $\nu_{\mathrm{spam}}$ the SPAM parameter $\eta(Q,\nu)$ is proportional to
\begin{equation}
\eta \approx \norm{Q - Q_{\mathrm{id}}}_2\norm{\nu - \nu_{\mathrm{id}}}_2 + \norm{Q - Q_{\mathrm{id}}}_2^2 + \norm{\nu - \nu_{\mathrm{id}}}_2^2
\end{equation}
with the exact expression being
\begin{align}\label{eq:eta_full_term}
\begin{split}
\eta(Q,\nu)&\leq 3\left[\norm{Q_{\mathrm{id}}}^2_2\norm{\nu-\nu_{\mathrm{id}}}_{2}^2+ \norm{Q - Q_{\mathrm{id}}}_2^2\norm{\nu_{\mathrm{id}}}_2^2 + \norm{Q -Q_{\mathrm{id}}}_2^2\norm{\nu - \nu_{\mathrm{id}}}_2^2\right]\\
&\hspace{8mm}+ 5\bigg[\norm{Q - Q_{\mathrm{id}}}_2\norm{\nu - \nu_{\mathrm{id}}}_2\bigg(\norm{Q - Q_{\mathrm{id}}}_2\norm{\nu - \nu_{\mathrm{id}}}_2 \\&\hspace{20mm}+ 2 \norm{\nu - \nu_{\mathrm{id}}}_2\norm{Q_{\mathrm{id}}}_2 + 2 \norm{Q - Q_{\mathrm{id}}}_2\norm{\nu_{\mathrm{id}}}_2 + 4\norm{\nu_{\mathrm{id}}}_2\norm{Q_{\mathrm{id}}}_2 \bigg)\bigg]
\end{split}
\end{align}
where the factors $3$ and $5$ arise from the fact that $|\mc{Z}_{\mathrm{d}}|= 3$ and $|Z_{[S]}\cup Z_{\{S\}}| =5 $ respectively (this is for $q\geq 3$, for $q=1$ we get the significantly better $|\mc{Z}_{\mathrm{d}}|= 2$ and $|Z_{[S]}\cup Z_{\{S\}}| =1 $ instead~\cite{Clifford2016}).

\subsection{Sample complexity of iteratively reweighted least squares}\label{subsec:sample complexity}
In this section we analyze the sample complexity of the RB fitting procedure using iteratively reweighted least squares, as outlined in \cref{subsec:fitting}. Given a set of sequence lengths $\md{M}$ we will assume that $N$ random sequences are sampled for each sequence length. It is possible to let $N$ be a function of the sequence length $m$ and prove a more general version of the theorem presented here but we will not pursue this here. We will also only be interested in the uncertainty around the estimate for the depolarizing parameter $f$, it is straightforward to extend our analysis to also include the uncertainty around estimate for the pre-factor $A$. The methods we use are all standard and can be found in \cite{fedorov2013theory,Seber1989}. See also \cite{Epstein2014} for an earlier calculation of this form in the context of randomized benchmarking (not taking into account the heteroskedasticity of randomized benchmarking data).
\begin{theorem}
Let $\md{M}$ be a set of integers denoting sequence lengths and let $\{k_{m,N}\}_{m\in\md{M}}$ be a set of RB data points obtained by sampling $N$ random sequences for each sequence length $m\in \md{M}$. Denote by $f^*,A^*$ the true values for the RB fitting parameters and denote by $f_{\rm est},A_{\rm est}$ their estimates as obtained by the iteratively reweighted least squares procedure outlined in \cref{alg:IRLS}. We then have that
\begin{equation}
\text{Pr}\left[|f^* - f_{\rm est}|\leq \epsilon\right]\geq 1-\delta
\end{equation}
where $\delta$ is upper bounded by
\begin{equation}
\delta\leq 2H[\bbV_{\mathrm{fit}},\epsilon_{\rm fit}]^{N|\md{M}|}
\end{equation}
with $H$ defined in \cref{eq:concentration} and 
\begin{align}
\bbV_{\mathrm{fit}} &= \frac{1}{|\md{M}|}\sum_{m\in \md{M}} \bbV_m(f^*)w(f_{\rm est},m)\\
\epsilon_{\mathrm{fit}} &= \frac{\epsilon [J^T J]}{J_1}
\end{align}
and 
\begin{equation}
J = \left[-\frac{1}{|\md{M}|}\sum_{m\in \md{M}} mA^* {f^*}^{m-1}w(f^*)\;\;,\;\;\frac{1}{|\md{M}|}\sum_{m\in \md{M}}{f^*}^m w(f^*,m) \right]
\end{equation}
and $J_{1}$ is the first entry of this vector.

\end{theorem}
\begin{proof}
The starting off point for this proof is given by Eq. 1.6.19 in \cite[Page 45]{fedorov2013theory} which states that the outcome of the IRLS procedure satisfies the following equality
\begin{equation}\label{eq:zero_point}
\frac{1}{|\md{M}|}\sum_{m\in \md{M}}(k_{m,N} - A_{\rm est}f_{\rm est}^m)w(f_{\rm est},m) = 0
\end{equation}
where $w(f,m)$ is the weight function given by the inverse of \cref{eq:variance} (we suppress the dependency on the unitarity here for notational simplicity). 
We can rewrite \cref{eq:zero_point} as
\begin{gather}
\frac{1}{|\md{M}|}\sum_{m\in \md{M}}(k_{m,N} +A^*{f^*}^m -A^*{f^*}^m - A_{\rm est}f_{\rm est}^m)w(f_{\rm est},m)=0\\
\iff\frac{1}{|\md{M}|}\sum_{m\in \md{M}}(A^*{f^*}^m - A_{\rm est}f_{\rm est}^m)w(f_{\rm est},m) = -\frac{1}{|\md{M}|}\sum_{m\in \md{M}}(k_{m,N} - A^* {f^*}^m)w(f_{\rm est},m).\label{eq:zero_point_other}
\end{gather}

We can now think of the LHS of \cref{eq:zero_point_other} as a function of the vector $[f_{\rm est},A_{\rm est}]$. Assuming $[f_{\rm est},A_{\rm est}]$ is close to $[f^*,A^*]$ we can expand the LHS of \cref{eq:zero_point_other} to first order to get
\begin{equation}\label{eq:taylor exp}
\frac{1}{|\md{M}|}\sum_{m\in \md{M}}(A^*{f^*}^m - A_{\rm est}{f_{\rm est}^m})w(f_{\rm est},m) \approx J\left[f^*-f_{\rm est}, A^* - A_{\rm est}\right]^T
\end{equation}
where $J$ is the Jacobian associated to the LHS of \cref{eq:zero_point_other}, that is:
\begin{equation}
J = \left[-\frac{1}{|\md{M}|}\sum_{m\in \md{M}} mA^* {f^*}^{m-1}w(f^*)\;\;,\;\;\frac{1}{|\md{M}|}\sum_{m\in \md{M}}{f^*}^m w(f^*,m) \right].
\end{equation}
Taking the Moore-Penrose inverse $J^{\rm MP} = (J^T J)^{-1} J^T$ of $J$  and inserting this in the first entry of \cref{eq:taylor exp} we can say that
\begin{equation}
f^*-f_{\rm est} \approx (J J^T)^{-1} J_1\frac{1}{|\md{M}|}\sum_{m\in \md{M}}(A^*{f^*}^m - A_{\rm est}{f_{\rm est}^m})w(f_{\rm est},m)
\end{equation}
where $J_1$ is the first entry of $J$. Now we can say that
\begin{align}
\text{Prob}\left[|f^* - f_{\rm est}| \geq \epsilon\right] &\approx \text{Prob}\left[\left|[J J^T]^{-1} J_1\frac{1}{|\md{M}|}\sum_{m\in \md{M}}(A^*{f^*}^m - A_{\rm est}f_{\rm est}^m)w(f_{\rm est},m)\right| \geq \epsilon\right]\\
&=\text{Prob}\left[\left|[J J^T]^{-1} J_1\frac{1}{|\md{M}|}\sum_{m\in \md{M}}(k_{m,N} - A^* {f^*}^m)w(f_{\rm est},m)\right| \geq \epsilon\right]\label{eq:indepen}
\end{align}
Now note that $k_{m,N}$ can be seen as a number drawn from a random variable $K_m$ with mean $A^*(f^*)^m$ and variance $\bbV_m(f^*)/N^2$ where $N$ is the number of random sequences drawn for each data-point $k_{m,N}$. Moreover $k_{m,N}$ and $k_{N,m'}$ for $m\neq m'$ are drawn from independent random variables $K_m$ and $K_{m'}$. 
Hence we can apply the concentration inequality given in \cref{eq:concentration} to \cref{eq:indepen} to get
\begin{equation}
\text{Prob}\left[|f^* - f_{\rm est}| \geq \epsilon\right] \leq 2H[\bbV_{\mathrm{fit}},\epsilon_{\rm fit}]^{N|\md{M}|}
\end{equation}
with $\bbV_{\mathrm{fit}},\epsilon_{\mathrm{fit}}$ given by
\begin{align}
\bbV_{\mathrm{fit}} &= \frac{1}{|\md{M}|}\sum_{m\in \md{M}} \bbV_m(f^*)w(f_{\rm est},m)\\
\epsilon_{\mathrm{fit}} &= \frac{\epsilon [J J^T]}{J_1}
\end{align}
which completes the proof.
\end{proof}
Using \cref{variance_final} or \cref{variance_final_SPAM} then gives an upper bound on total amount of data that needs to be gathered for rigorous RB.

\section{Technical lemmas}
In this section we give proofs of all technical lemmas used in the main result \cref{thm:variance bound}.

\subsection{Projectors in the traceless symmetric subspace}
In \cref{lem:projector_lemma} we prove a series of useful upper bounds on the trace overlap between the superoperator-projectors associated to the traceless-symmetric representation of the Clifford group and the normalized Pauli matrices. The saturated versions of these inequalities are critical to establishing the quadratic scaling with infidelity of the variance bound in the case of SPAM-free RB.
\begin{lemma}\label{lem:projector_lemma}
Let $\mc{E}:\M \to \M$ be a quantum channel and consider the twirled operator $\mc{T}_{\phi_{\mathrm{TS}}}(\mc{E}\tn{2})$ with respect to the traceless-symmetric representation. This operator can then be written as (\cref{lem:Clifford_irreps,lem:Schur})
\begin{equation}
\mc{T}_{\phi_{\mathrm{TS}}}(\mc{E}\tn{2}) = \sum_{i\in \mc{Z}}\frac{\tr(\mc{E}\mc{P}_i)}{\tr(\mc{P}_i)}\mc{P}_i
\end{equation}
with $\mc{Z} = \{\mathrm{tr},1,2,[1],[2],[3],\{1\},\{2\}\}$ and $\mc{P}_i$ the projector onto the spaces $V_i\subset \mc{M}_d\tn{2}$. Let $I(x\in A)$ be the indicator function for the set $A$ (that is $I(x\in A) =1$ if $x\in A$ and $I(x\in A) =0$ otherwise). We have the following statements
\begin{itemize}
	\item
 For $i \in \mc{Z}$ and $\sigma,\sigma' \in \bsq$ we have that
\begin{equation}
|\laa \sigma\tn{2} |\mc{P}_i| {\sigma'}\tn{2} \raa| = |\laa \sigma\tn{2} |\mc{P}_i| {\sigma'}\tn{2} \raa| I(i\in\mc{Z}_{\mathrm{d}}) \leq \frac{\tr(\mc{P}_i)I(i\in\mc{Z}_{\mathrm{d}})}{d^2-1}
\end{equation}
with equality when $\sigma =\sigma'$ .
\item
For $i \in \mc{Z}$, $\tau,\tau'\in \bsq$ and $\sigma\in \bf{C}_\tau,\sigma'\in \bf{C}_{\tau'}$ we have that
\begin{equation}
|\laa S_{\sigma,\sigma\cdot\tau} |\mc{P}_i| S_{\sigma',\sigma'\cdot\tau'} \raa| = |\laa S_{\sigma,\sigma\cdot\tau} |\mc{P}_i| S_{\sigma',\sigma'\cdot\tau} \raa| I(i\in\mc{Z}_{[S]})\delta_{\tau,\tau'} \leq \frac{2\tr(\mc{P}_i)I(i\in\mc{Z}_{[S]})\delta_{\tau,\tau'}}{(d^2-1)(d^2/2-2)}
\end{equation}
with equality when $\sigma = \sigma'$.
\item
For $i \in \mc{Z}$, $\tau,\tau'\in \bsq$ and $\sigma\in \bf{N}_\tau,\sigma'\in \bf{N}_{\tau'}$ we have that
\begin{equation}
|\laa S_{\sigma,i\sigma\cdot\tau} |\mc{P}_i| S_{\sigma',i\sigma'\cdot\tau'} \raa| = |\laa S_{\sigma,i\sigma\cdot\tau} |\mc{P}_i| S_{\sigma',i\sigma'\cdot\tau} \raa| I(i\in\mc{Z}_{\{S\}})\delta_{\tau,\tau'} \leq \frac{2\tr(\mc{P}_i)I(i\in\mc{Z}_{\{S\}})\delta_{\tau,\tau'}}{(d^2-1)(d^2/2)}
\end{equation}
with equality when $\sigma = \sigma'$.
\end{itemize}
where the sets $Z_{\mathrm{d}},\mc{Z}_{[S]},\mc{Z}_{\{S\}}$ are defined in \cref{lem:Clifford_irreps}.
\end{lemma}

\begin{proof}
We begin by proving the first claim. Let $\mc{P}_i$ be a projector as defined in the lemma statement with $i\in \mc{Z}$ and take $\sigma, {\sigma'}\in \bsq$. From \cref{lem:Clifford_irreps} we have immediately that
\begin{equation}
\laa \sigma\tn{2} |\mc{P}_i| {\sigma'}\tn{2} \raa = \laa \sigma\tn{2} |\mc{P}_i| {\sigma'}\tn{2} \raa I(i\in\mc{Z}_{\mathrm{d}}).
\end{equation}
Now consider $i \in \mc{Z}_{\mathrm{d}}$. Note that since $\mc{P}_i$ is a projector it is a real matrix and we have that $\mc{P}_i\geq 0$, that is $\mc{P}_i$ is a positive semidefinite matrix. This means that we have, by the Sylvester principal minor conditions, that
\begin{equation}
|\laa \sigma\tn{2} |\mc{P}_i| {\sigma'}\tn{2} \raa| \leq \sqrt{\laa {\sigma'}\tn{2} |\mc{P}_i| {\sigma'}\tn{2} \raa\laa \sigma\tn{2} |\mc{P}_i| {\sigma}\tn{2} \raa}
\end{equation}
for all $\sigma, \sigma' \in \bsq$. Now consider the case $\sigma = \sigma'$. Note that for all $\tau,\sigma\in \bsq$ there is a $G_\tau^\sigma\in \md{C}$ such that $\mc{G}_\tau^\sigma(\tau) = \pm \sigma$. That is, the Clifford group acts transitively on $\bsq$~\cite{Zhu2015}. This means we can write
\begin{align}
\begin{split}
\laa \sigma\tn{2} |\mc{P}_i| {\sigma}\tn{2} \raa  &= \frac{1}{d^2-1}\sum_{\tau \in \bsq}\laa \mc{G}_\tau^\sigma(
\tau)\tn{2} |\mc{P}_i| \mc{G}_\tau^\sigma(
\tau) \raa\\
&= \frac{1}{d^2-1}\sum_{\tau \in \bsq}\laa \tau\tn{2}|{\big(\mc{G}_\tau^\sigma\big)\ct}\tn{2}\mc{P}_i {\big(\mc{G}_\tau^\sigma\big)}\tn{2}|\tau\tn{2}\raa\\
&= \frac{1}{d^2-1}\sum_{\tau \in \bsq}\laa \tau\tn{2}|\mc{P}_i|\tau\tn{2}\raa\\
&= \frac{\tr(\mc{P}_i)}{d^2-1}
\end{split}
\end{align}
where we used the fact that $\mc{P}_i$ commutes with $\mc{G}\tn{2}$ for all $G\in \md{C}$ and the fact that $V_i\subset V_{\mathrm{d}}$ (where $V_{\mathrm{d}}$ is defined in \cref{lem:Clifford_irreps}). This proves the first claim of the lemma. 

Next we consider the second claim of the lemma. Let $\tau, \tau'\in\bsq$ and take $\sigma\in \bf{C}_\tau$ and $\sigma'\in \bf{C}_{\tau'}$. Again from \cref{lem:Clifford_irreps} we have immediately that
\begin{equation}
\laa S_{\sigma,\sigma\cdot\tau} |\mc{P}_i| S_{\sigma',\sigma'\cdot\tau'} \raa = \laa S_{\sigma,\sigma\cdot\tau} |\mc{P}_i| S_{\sigma',\sigma'\cdot\tau'} \raa I(i\in\mc{Z}_{[S]}).
\end{equation}
Now consider $i \in \mc{Z}_{[S]}$. From \cref{lem:Clifford_irreps} we have that we can write
\begin{equation}
\mc{P}_i = \sum_{\tau\in \bsq} \mc{P}_i^\tau
\end{equation}
where $\mc{P}^\tau_i$ has support in the space
\begin{equation}
V^{[\tau]} = \{S_{\sigma,\sigma\cdot\tau}\;\;\|\;\;\sigma \in \bf{C}_\tau\}.
\end{equation}
From this we immediately get
\begin{equation}
\laa S_{\sigma,\sigma\cdot\tau} |\mc{P}_i| S_{\sigma',\sigma'\cdot\tau'} \raa = \laa S_{\sigma,\sigma\cdot\tau} |\mc{P}_i| S_{\sigma',\sigma'\cdot\tau'} \raa\delta_{\tau,\tau'}.
\end{equation}
Now consider $\tau =\tau'$. Again from the Sylvester minor conditions we get for all $\sigma,\sigma'\in \bf{C}_\tau$ that
\begin{equation}
|\laa S_{\sigma,\sigma\cdot\tau} |\mc{P}_i| S_{\sigma',\sigma'\cdot\tau}\raa| \leq \sqrt{\laa S_{\sigma',\sigma'\cdot\tau} |\mc{P}_i| S_{\sigma',\sigma'\cdot\tau}\raa \laa S_{\sigma,\sigma\cdot\tau} |\mc{P}_i| S_{\sigma,\sigma\cdot\tau}\raa}.
\end{equation}
Now consider the case $\sigma =\sigma'$. From \cite{Zhu2015} we can see that the action of the Clifford group on the set $A = \{(\sigma,\sigma\cdot\tau)\;\|\;\;\tau\in \bsq,\; \sigma \in \bf{C}_\tau \}$ is $2$-transitive. That is, for all pairs $(\nu, \mu) \in A$ there is a $G_{\nu,\mu}^{\sigma,\tau}\in \md{C}$ such that
\begin{equation}
{\mc{G}_{\nu,\mu}^{\sigma,\tau}}\tn{2}\big(S_{\sigma,\sigma\cdot\tau}\big) = S_{\nu, \nu\cdot\mu}.
\end{equation}
This implies we can make essentially the same argument as before, that is
\begin{align}
\begin{split}
\laa S_{\sigma,\sigma\cdot\tau} |\mc{P}_i| S_{\sigma,\sigma\cdot\tau} \raa & = \frac{1}{|A|}\sum_{(\mu,\nu)\in A} \laa S_{\nu,\nu\cdot\mu} |{\left(\mc{G}_{\nu,\mu}^{\sigma,\tau}\right)\ct}\tn{2} \mc{P}_i {\mc{G}_{\nu,\mu}^{\sigma,\tau}}\tn{2}| S_{\nu,\nu\cdot\mu} \raa\\
& = \frac{1}{|A|}\sum_{(\mu,\nu)\in A} \laa S_{\nu,\nu\cdot\mu} |\mc{P}_i| S_{\nu,\nu\cdot\mu} \raa\\
&= \frac{2\tr(\mc{P}_i)}{(d^2-1)(d^2/2-2)}
\end{split}
\end{align}
where we have used the fact that $\mc{G}\tn{2}$ commutes with $\mc{P}_i$ for all $G \in \md{C}$ and also the definition of the space $V_{[S]}$ (given in \cref{lem:Clifford_irreps}). The factor of two appears from the fact that the set $A$ counts the basis of $V_{[S]}$ twice since $S_{\nu, \nu\cdot\mu} =S_{\nu\cdot\mu, \nu} $ for all $(\mu, \nu\cdot\mu)\in A$. We have also used that $|A| =| {\bsq}| | {\bf{C}_\tau}| = (d^2-1)(d^2/2-2)$. This proves the second claim of the lemma.

The proof of the third claim of the lemma proceeds in the same way as the proof of the second claim with the difference that anti-commuting, rather than commuting elements of the Pauli group must considered. We will not write it down explicitly.
\end{proof}

\subsection{Bound on sum of squares of the diagonal elements of a quantum channel}

This lemma (\cref{lem:diagonal_squared_channel}) proves that the diagonal elements of a CPTP map are generically quite close to their mean. The key technique used here is the fact that the diagonal elements of a CPTP map are invariant under Pauli twirling. This is a structural result about quantum channels on arbitrarily many qubits and might be of independent interest. We use it to establish the quadratic scaling of the variance in the infidelity in the case of SPAM-free RB.
\begin{lemma}\label{lem:diagonal_squared_channel}
Let $\mc{E}:\M \to \M$ be a quantum channel with infidelity $r$ and depolarizing parameter $f = 1- \frac{dr}{d-1}$. The quantity
\begin{equation}\label{quantity}
\frac{1}{d^2-1} \sum_{\tau \in \boldsymbol{\sigma}_q} \mc{E}^2_{\tau,\tau},
\end{equation}
where $\mc{E}_{\tau,\tau} = \inp{\tau}{\mc{E}(\tau)}$, has the following upper and lower bounds in terms of the infidelity $r$
\begin{align}\label{diag_squared_eq}
f^2 =1-\frac{2d}{d-1}r + \frac{d^2}{(d-1)^2} r^2 \leq \frac{1}{d^2-1}\sum_{\tau \in \boldsymbol{\sigma}_q}\mc{E}^2_{\tau,\tau} \leq 1 - \frac{2d}{d-1}r + \frac{2(d+1)}{(d-1)}r^2.
\end{align}
\end{lemma}

\begin{proof}
We begin by noting that upper and lower bounds of the quantity \cref{quantity} can be found by maximizing and minimizing respectively the following optimization
\begin{align}\label{optimization}
\begin{aligned}
& \underset{\{\mc{E}_{\tau\tau}\}_\tau}{\text{max (min)}}
& & \sum_{\tau \in \boldsymbol{\sigma}_q}\mc{E}^2_{\tau,\tau} \\
& \text{subject to}
& & \sum_{\tau \in \boldsymbol{\sigma}_q} \mc{E}_{\tau,\tau} =(d^2-1)f\\
& & &  \mc{E} \text{  a CPTP map}.
\end{aligned}
\end{align}
Here we maximize (minimize) the quantity \cref{quantity} over all possible CPTP maps which have depolarizing parameter $f$ .
Solving this optimization problem is not easy since it not clear how to express the CP condition in terms of the optimization parameters $\mc{E}_{\tau\tau}$. We will therefore relax this problem to an easier one which we can solve. We begin by noting that the optimization variables $\mc{E}_{\tau\tau}$ are invariant under the action of a Pauli channel, i.e. for all $G\in\md{P}$ with $\md{P}$ the Pauli group, we have that
\begin{align}\label{equiv}
\begin{split}
(\mc{G}\ct \mc{E}\mc{{G}})_{\tau,\tau}  = \inp{\tau}{G\mc{E}(G\ct \tau G)G\ct} &= \inp{G\ct\tau G}{\mc{E}(G\ct\tau G)}\\  
&= \left[\text{sgn}(\tau,G)\right]^2 \inp{G\ct G\tau }{\mc{E}(G\ct G\tau )}   = \inp{\tau}{\mc{E}(\tau)} = \mc{E}_{\tau,\tau},
\end{split}
\end{align}
for all $\tau \in \boldsymbol{\sigma}_q\cup\sigma_0$ where $\text{sgn}(\tau, G)$ is defined as
\begin{equation}\label{eq:sign_sign_function_pauli}
\text{sgn}(\tau, G) = \begin{cases} -1 \hspace{3mm} \text{if}\hspace{3mm} \{\tau, G\} = 0,\\
+1 \hspace{3mm} \text{if}\hspace{3mm} [G, \tau] = 0,
\end{cases}
\end{equation}
which, since $\tau\in \boldsymbol{\sigma}_q\cup\sigma_0 $ is a normalized element of the Pauli group, is well defined because elements of the multi-qubit Pauli group can either commute $([.,.])$ or anti-commute $(\{.,.\})$ with each other~\cite{Gottesman1998}. By \cref{equiv} and linearity we can now note that the optimization variables in the optimization \cref{optimization} are invariant under twirling over the Pauli group $\md{P}$, i.e.
\begin{equation}
\mc{T}_P(\mc{E})_{\tau,\tau}  =\frac{1}{|\md{P}|} \sum_{G\in \md{P}} \inp{G\ct\tau G}{\mc{E}(G\ct\tau G)} =\frac{1}{|\md{P}|} \sum_{G\in \md{P}} \mc{E}_{\tau,\tau} = \mc{E}_{\tau,\tau}.
\end{equation}
Note also that the ``twirl'' operation, for any group, preserves complete positivity~\cite{Wolf2012}. This means we can relax the optimization \cref{optimization} to
 \begin{align}\label{optimization_pauli}
\begin{aligned}
& \underset{\{\mc{T}_{\md{P}}(\mc{E})_{\tau,\tau}\}_\tau}{\text{max (min)}}
& & \sum_{\tau \in \boldsymbol{\sigma}_q}\mc{T}_{\md{P}}(\mc{E})^2_{\tau,\tau} \\
& \text{subject to}
& & \sum_{\tau \in \boldsymbol{\sigma}_q} \mc{T}_{\md{P}}(\mc{E})_{\tau,\tau} =(d^2-1)f\\
& & & \mc{T}_{\md{P}}(\mc{E}) \text{  a CPTP map.}
\end{aligned}
\end{align}
Note that this is a relaxation of the previous optimization because while the twirl of a CP map will always be CP the opposite need no be true.
Now we use the following result due to Holevo~\cite{Holevo2005} which states that any CPTP map $\mc{E}$, twirled over the Pauli group, is of the form
\begin{align}
\mc{T}_{\md{P}}(\mc{E})(X) = \sum_{G\in \md{P}} p_G G X G\ct \hspace{10mm}\forall X \in \bbC^{d\times d},
\end{align}
where $\{p_G\}_G$ is a probability distribution, i.e. $p_G \geq 0, \forall G\in \md{P}$ and $\sum_{G\in\md{P}}p_G = 1$. Let us now rewrite the optimization \cref{optimization_pauli} in terms of this probability distribution. We begin by noting that since $\mc{E}$ is TP we have that $\mc{E}_{\sigma_0\sigma_0} =1$ and hence we can write the depolarizing constraint in \cref{optimization_pauli} as
\begin{equation}
\sum_{\tau \in \boldsymbol{\sigma}_q} \mc{T}_{\md{P}}(\mc{E})_{\tau,\tau} =(d^2-1)f \iff\sum_{\tau \in \boldsymbol{\sigma}_q\cup\sigma_0} \mc{T}_{\md{P}}(\mc{E})_{\tau,\tau} =(d^2-1)f + 1.
\end{equation}
Now, using the form of the Pauli-twirled channel, we can write the RHS of this equivalence as
\begin{align}
\begin{split}
\sum_{\tau \in \boldsymbol{\sigma}_q\cup\sigma_0} \mc{T}_{\md{P}}(\mc{E})_{\tau,\tau} &= \sum_{\tau \in \boldsymbol{\sigma}_q\cup\sigma_0} \sum_{G \in \md{P}} p_G\inp{\tau}{G\tau G\ct}\\
&= \sum_{G\in \md{P}} p_G  \sum_{\tau \in \boldsymbol{\sigma}_q\cup\sigma_0} \text{sgn}(\tau, G)\\
&= p_{I}d^2,
\end{split}
\end{align}
where in the last line we used that the identity Pauli element $I$ commutes with all Pauli matrices $\tau \in \boldsymbol{\sigma}_q\cup\sigma_0$, whereas every non-identity Pauli $G$ commutes with exactly of the elements of $\boldsymbol{\sigma}_q\cup\sigma_0$ and anti-commutes with the other half. We also used that $|\boldsymbol{\sigma}_q\cup\sigma_0| = d^2$. We can make a similar calculation for the objective of \cref{optimization_pauli} which gives
\begin{align}
\begin{split}
\sum_{\tau \in \boldsymbol{\sigma}_q}\mc{T}_{\md{P}}(\mc{E})^2_{\tau,\tau} &= \sum_{\tau \in \boldsymbol{\sigma}_q\cup \sigma_0} \mc{T}_{\md{P}}(\mc{E})^2_{\tau,\tau} - 1\\
&= (-1)+\sum_{\tau \in \boldsymbol{\sigma}_q\cup \sigma_0 }\left(\sum_{G\in \md{P}} p_G \inp{\tau}{G\tau G\ct}\right)^2\\
&= (-1)+ \sum_{G, \hat{G} \in \md{P}} p_G p_{\hat{G}} \sum_{\tau \in \boldsymbol{\sigma}_q\cup \sigma_0 } \text{sgn}(\tau, G)\text{sgn}(\tau, \hat{G}\ct)\\
&= (-1) +  \sum_{G \in \md{P}} p_G^2 \sum_{\tau \in \boldsymbol{\sigma}_q\cup \sigma_0 }\text{sgn}(\tau, GG\ct)+ \sum_{\substack{G,\hat{G}\in\md{P} \\ G\neq\hat{G}}}p_G p_{\hat{G}}\;\sum_{\tau \in \boldsymbol{\sigma}_q\cup \sigma_0 }\text{sgn}(\tau, G\hat{G}\ct)\\
&= (-1) + d^2 \sum_{G\in \md{P}}p_G^2,
\end{split}
\end{align}
 where we have used that $\text{sgn}(\tau,G)\text{sgn}(\tau,\hat{G}) = \text{sgn}(\tau, G\hat{G})$, that $G G\ct = I , \forall G\in \md{P}$ and again that the Pauli identity $I$ commutes with all elements of $\boldsymbol{\sigma}_q\cup\sigma_0$ while every non-identity Pauli $G\hat{G}\ct, G\neq G\hat{G}$ commutes with exactly half of the elements of $\boldsymbol{\sigma}_q\cup\sigma_0$ and anti-commutes with the other half. We have now rewritten the optimization \cref{optimization_pauli} completely in terms of the probability distribution $\{p_G\}_G$. This becomes
  \begin{align}\label{optimization_distr}
\begin{aligned}
& \underset{\{p_G\}_G}{\text{max (min)}}
& & (-1) + d^2 \sum_{G\in \md{P}}p_G^2 \\
& \text{subject to}
& & d^2 p_I =(d^2-1)f + 1\\
& & &\sum_{G\in \md{P}}p_G = 1\\
& & & p_G \geq 0 \hspace{10mm}G\in \md{P}.
\end{aligned}
\end{align}
Noting that the element $p_I$ is essentially fixed we can eliminate this element from the optimization and obtain an even simpler optimization
  \begin{align}\label{optimization_distr_2}
\begin{aligned}
& \underset{\{p_G\}_G}{\text{max (min)}}
& & (-1) + d^2 \sum_{G\in \md{P}/\{I\}}p_G^2  + d^2\left(\frac{d^2-1}{d^2}f +\frac{1}{d^2}\right)^2\\
& \text{subject to}
& & \sum_{G\in \md{P}/\{I\}}p_G = 1 - \frac{d^2-1}{d^2}f -\frac{1}{d^2}\\
& & & p_G \geq 0 \hspace{10mm}G\in \md{P}/\{I\}.
\end{aligned}
\end{align}
The above optimization is a well studied instance of a class of optimization problems called quadratic programs~\cite{Boyd2004}. This problem has the minimum~\cite[Chapter 4, Section 4]{Boyd2004}:
\begin{equation}
p_{G, min} = \frac{1}{d^2-1} \left( 1 - \frac{d^2-1}{d^2}f -\frac{1}{d^2}\right) \hspace{10mm}\forall G \in \md{P}/\{I\},
\end{equation}
and has $d^2-1$ degenerate maxima indexed by the non-identity Pauli elements $\tilde{G}$ of the form
\begin{equation}
p_{G, max} = \begin{cases} 1 - \frac{d^2-1}{d^2}f -\frac{1}{d^2} \hspace{3mm} \text{if} \hspace{3mm} G = \tilde{G}\\
0 \hspace{6mm} \text{otherwise}.
\end{cases}
\end{equation}
This means we can lower bound the quantity \cref{quantity}, for any CPTP map $\mc{E}$, by:
\begin{align}
\frac{1}{d^2-1} \sum_{\tau \in \boldsymbol{\sigma}_q}\mc{E}^2_{\tau,\tau} \geq \frac{d^2}{d^2-1}\left(\frac{d^2-1}{d^2}f +\frac{1}{d^2}\right)^2 + \frac{d^2}{(d^2-1)^2} \left( 1 - \frac{d^2-1}{d^2}f -\frac{1}{d^2}\right)^2 - \frac{1}{d^2-1}.
\end{align}
By now using the relation $f = 1- \frac{dr}{d-1}$ we can rewrite this lower bound in terms of the infidelity $r$. This process is straightforward but rather tedious so we will not write it down. At the end of the calculation we obtain
\begin{equation}
\frac{1}{d^2-1} \sum_{\tau \in \boldsymbol{\sigma}_q}\mc{E}^2_{\tau,\tau} \geq 1 - \frac{2dr}{d-1} + \frac{d^2r^2}{(d-1)^2}.
\end{equation}
Similarly we can write for the upper bound
\begin{equation}
\frac{1}{d^2-1} \sum_{\tau \in \boldsymbol{\sigma}_q}\mc{E}^2_{\tau,\tau} \leq \frac{d^2}{d^2-1}\left( \frac{d^2-1}{d^2}f + \frac{1}{d^2}\right)^2 + \frac{d^2}{d^2-1} \left( 1 - \frac{d^2-1}{d^2}f -\frac{1}{d^2}\right)^2 - \frac{1}{d^2-1},
\end{equation}
which, by essentially the same tedious but straightforward calculation yields
\begin{equation}
\frac{1}{d^2-1} \sum_{\tau \in \boldsymbol{\sigma}_q}\mc{E}_{\tau\tau} \leq 1 - 2\frac{dr}{d-1} +\frac{2(d+1)}{(d-1)}r^2,
\end{equation}
which completes the lemma.

\end{proof}

\subsection{Eigenvalues of twirled quantum channels}
\Cref{lem:unitarity_upper_bound} proves that the unitarity upper bounds the eigenvalues of the twirled superoperator $\mc{T}_{\phi_{\mathrm{TS}}}(\mc{E}\tn{2})$. This resolves an open question posed in \cite{Wallman2014} and allows us to establish the long sequence length behavior of the variance of RB.
\begin{lemma}\label{lem:unitarity_upper_bound}
Let $\mc{E}:\M \to \M$ be a quantum channel with unitarity $u$ and consider the twirled operator $\mc{T}_{\phi_{\mathrm{TS}}}(\mc{E}\tn{2})$ with respect to the traceless-symmetric representation. This operator can then be written as (\cref{lem:Clifford_irreps,lem:Schur})
\begin{equation}
\mc{T}_{\phi_{\mathrm{TS}}}(\mc{E}\tn{2}) = \sum_{i\in \mc{Z}} \chi_i\mc{P}_i
\end{equation}
with $\mc{Z} = \{\mathrm{tr},1,2,[1],[2],[3],\{1\},\{2\}\}$, $\mc{P}_i$ the projector onto the spaces $V_i\subset \mc{M}_d\tn{2}$ and
\begin{equation}
\chi_i := \frac{\tr(\mc{E}\mc{P}_i)}{\tr(\mc{P}_i)},
\end{equation}
where the trace is taken over superoperators.
We now have for all $i\in \mc{Z}$ that
\begin{equation}
\chi_i \leq u.
\end{equation}
\end{lemma}
\begin{proof}
We begin by considering $i\in \mc{Z}_{\mathrm{d}}$. Note first that for $i =  \mathrm{tr}$ we have that
\begin{equation}
\chi_{\mathrm{i}} = \frac{\tr(\mc{P}_{\mathrm{tr}}\mc{E}\tn{2})}{\tr(P_{\mathrm{tr}})} = \frac{1}{d^2-1}\sum_{\tau,\tau'\in \bsq} \laa \tau\tn{2}|\mc{E}\tn{2}|{\tau'}\tn{2}\raa,
\end{equation}
where we have used the definition of $\mc{P}_{\mathrm{tr}}$ (\cref{lem:Clifford_irreps}). We can calculate
\begin{align}
\begin{split}
\frac{1}{d^2-1}\sum_{\tau,\tau'\in \bsq} \laa \tau\tn{2}|\mc{E}\tn{2}|{\tau'}\tn{2}\raa &=\frac{1}{d^2-1}\sum_{\tau,\tau'\in \bsq} \laa \tau|\mc{E}|{\tau'}\raa^2\\
&= \frac{1}{d^2-1}\sum_{\tau,\tau'\in \bsq} \laa \tau|\mc{E}|{\tau'}\raa\laa \tau'|\mc{E}\ct |\tau \raa\\
&= \frac{1}{d^2-1}\sum_{\tau,\tau'\in \bsq} \laa \tau|\mc{E}_u\mc{E}_u\ct|\tau\raa\\
&= u(\mc{E})
\end{split}
\end{align}
where we have used the definition of the unitarity. Now consider $i \in \mc{Z}_{\mathrm{d}}$. We have
\begin{align}
\begin{split}
\chi_{i}  &= \frac{\tr(\mc{P}_i\mc{E}\tn{2})}{\tr(\mc{P}_i)} \\
&= \frac{1}{\tr(\mc{P}_i)}\sum_{\tau \in \bsq}\laa \tau\tn{2} |\mc{P}_i \mc{E}\tn{2}|\tau\tn{2}\raa\\
& = \frac{1}{\tr(\mc{P}_i)}\sum_{\tau,\tau' \in \bsq}\laa \tau\tn{2} |\mc{P}_i|{\tau'}\tn{2}\raa\laa{\tau'}\tn{2}| \mc{E}\tn{2}|\tau\tn{2}\raa
\end{split}
\end{align}
Where we have used that the support of $\mc{P}_i$ lies in $V_{\mathrm{d}}$ (defined in \cref{lem:Clifford_irreps}).
Now we can use \cref{lem:projector_lemma} to upper bound this quantity. We have
\begin{align}
\begin{split}
\chi_i &\leq \frac{1}{\tr(\mc{P}_i)}\sum_{\tau,\tau' \in \bsq}\frac{\tr(\mc{P}_i)}{d^2-1)}\laa{\tau'}\tn{2}| \mc{E}\tn{2}|\tau\tn{2}\raa\\
&= \frac{1}{d^2-1}\sum_{\tau,\tau'\in \bsq} \laa{\tau'}| \mc{E}|\tau\raa \laa{\tau}| \mc{E}\ct|{\tau'}\raa\\
&= u
\end{split}
\end{align}
where we have again used the definition of the unitarity.\\ Next we consider the case of $i\in \mc{Z}_{[S]}$. We have
\begin{align}
\chi_i = \frac{\tr(\mc{P}_i\mc{E}\tn{2})}{\tr(\mc{P}_i)} = \frac{1}{4} \frac{1}{\tr(\mc{P}_i)} \sum_{\tau, \tau'\in \bsq}\sum_{\substack{\sigma\in \bf{C}_\tau\\\sigma'\in \bf{C}_{\tau'}}} \laa S_{\sigma, \sigma\cdot\tau}|\mc{P}_i|S_{\sigma',\sigma'\cdot\tau'}\raa \laa S_{\sigma',\sigma'\cdot\tau'}|\mc{E}\tn{2}| S_{\sigma,\sigma\cdot\tau}\raa
\end{align}
where we have used that the support of $\mc{P}_i$ lies in $V_{[S]}$ (defined in \cref{lem:Clifford_irreps}) and the factor of $1/4$ accounts for the fact that we are double counting the basis of $V_{[S]}$ since $S_{\sigma, \sigma\cdot \tau} = S_{\sigma\cdot \tau, \sigma}$ (we double count twice: once in the definition of the trace and once in the resolution of the identity on $V_{[S]}$). From \cref{lem:projector_lemma} we can lose one of the sums and get
\begin{align}
\label{eq:chi_sym_unitarity}
\begin{split}
\chi_i  &= \frac{1}{4} \frac{1}{\tr(\mc{P}_i)} \sum_{\tau, \tau'\in \bsq}\sum_{\substack{\sigma\in \bf{C}_\tau\\\sigma'\in \bf{C}_{\tau'}}} \laa S_{\sigma, \sigma\cdot\tau}|\mc{P}_i|S_{\sigma',\sigma'\cdot\tau'}\raa \delta_{\tau,\tau'}\laa S_{\sigma',\sigma'\cdot\tau'}|\mc{E}\tn{2}| S_{\sigma,\sigma\cdot\tau}\raa\\
&= \frac{1}{4} \frac{1}{\tr(\mc{P}_i)} \sum_{\tau\in \bsq}\sum_{\sigma,\sigma'\in \bf{C}_\tau} \laa S_{\sigma, \sigma\cdot\tau}|\mc{P}_i|S_{\sigma',\sigma'\cdot\tau}\raa \laa S_{\sigma',\sigma'\cdot\tau}|\mc{E}\tn{2}| S_{\sigma,\sigma\cdot\tau}\raa.
\end{split}
\end{align}
We can further use \cref{lem:projector_lemma} to upper bound this quantity as
\begin{align}
\begin{split}
\chi_i &\leq \frac{1}{4} \frac{1}{\tr(\mc{P}_i)} \sum_{\tau\in \bsq}\sum_{\sigma,\sigma'\in \bf{C}_\tau} |\laa S_{\sigma, \sigma\cdot\tau}|\mc{P}_i|S_{\sigma',\sigma'\cdot\tau}\raa| |\laa S_{\sigma',\sigma'\cdot\tau}|\mc{E}\tn{2}| S_{\sigma,\sigma\cdot\tau}\raa|\\
&\leq \frac{1}{4} \frac{1}{\tr(\mc{P}_i)} \sum_{\tau\in \bsq}\sum_{\sigma,\sigma'\in \bf{C}_\tau} \frac{2\tr(\mc{P}_i)}{(d^2-1)(d/2-2)}|\laa S_{\sigma',\sigma'\cdot\tau}|\mc{E}\tn{2}| S_{\sigma,\sigma\cdot\tau}\raa|\\
&= \frac{1}{2}\frac{1}{(d^2-1)(d^2/2-2)}\sum_{\tau\in \bsq}\sum_{\sigma,\sigma'\in \bf{C}_\tau}|\laa\sigma|\mc{E}|\sigma'\raa \laa \sigma\cdot\tau|\mc{E}|\sigma'\cdot\tau\raa +  \laa\sigma\cdot\tau|\mc{E}|\sigma'\raa \laa \sigma|\mc{E}|\sigma'\cdot\tau\raa|
\end{split}
\end{align}
where we have also used the triangle inequality for the absolute value. Using the triangle inequality again together with the fact that $2|ab|\leq a^2+b^2$ for all $a, b\in \mathbb{R}$ we can write
\begin{align}
\begin{split}
\chi_i & \leq \frac{1}{2}\frac{1}{(d^2-1)(d^2/2-2)}\sum_{\tau\in \bsq}\sum_{\sigma,\sigma'\in \bf{C}_\tau}|\mc{E}_{\sigma,\sigma'} \mc{E}_{\sigma\cdot\tau,\sigma'\cdot \tau}| + |\mc{E}_{\sigma\cdot\tau,\sigma'} \mc{E}_{\sigma,\sigma'\cdot \tau}|\\
&\leq \frac{1}{4}\frac{1}{(d^2-1)(d^2/2-2)}\sum_{\tau\in \bsq}\sum_{\sigma,\sigma'\in \bf{C}_\tau} \mc{E}_{\sigma,\sigma'}^2 + \mc{E}_{\sigma\cdot\tau,\sigma'\cdot \tau}^2 + \mc{E}_{\sigma\cdot \tau,\sigma'}^2 + \mc{E}_{\sigma,\sigma'\cdot \tau}^2
\end{split}
\end{align}
Now since $\sigma\in \bf{C}_\tau \iff \sigma\cdot\tau \in \bf{C}_\tau$ we can roll the four sums in the above expression into one, that is
\begin{align}
\begin{split}
\chi_i &\leq \frac{1}{(d^2-1)(d^2/2-2)}\sum_{\tau\in \bsq}\sum_{\sigma,\sigma'\in \bf{C}_\tau}\mc{E}_{\sigma,\sigma'}^2\\
		& =\sum_{\sigma,\sigma'\in \bsq} \sum_{\tau \in \bf{C}_\sigma\cap \bf{C}_{\sigma'}}  \mc{E}_{\sigma,\sigma'}^2\\
		&\leq \frac{1}{(d^2-1)} \sum_{\sigma,\sigma'\in \bsq}\mc{E}_{\sigma,\sigma'}^2\\
		&=  u
\end{split}
\end{align}
where we used the fact that $\sigma \in \bf{C}_\tau \iff \tau \in {\bf C}_{\sigma}$, the fact that $|{\bf{C}_{\sigma}}\cap {\bf{C}_{\sigma'}}|\leq |{\bf{C}_{\sigma}}|=d^2/2-2$ and the definition of the unitarity. This means we have $\chi_i\leq u$ for all $i \in \mc{Z}_{[S]}$. The argument for $i\in \mc{Z}_{\{S\}}$ is conceptually the same as that for $i \in \mc{Z}_{[S]}$ so we will not write it down.
\end{proof}

\Cref{lem:chi_f_bound} proves that the eigenvalues of the twirled superoperator $\mc{T}_{\phi_{\mathrm{TS}}}(\mc{E}\tn{2})$ are close to the depolarizing parameter $f$. This fact is key in our analysis of the variance of RB in the presence of SPAM.
\begin{lemma}\label{lem:chi_f_bound}
Let $\mc{E}:\M \to \M$ be a quantum channel with infidelity $r$ and depolarizing parameter $f = 1- \frac{dr}{d-1}$ and consider the twirled operator $\mc{T}_{\phi_{\mathrm{TS}}}(\mc{E}\tn{2})$ with respect to the traceless-symmetric representation. This operator can then be written as (\cref{lem:Clifford_irreps,lem:Schur})
\begin{equation}
\mc{T}_{\phi_{\mathrm{TS}}}(\mc{E}\tn{2}) = \sum_{i\in \mc{Z}} \chi_i\mc{P}_i
\end{equation}
with $\mc{Z} = \{\mathrm{tr},1,2,[1],[2],[3],\{1\},\{2\}\}$, $\mc{P}_i$ the projector onto the spaces $V_i\subset \mc{M}_d\tn{2}$ and
\begin{equation}
\chi_i := \frac{\tr(\mc{E}\mc{P}_i)}{\tr(\mc{P}_i)},
\end{equation}
where the trace is taken over superoperators.
We now have for all $i\in \mc{Z}_{\mathrm{d}}$
\begin{equation}
|\chi_i-f^2| \leq \frac{2dr}{d-1},
\end{equation}
and for all $i \in \mc{Z}_{[S]}\cup\mc{Z}_{\{S\}}$
\begin{equation}
|\chi_i-f^2| \leq \frac{2dr}{d-1} .
\end{equation}
subject to the constraint $r\leq \frac{1}{3}$
\end{lemma}
\begin{proof}
From \cref{lem:unitarity_upper_bound} we have that $\chi_i\leq u$ for all $i \in \mc{Z}$. And since $u\leq 1$ for all quantum channels~\cite{Wallman2015} we certainly have that
\begin{equation}
\chi_i-f^2\leq 1- \left(1-\frac{dr}{d-1}\right)^2\leq \frac{2dr}{d-1}.
\end{equation}
Hence we are only interested in upper bounding $f^2-\chi_i$, and thus lower bounding $\chi_i$ for all $i \in \mc{Z}$. First consider $i\in \mc{Z}_{\mathrm{d}}$. We proceed in much the same way as \cref{lem:unitarity_upper_bound}. We have
\begin{align}
\label{eq:diag_sec_chi_lower_bound}
\begin{split}
\chi_i 	&= \frac{\tr(\mc{P}_i\mc{E}\tn{2})}{\tr(\mc{P}_i)}\\
	 	&= \frac{1}{\tr(\mc{P}_i)} \sum_{\tau,\tau'\in\bsq} \laa \tau\tn{2}|\mc{P}_i|{\tau'}\tn{2}\raa \laa {\tau'}\tn{2}|\mc{E}|\tau\tn{2}\raa\\
	 	&= \frac{1}{\tr(\mc{P}_i)} \sum_{\tau\in\bsq}\laa \tau\tn{2}|\mc{P}_i|{\tau}\tn{2}\raa\mc{E}_{\tau,\tau}^2 + \frac{1}{\tr(\mc{P}_i)} \sum_{\substack{\tau,\tau'\in\bsq\\\tau\neq\tau'}}\laa \tau\tn{2}|\mc{P}_i|{\tau'}\tn{2}\raa\mc{E}_{\tau',\tau}^2
\end{split}
\end{align}
We begin by considering the first term in \cref{eq:diag_sec_chi_lower_bound}. Using \cref{lem:projector_lemma} we can say
\begin{align}
\frac{1}{\tr(\mc{P}_i)} \sum_{\tau\in\bsq}\laa \tau\tn{2}|\mc{P}_i|{\tau}\tn{2}\raa\mc{E}_{\tau,\tau}^2 = \frac{\tr(\mc{P}_i)}{(d^2-1)\tr(\mc{P}_i)} \sum_{\tau\in \bsq} \mc{E}_{\tau,\tau}^2 \geq f^2
\end{align}
where we have also used the lower bound from \cref{lem:diagonal_squared_channel}. Now let us consider the second term in \cref{eq:diag_sec_chi_lower_bound}. We have
\begin{align}
\begin{split}
\frac{1}{\tr(\mc{P}_i)} \sum_{\substack{\tau,\tau'\in\bsq\\\tau\neq\tau'}}\laa \tau\tn{2}|\mc{P}_i|{\tau'}\tn{2}\raa\mc{E}_{\tau',\tau}^2 &\geq - \frac{1}{\tr(\mc{P}_i)} \sum_{\substack{\tau,\tau'\in\bsq\\\tau\neq\tau'}}|\laa \tau\tn{2}|\mc{P}_i|{\tau'}\tn{2}\raa|\mc{E}_{\tau',\tau}^2\\
&\geq- \frac{1\tr(\mc{P}_i)}{(d^2-1)\tr(\mc{P}_i)} \sum_{\substack{\tau,\tau'\in\bsq\\\tau\neq\tau'}}\mc{E}_{\tau',\tau}^2\\
&= -\frac{1}{d^2-1} \sum_{\tau,\tau'\in\bsq}\mc{E}_{\tau',\tau}^2 + \frac{1}{d^2-1}\sum_{\tau\in \bsq} \mc{E}_{\tau,\tau}^2\\
&\geq-u + f^2
\end{split}
\end{align}
where we have again used \cref{lem:projector_lemma}, the lower bound from \cref{lem:diagonal_squared_channel} and the definition of unitarity. We can now see that for $i \in \mc{Z}_{\mathrm{d}}$ we have
\begin{align}
f^2- \chi_i \leq f^2 - 2f^2 + u = u- f^2 \leq 1 - \left(1-\frac{dr}{d-1}\right)^2\leq \frac{2dr}{d-1}.
\end{align}

Now consider $i \in \mc{Z}_{[S]}$ (note that we are implicitly taking $d\geq 4$ for this part of the proof, this is justified as the set $\mc{Z}_{[S]}$ is empty for $q=1$). From \cref{lem:unitarity_upper_bound} and in particular \cref{eq:chi_sym_unitarity} we get
\begin{equation}
\chi_i = \frac{1}{4} \frac{1}{\tr(\mc{P}_i)} \sum_{\tau\in \bsq}\sum_{\sigma,\sigma'\in \bf{C}_\tau} \laa S_{\sigma, \sigma\cdot\tau}|\mc{P}_i|S_{\sigma',\sigma'\cdot\tau}\raa \laa S_{\sigma',\sigma'\cdot\tau}|\mc{E}\tn{2}| S_{\sigma,\sigma\cdot\tau}\raa.
\end{equation}
We can rewrite this a little bit as follows
\begin{align}
\chi_i &= \frac{1}{4} \frac{1}{\tr(\mc{P}_i)} \sum_{\tau\in \bsq}\sum_{\sigma,\sigma'\in \bf{C}_\tau} \laa S_{\sigma, \sigma\cdot\tau}|\mc{P}_i|S_{\sigma',\sigma'\cdot\tau}\raa(\mc{E}_{\sigma',\sigma}\mc{E}_{\sigma'\cdot\tau,\sigma\cdot \tau} +\mc{E}_{\sigma',\sigma\cdot\tau}\mc{E}_{\sigma'\cdot\tau,\sigma} )\\
\begin{split}
&= \frac{1}{4} \frac{1}{\tr(\mc{P}_i)} \sum_{\tau\in \bsq}\sum_{\sigma,\sigma'\in \bf{C}_\tau} \laa S_{\sigma, \sigma\cdot\tau}|\mc{P}_i|S_{\sigma',\sigma'\cdot\tau}\raa\mc{E}_{\sigma',\sigma}\mc{E}_{\sigma'\cdot\tau,\sigma\cdot \tau}\\&\hspace{10mm} +\frac{1}{4} \frac{1}{\tr(\mc{P}_i)} \sum_{\tau\in \bsq}\sum_{\sigma,\sigma'\in \bf{C}_\tau} \laa S_{\sigma, \sigma\cdot\tau}|\mc{P}_i|S_{\sigma',\sigma'\cdot\tau}\raa \mc{E}_{\sigma',\sigma\cdot\tau}\mc{E}_{\sigma'\cdot\tau,\sigma}
\end{split}\\
\begin{split}
&=\frac{1}{4} \frac{1}{\tr(\mc{P}_i)} \sum_{\tau\in \bsq}\sum_{\sigma,\sigma'\in \bf{C}_\tau} \laa S_{\sigma, \sigma\cdot\tau}|\mc{P}_i|S_{\sigma',\sigma'\cdot\tau}\raa\mc{E}_{\sigma',\sigma}\mc{E}_{\sigma'\cdot\tau,\sigma\cdot \tau}\\&\hspace{10mm} +\frac{1}{4} \frac{1}{\tr(\mc{P}_i)} \sum_{\tau\in \bsq}\sum_{\sigma,\sigma'\in \bf{C}_\tau} \laa S_{\sigma, \sigma\cdot\tau}|\mc{P}_i|S_{\sigma'\cdot \tau,(\sigma'\cdot\tau)\cdot \tau}\raa \mc{E}_{\sigma'\cdot \tau,\sigma\cdot\tau}\mc{E}_{(\sigma'\cdot\tau)\cdot \tau,\sigma}
\end{split}\\
&=\frac{1}{2} \frac{1}{\tr(\mc{P}_i)} \sum_{\tau\in \bsq}\sum_{\sigma,\sigma'\in \bf{C}_\tau} \laa S_{\sigma, \sigma\cdot\tau}|\mc{P}_i|S_{\sigma',\sigma'\cdot\tau}\raa\mc{E}_{\sigma',\sigma}\mc{E}_{\sigma'\cdot\tau,\sigma\cdot \tau}
\end{align}
where we used that $\sigma' \in {\bf{C}_\tau} \iff\sigma'\cdot \tau \in {\bf{C}_\tau} $, that $(\sigma'\cdot \tau)\cdot \tau = \sigma'$ and that $S_{\sigma',\sigma'\cdot\tau} = S_{\sigma'\cdot \tau,\sigma'}$.
We can again separate off the `diagonal' terms to get
\begin{subequations}
\begin{align}
\chi_i &= \frac{1}{2} \frac{1}{\tr(\mc{P}_i)} \sum_{\tau\in \bsq}\sum_{\sigma\in \bf{C}_\tau} \laa S_{\sigma, \sigma\cdot\tau}|\mc{P}_i|S_{\sigma,\sigma\cdot\tau}\raa\mc{E}_{\sigma,\sigma}\mc{E}_{\sigma\cdot\tau,\sigma\cdot \tau}\label{eq:chi_sym_diag_term}\\
&\hspace{10mm}+\frac{1}{2} \frac{1}{\tr(\mc{P}_i)} \sum_{\tau\in \bsq}\sum_{\substack{\sigma,\sigma'\in \bf{C}_\tau\\\sigma\neq\sigma'}} \laa S_{\sigma, \sigma\cdot\tau}|\mc{P}_i|S_{\sigma',\sigma'\cdot\tau}\raa\mc{E}_{\sigma',\sigma}\mc{E}_{\sigma'\cdot\tau,\sigma\cdot \tau}\label{eq:chi_sym_off_diag_term}.
\end{align}
\end{subequations}
We will analyze the terms \cref{eq:chi_sym_diag_term} and \cref{eq:chi_sym_off_diag_term} separately. We begin with \cref{eq:chi_sym_diag_term}. We can use \cref{lem:projector_lemma} to get
\begin{equation}
\text{\cref{eq:chi_sym_diag_term}} = \frac{1}{(d^2-1)\left(\frac{d^2}{2}-2\right)} \sum_{\tau\in \bsq}\sum_{\sigma\in \bf{C}_\tau} \mc{E}_{\sigma,\sigma}\mc{E}_{\sigma\cdot\tau,\sigma\cdot \tau}.
\end{equation}
Now we use the generic statement $2ab = a^2 + b^2 - (a-b)^2$ for all $a,b\in \mathbb{R}$ to write
\begin{align}
\begin{split}
\text{\cref{eq:chi_sym_diag_term}}  &= \frac{1}{2}\frac{1}{(d^2-1)\left(\frac{d^2}{2}-1\right)} \sum_{\tau\in \bsq}\sum_{\sigma\in \bf{C}_\tau} \mc{E}_{\sigma,\sigma}^2 + \mc{E}_{\sigma\cdot\tau,\sigma\cdot \tau}^2\\&\hspace{15mm} - \frac{1}{2}\frac{1}{(d^2-1)\left(\frac{d^2}{2}-2\right)} \sum_{\tau\in \bsq}\sum_{\sigma\in \bf{C}_\tau}(\mc{E}_{\sigma,\sigma} -\mc{E}_{\sigma\cdot\tau,\sigma\cdot \tau})^2
\end{split}\\
\begin{split}
&= \frac{1}{(d^2\!-\!1)\left(\frac{d^2}{2}-2\right)} \sum_{\tau\in \bsq}\sum_{\sigma\in \bf{C}_\tau} \mc{E}_{\sigma,\sigma}^2\\&\hspace{15mm} - \frac{1}{2}\frac{1}{(d^2\!-\!1)\left(\frac{d^2}{2}\!-\!2\right)} \sum_{\tau\in \bsq}\sum_{\sigma\in \bf{C}_\tau}(\mc{E}_{\sigma,\sigma} -\mc{E}_{\sigma\cdot\tau,\sigma\cdot \tau})^2
\end{split}\\
\begin{split}
&= \frac{1}{(d^2-1)\left(\frac{d^2}{2}-2\right)} \sum_{\sigma\in \bsq}\sum_{\tau\in \bf{C}_\sigma} \mc{E}_{\sigma,\sigma}^2 \\&\hspace{15mm}- \frac{1}{2}\frac{1}{(d^2-1)\left(\frac{d^2}{2}-2\right)} \sum_{\tau\in \bsq}\sum_{\sigma\in \bf{C}_\tau}(\mc{E}_{\sigma,\sigma} -\mc{E}_{\sigma\cdot\tau,\sigma\cdot \tau})^2
\end{split}\\
&= \frac{1}{d^2-1} \sum_{\sigma\in \bsq} \mc{E}_{\sigma,\sigma}^2 - \frac{1}{2}\frac{1}{(d^2-1)\left(\frac{d^2}{2}-2\right)} \sum_{\tau\in \bsq}\sum_{\sigma\in \bf{C}_\tau}(\mc{E}_{\sigma,\sigma} -\mc{E}_{\sigma\cdot\tau,\sigma\cdot \tau})^2\\
&\geq f^2 - \frac{1}{2}\frac{1}{(d^2-2)\left(\frac{d^2}{2}-1\right)} \sum_{\tau\in \bsq}\sum_{\sigma\in \bf{C}_\tau}(\mc{E}_{\sigma,\sigma} -\mc{E}_{\sigma\cdot\tau,\sigma\cdot \tau})^2\label{eq:chi_sym_diag_side_term}
\end{align}
where we again used that $\sigma\in {\bf{C}_\tau}\iff \sigma\cdot \tau \in {\bf{C}_\tau}$ and that $\sigma \in{\bf{C}_\tau}\iff  \tau\in {\bf{C}_\sigma}$ and also the lower bound from \cref{lem:diagonal_squared_channel}. It remains to bound the second term in \cref{eq:chi_sym_diag_side_term}. To do this we will maximize the quantity $(\mc{W}_{\nu,\nu} -\mc{W}_{\nu\cdot\mu,\nu \cdot\mu})^2$ for $\mu\in \bsq$ and $\nu\in \bf{C}_\mu$ subject to the constraint that $\mc{W}$ is a CPTP map with depolarizing parameter $f$. That is, we will try to solve the maximization problem
 \begin{align}
\begin{aligned}
& \text{max}
& & (\mc{W}_{\nu,\nu} -\mc{W}_{\mu,\mu})^2\\
& \text{subject to}
& & \sum_{\tau \in \boldsymbol{\sigma}_q} \mc{W}_{\tau\tau} =(d^2-1)f\\
& & & \mc{W} \text{  a CPTP map.}
\end{aligned}
\end{align}
As in \cref{lem:diagonal_squared_channel} we can restrict ourselves to $\mc{W}$ being a Pauli channel (since the optimization function is a function of only the diagonal elements of $\mc{W}$). That is we can consider $\mc{W}(X) = \sum_{G\in \md{P}} p_G GXG\ct$ where $\{p_G\}_G$ is a probability distribution over the Pauli group. We can write the optimization objective as
\begin{align}
\begin{split}
(\mc{W}_{\nu,\nu} -\mc{W}_{\nu\cdot\mu,\nu\cdot\mu})^2 &= \left[ \sum_{G\in \md{P}} p_G\inp{\nu}{G\nu G\ct} - \inp{\nu\cdot\mu}{G\nu \cdot\mu G\ct}\right]^2 \\
&= \left[ \sum_{G\in \md{P}} p_G\inp{\nu}{G\nu G\ct} - \inp{\nu\cdot\mu}{(G\nu G\ct) \cdot(G \mu G\ct)}\right]^2 \\
& = \left[ \sum_{G\in \md{P}} p_G \text{sgn}(\nu,G)\big(1 -\text{sgn}(\mu,G)\big)\right]^2
\end{split}
\end{align}
where the $\text{sgn}(\nu,G)$ (as defined in \cref{eq:sign_sign_function_pauli}) encodes the commutation relations of the elements of the Pauli group. Note that the above quantity does not depended on $p_\id$ (the weight associated with the Pauli identity) since $\text{sgn}(\mu,\id)=1$ for all $\mu \in \bsq$. Hence we can solve the optimization problem
 \begin{align}\label{eq:chi_sym_opt_problem}
\begin{aligned}
& \text{max}
& & \left[ \sum_{G\in \md{P}/\{\id\}} p_G \text{sgn}(\nu,G)\big(1 -\text{sgn}(\mu,G)\big)\right]^2\\
& \text{subject to}
& & \sum_{G\in \md{P}/\{\id\}} p_G = 1 - \frac{d^2-1}{d^2}f^2 - \frac{1}{d^2}\\
& & & p_G\geq 0 \;\;\;\;\forall G\in \md{P}.
\end{aligned}
\end{align}
This problem has an easily spotted maximum in that we want to put all probability weight on a single $G\in C_\nu\cap N_\mu $ and set all other $p_G$ to zero (subject to the constraint that the overall channel must have depolarizing parameter $f$, which is encoded in the first constraint of \cref{eq:chi_sym_opt_problem} ). Hence we have
\begin{equation}
\left[ \sum_{G\in \md{P}} p_G \text{sgn}(\nu,G)\big(1 -\text{sgn}(\mu,G)\big)\right]^2 \leq \left[\frac{d^2-1}{d^2}(1-f^2)\right]^2.
\end{equation}
We can feed this back into \cref{eq:chi_sym_diag_side_term} to obtain
\begin{align}
\begin{split}
\text{\cref{eq:chi_sym_diag_term}}&\geq  f^2 - \frac{1}{2}\frac{1}{(d^2-1)\left(\frac{d^2}{2}-2\right)} \sum_{\tau\in \bsq}\sum_{\sigma\in \bf{C}_\tau}\left[\frac{d^2-1}{d^2}(1-f^2)\right]^2\\
&= f^2 - \frac{1}{2}\left[\frac{d^2-1}{d^2}(1-f^2)\right]^2.
\end{split}
\end{align}
This is a suitable lower bound on \cref{eq:chi_sym_diag_term}. Next we consider \cref{eq:chi_sym_off_diag_term}. We have
\begin{align}
\begin{split}
\text{\cref{eq:chi_sym_off_diag_term}}&\geq-\frac{1}{(d^2-1)\left(\frac{d^2}{2}-2\right)} \sum_{\tau\in \bsq}\sum_{\substack{\sigma,\sigma'\in \bf{C}_\tau\\\sigma\neq\sigma'}} |\mc{E}_{\sigma,\sigma'}\mc{E}_{\sigma\cdot\tau,\sigma'\cdot\tau}|\\
&\geq -\frac{1}{(d^2-1)\left(\frac{d^2}{2}-2\right)} \sum_{\tau\in \bsq}\sum_{\substack{\sigma,\sigma'\in \bf{C}_\tau\\\sigma\neq\sigma'}} \frac{1}{2}(\mc{E}_{\sigma,\sigma'}^2 + \mc{E}_{\sigma\cdot\tau,\sigma'\cdot\tau}^2)\\
&=-\frac{1}{(d^2-1)\left(\frac{d^2}{2}-2\right)}\sum_{\substack{\sigma,\sigma'\in \bsq\\\sigma\neq\sigma'}}\sum_{\tau\in \bf{C}_\sigma\cap\bf{C}_{\sigma'}} \mc{E}_{\sigma,\sigma'}^2\\
&=-\frac{\frac{d^4}{4} -3}{(d^2-1)\left(\frac{d^2}{2}-2\right)}\left[\sum_{\sigma,\sigma'\in \bsq}\mc{E}_{\sigma,\sigma'}^2 - \sum_{\sigma\in \bsq}\mc{E}_{\sigma,\sigma'}^2\right]\\
&\geq- \frac{\frac{d^4}{4} -3}{\frac{d^2}{2}-2}(u-f^2)
\end{split}
\end{align}
where we used an array of steps that have been used before: the triangle inequality and \cref{lem:projector_lemma} for the first inequality, the fact that $2|ab|\leq a^2 + b^2$ for all $a,b \in \mathbb{R}$ for the second inequality, the fact that $\sigma \in \bf{C}_\tau \iff \tau \in \bf{C}_\sigma$ for the third equality, the fact that $|{\bf{C}_\sigma\cap \bf{C}_{\sigma}}| = d^2/4-3$ for  $\sigma\neq \sigma'$~\cite{Clifford2016}  for the fourth equality and \cref{lem:diagonal_squared_channel} and the definition of unitarity for the last equality. This is a good lower bound on \cref{eq:chi_sym_off_diag_term}. We can now combine the lower bounds on \cref{eq:chi_sym_diag_term} and \cref{eq:chi_sym_off_diag_term} to get
\begin{equation}
\chi_i \geq f^2 - \frac{1}{2}\left[\frac{d^2-1}{d^2}(1-f^2)\right]^2  - \frac{\frac{d^4}{4} -3}{\left(\frac{d^2}{2}-2\right)}(u-f^2)
\end{equation}
for $i\in \mc{Z}_{[S]}$.
This gives a final bound (using $u\leq 1$)
\begin{equation}
f^2-\chi_i \leq f^2 - f^2 +\frac{1}{2}\left[\frac{d^2-1}{d^2}(1-f^2)\right]^2  + \frac{\frac{d^4}{4} -3}{\frac{d^2}{2}-2}(1-f^2)
\end{equation}
which we can rewrite to yield
\begin{equation}
f^2-\chi_i \leq \frac{2dr}{d-1}\left(\frac{\frac{d^4}{4} -3}{\frac{d^2}{2}-2} \left(1- \frac{1}{2}\frac{dr}{d-1}\right) + \frac{1}{2}\frac{(d^2-1)^2}{d^4} \frac{2dr}{d-1}\left(1- \frac{1}{2}\frac{dr}{d-1}\right)^2 \right)
\end{equation}
Setting $\left(1- \frac{1}{2}\frac{dr}{d-1}\right)\leq 1$ and working out we get
\begin{equation}
	f^2-\chi \leq  \frac{2d}{d-1}r
\end{equation}
for
\begin{equation}
r\leq \left(1-\frac{\frac{d^4}{4} -3}{\frac{d^2}{2}-2}\right) \frac{d^3(d-1)}{(d^2-1)^2}.
\end{equation}
This completes the proof for $i \in \mc{Z}_{[S]}$. The proof for $i \in \mc{Z}_{\{S\}}$ is conceptually the same as that of  $i \in \mc{Z}_{[S]}$ and yields the same bound so we will not write it down here. The only notable difference is the difference in size for the sets ${\bf{N}}_\tau$ and ${\bf{N}}_\tau \cap {\bf{N}}_{\tau'}  $ for $\tau,\tau'\in \bsq$ which gives a different area of validity for the bound, namely
\begin{equation}
r\leq\frac{1}{3}\leq \left(1-\frac{\frac{d^4}{4}}{\frac{d^2}{2}}\right) \frac{d^3(d-1)}{(d^2-1)^2}.
\end{equation}
 Choosing $r\leq 1/3$ satisfies both constraints for all $d$ and thus completes the proof.
\end{proof}

\subsection{Telescoping series}

\Cref{lem:telescoping series,cor:second_order_series} provide us with a powerful tool to break up the analysis of the variance of RB into manageable pieces. 
\begin{lemma}\label{lem:telescoping series}
	For two arbitrary ordered lists of $m$ elements $\{a_1,\ldots,a_m\}$ and $\{b_1,\ldots,b_m\}$ of an algebra with associative and distributed addition and multiplication we have,
	\begin{align}
	a_{m:1} - b_{m:1} = \sum_{j=1}^m a_{m:j+1}(a_j-b_j) b_{j-1:1}.
	\end{align}
	where $a_{j:k}$ with $j\geq k$ is defined with respect to the list $\{a_1,\ldots,a_m\}$ as
	\begin{equation}
	a_{j:k} = a_j a_{j+1}\cdots a_{k-1}a_k.
	\end{equation}
\end{lemma}
\begin{proof}
	We will prove this by induction. For $m=1$ the statement is trivial. For $m+1$,
	we have
	\begin{align*}
	a_{m+1:1} - b_{m+1:1} &= a_{m+1} a_{m:1} - a_{m+1} b_{m:1} + a_{m+1} b_{m:1} - b_{m+1} b_{m:1} \notag\\
	&= a_{m+1}(a_{m:1}-b_{m:1}) + (a_{m+1} - b_{m+1})b_{m:1} \notag\\
	&= \sum_{j=1}^{m+1} a_{m:j+1}(a_j-b_j) b_{j-1:1}
	\end{align*}
	by induction hypothesis. This proves the lemma.
\end{proof}

\begin{cor}\label{cor:second_order_series}
For $a,b,c \in \mathbb{C}$ with $c\geq a$, we have
\begin{align*}
a^m - b^m  &= mb^{m-1}(a-b) + (a-b)^2 a^{m-2}\frac{(m-1)(b/a)^{m} - m(b/a)^{m-1} + 1}{(1-(b/a))^2} \\
&\leq mb^{m-1}(a-b) + (a-b)^2
\frac{(m-1)b^{m} - m c b^{m-1} + c^m}{(c-b)^2}
\end{align*}
\end{cor}
\begin{proof}
Note first that the statement is trivial if $a=b$. Therefore assume $a\neq b$.
We begin by applying \cref{lem:telescoping series} to $a^m-b^m$. This gives

\begin{align}
a^m - b^m &= \sum_{j=1}^m a^{m-j}(a-b) b^{j-1}.
\end{align}
We now perform the following manipulation
\begin{align}
\begin{split}
a^m - b^m &=\sum_{j=1}^m a^{m-j}(a-b) b^{j-1}\\
& =\sum_{j=1}^m (a^{m-j}-b^{m-j} + b^{m-j})(a-b) b^{j-1} \\
& = (a-b)\sum_{j=1}^m b^{m-j+j-1} + \sum_{j=1}^m (a^{m-j}-b^{m-j})(a-b) b^{j-1} \\
&= mb^{m-1}(a-b) + \sum_{j=1}^m (a^{m-j}-b^{m-j})(a-b) b^{j-1}.
\end{split}
\end{align}
Note that be have used the fact that $a,b\in \mathbb{C}$ are commutative. Now we can apply \cref{lem:telescoping series} again to the factors $(a^{m-j}-b^{m-j})$ in the second term in the above to obtain
\begin{align}
\begin{split}
a^m - b^m &= mb^{m-1}(a-b) +\sum_{j=1}^m \sum_{t=1}^{m-j}a^{m-j-t}(a-b) b^{j-t-1}(a-b)b^{j-1}\\
&= mb^{m-1}(a-b) + (a-b)^2 \sum_{j=1}^m \sum_{t=1}^{m-j}a^{m-(j+t)}b^{j+t-2}.
\end{split}
\end{align}
Performing the substitution $s= j+t$ and working out we obtain
\begin{align}
\begin{split}
a^m + b^m  &= mb^{m-1}(a-b) + (a-b)^2 \sum_{j=1}^m \sum_{t=1}^{m-j}a^{m-(j+t)}b^{j+t-2}\\
&= mb^{m-1}(a-b) + (a-b)^2 \sum_{j=1}^m \sum_{s=j+1}^m a^{m-s}b^{s-2}\\
&= mb^{m-1}(a-b) + (a-b)^2 \sum_{s=2}^m \sum_{j=1}^{s-1} a^{m-s}b^{s-2}\\
&= mb^{m-1}(a-b) + (a-b)^2 \sum_{s=2}^m (s-1) a^{m-s}b^{s-2}
\end{split}
\end{align}
Now we can factor out $a^{m-1}$ from the second term to obtain
\begin{align}
a^m + b^m  &= mb^{m-1}(a-b) + (a-b)^2 a^{m-2} \sum_{s=2}^m (s-1) (b/a)^{s-2}.
\end{align}
We can further rewrite this using the standard series identity
\begin{equation}
\sum_{k=1}^{m}(k-2)x^{k-2} = \frac{(m-1)x^{m} - mx^{m-1} + 1}{(1-x)^2}.
\end{equation}
The upper bound follows by upper bounding each term in the sum.
\end{proof}

\end{document}